\documentclass[acmsmall, screen]{acmart}

\usepackage{adjustbox}
\usepackage{amsmath}
\usepackage{amsthm}
\usepackage{amsfonts}
\usepackage{mathrsfs}
\usepackage{mathtools}
\usepackage{bbold}
\usepackage{listings}
\usepackage{xcolor}
\usepackage{tabularx}
\usepackage{nameref}
\usepackage{graphicx}
\usepackage{xspace}
\usepackage{microtype}
\usepackage{url}
\usepackage{booktabs}
\usepackage[inline]{enumitem}
\usepackage{subcaption}
\usepackage{marvosym}
\usepackage{mathpartir}
\usepackage{stmaryrd}
\usepackage{accents}
\usepackage{bbm}
\usepackage[shortcuts]{extdash}
\usepackage{hyperref}
\usepackage{empheq}
\usepackage{skull}
\usepackage[T2A, T1]{fontenc}
\usepackage[utf8]{inputenc}
\usepackage[russian, english]{babel}
\usepackage{array}

\usepackage{tikz}
\usetikzlibrary{shapes, arrows, positioning, decorations.pathreplacing}

\usepackage{etoolbox}
\newtoggle{isTechReport}
\newcommand{\iftechreport}[2]{\iftoggle{isTechReport}{#1}{#2}}
\toggletrue{isTechReport}

\setcopyright{rightsretained}
\acmJournal{PACMPL}
\acmYear{2022}
\copyrightyear{2022}
\acmVolume{6}
\acmNumber{POPL}
\acmArticle{27}
\acmMonth{1}
\acmPrice{}
\acmDOI{10.1145/3498688}



\let\svthefootnote\thefootnote
\newcommand\blankfootnote[1]{%
  \let\thefootnote\relax\footnotetext{#1}%
  \let\thefootnote\svthefootnote%
}
\let\svfootnote\footnote
\renewcommand\footnote[2][?]{%
  \if\relax#1\relax%
  \blankfootnote{#2}%
  \else%
  \if?#1\svfootnote{#2}\else\svfootnote[#1]{#2}\fi%
  \fi
}

\newtheorem{theorem}{Theorem}
\newtheorem{lemma}{Lemma}
\newtheorem{definition}{Definition}
\newtheorem{proposition}{Proposition}

\newenvironment{subproof}[1][\proofname]{%
  \begin{proof}[#1]%
  }{%
  \end{proof}%
}

\newcommand{\pfcase}[1]{\item[Case] #1}

\newcommand\StrongNI{Disjoint Noninterference}

\newcommand\strongni{disjoint noninterference}

\newcommand\sectionref[1]{\hyperref[#1]{Section~\ref*{#1}}}
\newcommand\secref[1]{\hyperref[#1]{\S\ref*{#1}}}
\newcommand\figref[1]{\hyperref[#1]{Figure~\ref*{#1}}}
\newcommand\lemref[1]{\hyperref[#1]{Lemma~\ref*{#1}}}
\newcommand\thmref[1]{\hyperref[#1]{Theorem~\ref*{#1}}}
\newcommand\appref[1]{\hyperref[#1]{Appendix~\ref*{#1}}}
\newcommand\coderef[1]{\hyperref[#1]{Line~\ref*{#1}}}
\newcommand\defref[1]{\hyperref[#1]{Definition~\ref*{#1}}}
\newcommand\tabref[1]{\hyperref[#1]{Table~\ref*{#1}}}
\newcommand\listref[1]{\hyperref[#1]{Listing~\ref*{#1}}}

\newcommand\Red[1]{{\color{red} #1}}

\newcommand\ignore[1]{}

\newcommand\tocite[1]{\Red{[??]}}
\newcommand\para[1]{\smallskip\noindent\textbf{{#1.}\xspace}}

\newcommand\langname{SFIasm\xspace}
\newcommand\olangname{oSFIasm\xspace}
\newcommand\verifname{VeriZero\xspace}

\def\C++{
  C\kern-.1667em\raise.30ex\hbox{\smaller{++}}%
  \spacefactor1000
}

\newenvironment{CompactItemize}%
{\begin{list}{$\blacktriangleright$}%
    {\leftmargin=\parindent \itemsep=2pt \topsep=2pt
      \parsep=0pt \partopsep=0pt}}%
  {\end{list}}

\newenvironment{CompactEnumerate}{
  \begin{enumerate}[leftmargin=*]
  }{\end{enumerate}}

\makeatletter
\newcommand\ntuple[1]{%
  ({#1}\@tuplehelpcheck
}
\newcommand\@tuplehelpcheck{%
  \@ifnextchar\bgroup{\@tuplehelpnext}{)}
}
\newcommand{\@tuplehelpnext}[1]{%
  ,\, {#1} \@ifnextchar\bgroup{\@tuplehelpnext}{)}
}
\makeatother

\newcommand\nats{\ensuremath{\mathbb{N}}}
\newcommand\prop{\ensuremath{\mathbb{P}}}
\newcommand\powerset[1]{\ensuremath{\wp({#1})}}

\newcommand\oname[1]{\operatorname{#1}}
\newcommand\concat{\mathbin{{+}\mspace{-8mu}{+}}}
\newcommand\lesstrusted{\sqsubseteq}
\newcommand\nlesstrusted{\not\sqsubseteq}

\newcommand\dom[1]{\oname{dom}({#1})}
\newcommand\cod[1]{\oname{cod}({#1})}
\newcommand\assign{\mathbin{\vcentcolon=}}

\newcommand\relDescription[1]{#1}
\newcommand\judgmentHead[2]{\relDescription{#1}\hfill\fbox{\ensuremath{#2}}}

\newcommand\bnfdef{\mathrel{\vcentcolon\vcentcolon=}}
\newcommand\bnfalt{\mid}
\newcommand\bnftypes{\mathrel{\vcentcolon}}

\newcommand\cinst[1]{\ensuremath{\mathtt{#1}}}

\newcommand\cpush[2]{\ensuremath{\cinst{push}_{#1}\ {#2}}}
\newcommand\cpop[2]{\ensuremath{{#1} \leftarrow \cinst{pop}_{#2}}}
\newcommand\cload[3]{\ensuremath{{#1} \leftarrow \cinst{load}_{#2}\ {#3}}}
\newcommand\cstore[3]{\ensuremath{\cinst{store}_{#1}\ {#2} \assign {#3}}}
\newcommand\ccall[2]{\cinst{call}_{{#1}}\ {#2}}
\newcommand\cgatecall[2]{\ensuremath{\cinst{gatecall}_{#1}\ {#2}}}
\newcommand\cgateret{\ensuremath{\cinst{gateret}}}
\newcommand\cret[1]{\ensuremath{\cinst{ret}_{{#1}}}}
\newcommand\cjmp[2]{\ensuremath{\cinst{jmp}_{{#1}}\ {#2}}}
\newcommand\cmov[2]{\ensuremath{{#1} \leftarrow \cinst{mov}\ {#2}}}
\newcommand\cmovlabel[2]{\ensuremath{{#1} \leftarrow \cinst{movlabel}_{#2}}}
\newcommand\cstorelabel[2]{\ensuremath{\cinst{storelabel}_{#1}\ {#2}}}

\newcommand\error{\ensuremath{\mathtt{error}}}

\newcommand\vcinit{\ensuremath{\mathtt{Initialized}}}
\newcommand\vcuninit{\ensuremath{\mathtt{Uninitialized}}}
\newcommand\vccallee[1]{\ensuremath{\mathtt{UninitializedCallee}({#1})}}

\newcommand\immval[2]{\ensuremath{\mathcal{V}_{#1}({#2})}}
\newcommand\oimmval[2]{\ensuremath{\oc{\mathcal{V}}_{#1}({#2})}}

\newcommand\trusted{\ensuremath{\mathtt{app}}}
\newcommand\untrusted{\ensuremath{\mathtt{lib}}}

\newcommand\val[2]{\ensuremath{\oc{\langle} {#1}, {#2} \oc{\rangle}}}
\newcommand\natval[1]{\ensuremath{\oc{\langle} {#1} \oc{\rangle}}}

\newcommand\SF{\ensuremath{\mathit{SF}}}

\newcommand\privs{\ensuremath{\mathit{Priv}}}
\newcommand\codes{\ensuremath{\mathit{Code}}}
\newcommand\vals{\ensuremath{\mathit{Val}}}
\newcommand\regs{\ensuremath{\mathit{Reg}}}
\newcommand\commands{\ensuremath{\mathit{Command}}}
\newcommand\memories{\ensuremath{\mathit{Memory}}}
\newcommand\regions{\ensuremath{\mathit{Region}}}
\newcommand\checks{\ensuremath{\mathit{Check}}}
\newcommand\regvals{\ensuremath{\mathit{RegVals}}}
\newcommand\frames{\ensuremath{\mathit{Frame}}}
\newcommand\programs{\ensuremath{\mathit{Program}}}

\newcommand\expressions{\ensuremath{\mathit{Expr}}}
\newcommand\functions{\ensuremath{\mathit{Function}}}

\newcommand\states{\ensuremath{\mathit{State}}}
\newcommand\ostates{\ensuremath{\mathit{oState}}}

\newcommand\currentop[2]{\ensuremath{{#1}\llparenthesis{#2}\rrparenthesis}}
\newcommand\currentcom[3]{\ensuremath{{#1}\llparenthesis{#2}\rrparenthesis_{#3}}}
\newcommand\pcinc[1]{\ensuremath{{#1}^{++}}}

\newcommand\step{\mathrel{\rightarrow}}
\newcommand\stepstar{\mathrel{\step^{\ast}}}

\newcommand\stepn[1]{\mathrel{\step^{#1}}}

\newcommand\ostep{\mathrel{\leadsto}}
\newcommand\ostepstar{\mathrel{\ostep^{\ast}}}

\newcommand\ostepn[1]{\mathrel{\ostep^{#1}}}

\newcommand\overstep[1]{\mathrel{\xrightarrow{#1}}}
\newcommand\overstepstar[1]{\mathrel{\xrightarrow{#1}\mathrel{\vphantom{\to}^{\ast}}}}

\newcommand\overstepn[2]{\mathrel{\xrightarrow{#1}\mathrel{\vphantom{\to}^{#2}}}}

\newcommand\osteprhon[2]{\mathrel{\accentset{#1}{\leadsto}^{#2}}}

\newcommand\stepp[1]{\overstep{#1}}
\newcommand\steppstar[1]{\overstepstar{#1}}

\newcommand\steppn[2]{\overstepn{#1}{#2}}

\newcommand\stepbox{\overstep{\!\square}}
\newcommand\stepboxstar{\overstepstar{\!\square}}

\newcommand\steplowstar{\steppstar{\untrusted}}

\newcommand\steplown[1]{\steppn{\untrusted}{#1}}

\newcommand\stephigh{\stepp{\trusted}}

\newcommand\stepwb{\overstep{\!\text{wb}}}

\newcommand\evalsto{\mathrel{\Downarrow}}

\makeatletter

\newcommand\makerulename[2]{%
  \expandafter\newcommand\csname def#1\endcsname{%
    \Hy@raisedlink{\hypertarget{rule:def:#1}{}}\hyperlink{rule:explanation:#1}{#2}%
  }%
  \expandafter\newcommand\csname explain#1\endcsname{%
    \Hy@raisedlink{\hypertarget{rule:explanation:#1}{}}\hyperlink{rule:def:#1}{\LabTirName{#2}}\xspace%
  }%
  \expandafter\newcommand\csname #1\endcsname{%
    \hyperlink{rule:def:#1}{\LabTirName{#2}}\xspace%
  }%
}

\makeatother

\makerulename{redFblock}{red\=/F\=/block}
\makerulename{redFjump}{red\=/F\=/jump}
\makerulename{redFindirect}{red\=/F\=/indirect}
\makerulename{redFapp}{red\=/F\=/app}
\makerulename{redFgatecall}{red\=/F\=/gatecall}
\makerulename{redFgateret}{red\=/F\=/gateret}

\makerulename{redBsame}{red\=/B\=/same}
\makerulename{redBcall}{red\=/B\=/call}
\makerulename{redBret}{red\=/B\=/ret}
\makerulename{redBgatecall}{red\=/B\=/gatecall}
\makerulename{redBgateret}{red\=/B\=/gateret}
\makerulename{redBtoIB}{red\=/B\=/to\=/IB}
\makerulename{redBinIB}{red\=/B\=/in\=/IB}

\makerulename{redWBbase}{WB\=/base}
\makerulename{redWBrec}{WB\=/rec}

\makerulename{redOcall}{oCall}
\makerulename{redOret}{oRet}
\makerulename{redOjmp}{oJmp}
\makerulename{redOstore}{oStore}

\newcommand\oc[1]{{\color{blue}{#1}}}

\newcommand\oPhi{\oc{\Phi}}
\newcommand\oerror{\oc{\mathtt{oerror}}}

\newcommand\Lrel{\mathcal{L}}

\newcommand\Frel{\mathcal{F}}

\newcommand\later{\mathord{\triangleright}}

\newcommand\K[1]{\ensuremath{\textsf{\sf#1}}}
\newcommand\KK[1]{\ensuremath{\K{\textbf{#1}}}}

\newcommand{\Wfunc}[2]{\ensuremath{(\KK{func}~\f{#1}~{#2}}\xspace}
\newcommand{\Wblock}[1]{\ensuremath{(\KK{block}~{#1}}\xspace}
\newcommand{\Wloop}[1]{\ensuremath{(\KK{loop}~{#1}}\xspace}
\newcommand{\Wif}{\ensuremath{(\KK{if}}\xspace}

\newcommand{\Wframe}[1]{\ensuremath{(\KK{frame}~{#1}}\xspace}

\newcommand{\Wend}{\ensuremath{\KK{end}})\xspace}

\newcommand{\trnative}{\texttt{Vanilla}\xspace}
\newcommand{\trlucet}{\texttt{WasmLucet}\xspace}
\newcommand{\trfullswitch}{\texttt{WasmHeavy}\xspace}
\newcommand{\trregsave}{\texttt{WasmReg}\xspace}
\newcommand{\trfast}{\texttt{WasmZero}\xspace}
\newcommand{\tridealheavy}{\texttt{IdealHeavy32}\xspace}
\newcommand{\tridealheavysixfour}{\texttt{IdealHeavy64}\xspace}
\newcommand{\trnacl}{\texttt{NaCl32}\xspace}
\newcommand{\trsegmentsfi}{\texttt{SegmentZero32}\xspace}
\newcommand{\libjpeg}{\texttt{libjpeg}\xspace}
\newcommand{\libgraphite}{\texttt{libgraphite}\xspace}
\newcommand{\libogg}{\texttt{libogg}\xspace}
\newcommand{\hunspell}{\texttt{hunspell}\xspace}
\newcommand{\libexpat}{\texttt{libexpat}\xspace}
\newcommand{\libsoundtouch}{\texttt{libsoundtouch}\xspace}
\newcommand{\simplejpeg}{\texttt{SimpleImage}\xspace}
\newcommand{\randomjpeg}{\texttt{RandomImage}\xspace}
\newcommand{\stockjpeg}{\texttt{StockImage}\xspace}
\newcommand{\gettimeofday}{\texttt{gettimeofday}\xspace}

\def\dash---{\kern.16667em---\penalty\exhyphenpenalty\hskip.16667em\relax}

\newcommand{\tranSegzeroNativeFuncDiff}{23ns\xspace}

\newcommand{\ffMaxFontSpeedupWasmZeroHeavy}{10\%\xspace}
\newcommand{\ffMaxFontSlowdownWasmHeavyZero}{10\%\xspace}
\newcommand{\ffMaxFontSlowdownWasmLucetZero}{4$\times$\xspace}
\newcommand{\ffMaxFontStackSwitchWasmOverhead}{0.8\%\xspace}
\newcommand{\ffMaxImgSpeedupWasmZeroHeavy}{29.7\%\xspace}
\newcommand{\ffMaxImgSlowdownWasmLucetZero}{9.2$\times$\xspace}
\newcommand{\ffMaxImgSimpleSlowdownWasmHeavyZero}{29.7\%\xspace}
\newcommand{\ffMaxImgSimpleSlowdownWasmLucetZero}{9.2$\times$\xspace}
\newcommand{\ffMaxImgStockRandomSlowdownWasmHeavyZero}{4.5\%\xspace}

\newcommand{\ffMaxFontOverheadNaClNative}{92\%\xspace}
\newcommand{\ffMaxFontOverheadIdealNative}{66\%\xspace}
\newcommand{\ffMaxFontOverheadSegzeroNative}{22.5\%\xspace}
\newcommand{\ffMaxImgOverheadSegzeroNative}{24\%\xspace}
\newcommand{\ffMaxImgSimpleOverheadSegzeroNative}{24\%\xspace}
\newcommand{\ffMaxImgStockOverheadSegzeroNative}{1\%\xspace}
\newcommand{\ffMaxImgRandomOverheadSegzeroNative}{6.5\%\xspace}

\newcommand{\ffMaxImgSimpleOverheadNaClNative}{312\%\xspace}
\newcommand{\ffMaxImgStockOverheadNaClNative}{29\%\xspace}
\newcommand{\ffMaxImgRandomOverheadNaClNative}{66\%\xspace}

\newcommand{\ffMaxImgSimpleOverheadIdealNative}{208\%\xspace}
\newcommand{\ffMaxImgStockOverheadIdealNative}{28\%\xspace}
\newcommand{\ffMaxImgRandomOverheadIdealNative}{45\%\xspace}

\newcommand\Intel{Intel\textsuperscript{\textregistered}\xspace}
\newcommand\ARM{ARM\textsuperscript{\textregistered}\xspace}
\newcommand\SPARC{SPARC\textsuperscript{\textregistered}\xspace}
\newcommand\SPECOhSix{SPEC CPU\textsuperscript{\textregistered} 2006\xspace}

\def\C++{C\nolinebreak[4]\hspace{-.05em}\raisebox{.4ex}{\tiny\bf ++}}

\definecolor{Ckeywordblue}{RGB}{36, 75, 133}
\definecolor{Ccommentbrown}{RGB}{142, 89, 19}
\definecolor{Csymbolorange}{RGB}{205, 92, 25}
\definecolor{Cconstantblue}{RGB}{7, 27, 203}

\lstdefinestyle{C}{%
  numbers=left,
  numberstyle=\tiny,
  basicstyle=\footnotesize\ttfamily,
  columns=fullflexible,
  keepspaces=true,
  otherkeywords={*, =},
  keywords=[1]{int, void},
  keywordstyle=[1]\color{Ckeywordblue}\bfseries,
  keywords=[2]{*, =},
  keywordstyle=[2]\color{Csymbolorange}\bfseries,
  comment=[l]{//},
  commentstyle=\color{Ccommentbrown},
  xleftmargin=2.5em
}

\newcommand\Cinline[1]{\lstinline[style=C, basicstyle=\small\ttfamily]{#1}}
\definecolor{ASMdefinitionorange}{RGB}{243, 120, 33}
\definecolor{ASMcommentbrown}{RGB}{142, 89, 19}
\definecolor{ASMsymbolorange}{RGB}{205, 92, 25}
\definecolor{ASMregisterblue}{RGB}{36, 75, 133}
\definecolor{ASMconstantblue}{RGB}{7, 27, 203}

\lstdefinestyle{asm}{%
  numbers=left,
  numberstyle=\tiny,
  basicstyle=\footnotesize\ttfamily,
  columns=fullflexible,
  keepspaces=true,
  alsoletter={:},
  otherkeywords={+, -, :=, <-},
  keywordstyle=[1]\color{ASMdefinitionorange},
  keywords=[2]{+, -, :=},
  keywordstyle=[2]\color{ASMsymbolorange}\bfseries,
  keywords=[3]{r1, r2, r3, r4, r5, r6, r7, r8, r9, r10, r11, r12, sp},
  keywordstyle=[3]\color{ASMregisterblue},
  comment=[l]{;},
  commentstyle=\color{ASMcommentbrown},
  literate={
    {<-}{{\color{ASMsymbolorange}$\leftarrow$}}2
    {>>}{{{\color{ASMsymbolorange}$\rightarrow$}}}2
  },
  xleftmargin=2.5em
}

\newcommand\Asminline[1]{\lstinline[style=asm, basicstyle=\small\ttfamily]{#1}}

\begin{document}

\title[Isolation Without Taxation]{Isolation Without Taxation}
\subtitle{Near-Zero-Cost Transitions for WebAssembly and SFI}

\author{Matthew Kolosick}
\affiliation{\institution{UC San Diego} \country{USA}}
\email{mkolosick@eng.ucsd.edu}

\author{Shravan Narayan}
\affiliation{\institution{UC San Diego} \country{USA}}
\email{srn002@eng.ucsd.edu}

\author{Evan Johnson}
\affiliation{\institution{UC San Diego} \country{USA}}
\email{e5johnso@eng.ucsd.edu}

\author{Conrad Watt}
\affiliation{\institution{University of Cambridge} \country{UK}}
\email{conrad.watt@cl.cam.ac.uk}

\author{Michael LeMay}
\orcid{0000-0001-6206-9642}
\affiliation{\institution{Intel Labs} \country{USA}}
\email{michael.lemay@intel.com}

\author{Deepak Garg}
\affiliation{\institution{Max Planck Institute for Software Systems} \country{Germany}}
\email{dg@mpi-sws.org}

\author{Ranjit Jhala}
\affiliation{\institution{UC San Diego} \country{USA}}
\email{rjhala@eng.ucsd.edu}

\author{Deian Stefan}
\affiliation{\institution{UC San Diego} \country{USA}}
\email{deian@cs.ucsd.edu}

\begin{abstract}
Software sandboxing or software-based fault isolation (SFI) is a lightweight
approach to building secure systems out of untrusted components.
Mozilla, for example, uses SFI to harden the Firefox browser by sandboxing
third-party libraries, and companies like Fastly and Cloudflare use SFI to
safely co-locate untrusted tenants on their edge clouds.
While there have been significant efforts to optimize and verify SFI
enforcement, context switching in SFI systems remains largely unexplored:
almost all SFI systems use \emph{heavyweight transitions} that are not only
error-prone but incur significant performance overhead from saving, clearing,
and restoring registers when context switching.
We identify a set of \emph{zero-cost conditions} that characterize when
sandboxed code has sufficient structured to guarantee security via lightweight
\emph{zero-cost} transitions (simple function calls).
We modify the Lucet Wasm compiler and its runtime to use zero-cost transitions,
eliminating the undue performance tax on systems that rely on Lucet for
sandboxing (e.g., we speed up image and font rendering in Firefox by up to
\ffMaxImgSpeedupWasmZeroHeavy and \ffMaxFontSpeedupWasmZeroHeavy respectively).
To remove the Lucet compiler and its correct implementation of the Wasm
specification from the trusted computing base, we (1) develop a \emph{static
binary verifier}, \verifname{}, which (in seconds) checks that binaries produced
by Lucet satisfy our zero-cost conditions, and (2)
prove the soundness of \verifname{} by developing a logical relation that
captures when a compiled Wasm function is semantically well-behaved with respect
to our zero-cost conditions.
Finally, we show that our model is useful beyond Wasm by describing a new,
purpose-built SFI system, \trsegmentsfi, that uses x86 segmentation and LLVM
with mostly off-the-shelf passes to enforce our zero-cost conditions; our
prototype performs on-par with the state-of-the-art Native Client SFI system.
\end{abstract}

\maketitle
\renewcommand{\shortauthors}{M. Kolosick, S. Narayan, E. Johnson, C. Watt, M. LeMay, D. Garg, R. Jhala, D. Stefan}

\section{Introduction}
\label{sec:intro}

Memory safety bugs are the single largest source of critical vulnerabilities 
in modern software. Recent studies found that roughly 70\% of all critical 
vulnerabilities were caused by memory safety bugs \cite{msf-bugs,chr-bugs} and 
that malicious attackers are exploiting these bugs before they can be 
patched~\cite{p0:in-the-wild:21, fireeye-study}.
Software sandboxing\dash---or software-based fault isolation
(SFI)\dash---promises to reduce the impact of such memory safety
bugs~\cite{gang-sfi, wahbe_efficient_1993}.
SFI toolkits like Native Client (NaCl)~\cite{yee_native_2009} and WebAssembly (Wasm) allow
developers to restrict untrusted components to their own \emph{sandboxed}
regions of memory thereby isolating the damage that can be caused by bugs in
these components.
Mozilla, for example, uses Wasm to sandbox third-party C libraries in
Firefox~\cite{rlbox, rlbox-blog}; SFI allows the browser to use libraries
like \libgraphite (font rendering), \libexpat (XML parsing), \libsoundtouch
(audio processing), and \hunspell (spell checking) without risking
whole-browser compromise due to library vulnerabilities.
Others have used SFI to isolate code in 
OS kernels~\cite{xfi, bgi,herder2009fault, vino}, 
databases~\cite{vx32, vxa, wahbe_efficient_1993}, 
browsers~\cite{omniware,yee_native_2009,haas_bringing_2017}, 
language runtimes~\cite{robusta, rockjit}, and 
serverless clouds~\cite{lucet-talk, cloudflare, sledge}.

SFI toolkits enforce memory isolation by placing untrusted 
code into a sandboxed environment within which every memory access 
is dynamically checked to be safe.
For example, NaCl and Wasm toolkits (e.g., Lucet~\cite{lucet} and
WAMR~\cite{wamr}) instrument memory accesses to ensure they are 
within the sandbox region and add runtime checks to ensure that 
all control flow is confined to the sandboxed paths with 
instrumented memory accesses.
There is a large body of work that ensures 
the runtime checks are \emph{fast} on different 
architectures, e.g., x86~\cite{mccamant_evaluating_2006, yee_native_2009, vx32, payer2011fine}, 
x86-64~\cite{sehr_adapting_2010}, \SPARC~\cite{omniware-pldi},
and \ARM~\cite{armor, armlock, sehr_adapting_2010}, as otherwise 
they incur unacceptable overheads on the code executing in the sandbox.
Similarly, there is a considerable literature 
that establishes that the checks are \emph{correct}~\cite{rocksalt, compcert-sfi, veriwasm, sfi-as-ai,besson2019compiling},
as even a single missing check can let 
the attacker escape the sandbox.

However, the security and overhead of software sandboxing also
crucially depends on the correctness and cost of context switching 
\dash---the \emph{trampolines} and \emph{springboards} used 
to transition into and out of sandboxes.
Almost all SFI systems, from \cite{wahbe_efficient_1993}'s original
SFI implementation to recent Wasm SFI toolkits~\cite{lucet,
wamr}, use \emph{heavyweight transitions} for context switching.\footnote{The
one exception is WasmBoxC~\cite{wasmboxc}, discussed in
Section~\ref{sec:related}.}
These transitions (1)~switch protection domains 
by tying into the underlying memory isolation mechanism 
(e.g., by setting segment registers~\cite{yee_native_2009}, 
memory protection keys~\cite{vahldiek-oberwagner_erim_2019, hodor}, 
or sandbox context registers~\cite{lucet, wamr}), and 
(2)~save, scrub, and restore machine state (e.g. 
the stack pointer, program counter, and callee-save registers) 
across the boundary.
This code is complicated and hard to get right, as it has to 
account for the particular quirks of different architectures 
and operating system platforms~\cite{alder2020faulty}.
Consequently, bugs in transition code have led to
vulnerabilities in both NaCl and Wasm\dash---from sandbox
breakouts~\cite{nacl-bug-1607, nacl-bug-1633}, to information
leaks~\cite{nacl-bug-775, nacl-bug-2919}, and application state
corruption~\cite{cranelift-bug-1177}.
Furthermore, in applications with high application-sandbox 
context switching rates, the cost of transitions
dominates the overall sandboxing overhead.
For example, heavyweight transitions prohibitively slowed 
down font rendering in Firefox, preventing Mozilla from 
shipping a sandboxed \libgraphite~\cite{rlbox}.

In this paper, we develop the principles and pragmatics needed to implement
SFI systems with near-zero-cost transitions, thereby realizing the
three-decade-old vision of reducing the cost of SFI context switches to
(almost) that of a function call.
%
% This takes the form of the following contributions:
We do this via five contributions:

\para{1. Formal model of secure transitions (\secref{sec:model})}
Simply eliminating heavyweight transitions is unsafe, potentially allowing an
attacker to escape the SFI sandbox.
To understand this threat, our first contribution is the first formal,
declarative, and high-level model that elucidates the role of transitions in
making SFI secure.
Intuitively, our model shows how secure transitions 
protect the integrity and confidentiality of machine 
state across the domain transition by providing 
\emph{well-bracketed} control flow, i.e., ensuring 
that returns actually return to their call sites.
 
\para{2. Zero-cost conditions for isolation (\secref{sec:formal-conditions})}
Heavyweight transitions provide security by wrapping 
cross-domain calls and returns to ensure that sandboxed code cannot, 
for example, read secret registers or tamper with the 
stack pointer.
While this wrapping is necessary when sandboxing arbitrary code, our insight is
these wrappers can be made redundant when the code enjoys additional structure,
not dissimilar to the additional structure typically imposed by most
SFI systems to ensure memory isolation.
For example, NaCl uses \emph{coarse-grained} control-flow integrity
(CFI) to restrict the sandbox's control flow to its own code
region~\cite{gang-sfi, yee_native_2009, haas_bringing_2017}.

We concretize this insight via our second contribution, a precise definition of
\emph{zero-cost conditions} that guarantee that sandboxed code can safely use
zero-cost transitions.
In particular, we show that transitions can be eliminated 
when sandboxed code follows a \emph{type-directed} CFI discipline, 
has well-bracketed control flow, enforces local state 
(stack and register) encapsulation, and ensures registers 
and stack slots are initialized before use.
Our notion of zero-cost conditions is inspired, in part, by techniques that use
type- and memory-safe languages to isolate code via language-level enforcement
of well-bracketed control flow and local state encapsulation~\cite{ocap, joe,
trufflec, caja, hunt2007singularity, morrisett_talx86_1999}.
However, instead of requiring developers to rewrite millions of lines of code
in high-level languages~\cite{gang-sfi}, our zero-cost conditions distill the
semantic guarantees provided by high-level languages to allow retrofitting
zero-cost transitions in the SFI setting.
In other part, our work is inspired by \citet{sfi-as-ai}, who define a
defensive semantics for SFI that captures a notion of sandboxing via simple
function calls with a stack shared between the sandbox and host application.
Our work builds on this work by addressing two shortcomings:
First, their definition does not account for confidentiality of application
data, and implementations based on their system would thus need heavyweight
transitions to prevent such attacks.
Second, their defensive semantics makes fundamental use of guard zones, which
limits the flexibility of the framework.
Our definitions of zero-cost transitions have no such limitations and fully
realize their goal of defining flexible, secure SFI with zero-cost transitions
between application and sandbox.

\para{3. Instantiating the zero-cost model (\secref{sec:implementation-security})}
We demonstrate the retrofitting of zero-cost transitions via our third
contribution, an instantiation of our zero-cost model to two SFI systems: Wasm
and \trsegmentsfi.
Previous work has shown how Wasm can provide SFI by compiling untrusted C/C++
libraries to native code using Wasm as an IR~\cite{wasm-il,wasmboxc,gobi,rlbox}.
We show that Wasm satisfies our zero-cost conditions, and replace the
heavyweight transitions used by the industrial Lucet Wasm SFI toolkit with
zero-cost transitions.
Wasm imposes more structure than required by our zero-cost conditions (and
Wasm compilers are still relatively new and slow~\cite{not-so-fast}), so, in order to
compare the overhead of our zero-cost model to the still fastest SFI
implementation\dash---NaCl~\cite{yee_native_2009}\dash---we design a new
prototype SFI system (\trsegmentsfi) that: (1) enforces our zero-cost conditions
through LLVM-level transformations, and (2) enforces memory isolation in
hardware, using 32-bit x86 segmentation.\footnote{While the prevalence of 32-bit
x86 systems is declining, it nevertheless still constitutes over 20\% of the
Firefox web browser's user base (over forty million users)~\cite{ff-stats};
\trsegmentsfi would allow for high performance library sandboxing on these
machines.}

\para{4. Verifying security at the binary level (\secref{sec:wasm-verifier})}
Our fourth contribution is a \emph{static verifier}, \verifname, that checks
whether a potentially malicious binary produced by the Lucet toolkit
satisfies our zero-cost conditions.
This removes the need to trust the Lucet compiler when, for example,
compiling third-party Firefox libraries~\cite{rlbox} or untrusted tenant code
running on Fastly's serverless cloud~\cite{lucet-talk}.
To prove the soundness of \verifname, we develop a logical relation that
captures when a compiled Wasm function is well-behaved with respect
to our zero-cost conditions and use it to prove that the checks of
\verifname{} guarantee that the zero-cost conditions are met.
We implement \verifname by extending VeriWasm~\cite{veriwasm} 
and show that in just a few seconds, it can
(1)~verify sandboxed libraries that ship 
    (or are in the process of being shipped) with Firefox, 
    Wasm-compiled \SPECOhSix benchmarks, and 100,000 programs
    randomly generated by Csmith~\cite{csmith}, and 
(2)~catch previous NaCl and Wasm vulnerabilities (\secref{subsec:verifier-eval}).
\verifname is being integrated into the Lucet industrial Wasm
compiler~\cite{verizero-integration}.

\para{5. Implementation and evaluation (\secref{sec:eval})}
Our last contribution is an implementation of our zero-cost 
sandboxing toolkits, and an evaluation of how they improve 
the performance of a transition micro-benchmark and two 
macro-benchmarks\dash---image decoding (\libjpeg) 
and font rendering (\libgraphite) in Firefox.
First, we demonstrate the potential performance 
of a purpose-built zero-cost SFI system, by evaluating 
\trsegmentsfi on \SPECOhSix and our macro-benchmarks.
We find that \trsegmentsfi imposes at most 25\%
overhead on \SPECOhSix (nc), and at most \ffMaxImgOverheadSegzeroNative 
on image decoding and \ffMaxFontOverheadSegzeroNative on font rendering.
These overheads are lower than the state-of-art NaCl SFI system.
On the macro-benchmarks, \trsegmentsfi even outperforms an idealized SFI system
that enforces memory isolation for free but requires heavyweight transitions.
Second, we find that zero-cost transitions speed up 
Wasm-sandboxed image decoding by (up to) \ffMaxImgSpeedupWasmZeroHeavy 
and font rendering by \ffMaxFontSpeedupWasmZeroHeavy.
The speedup resulting from our zero-cost transitions allowed 
Mozilla to ship the Wasm-sandboxed \libgraphite library 
in production.

\para{Open source and data}
Our code and data will be made available under an open source license.

\section{Overview} \label{sec:isolation}

In this section we describe the role of transitions in making SFI secure, give
an overview of existing heavyweight transitions, and introduce our zero-cost
model, which makes it possible for SFI systems to replace heavyweight
transitions with simple function calls.

\subsection{The need for secure transitions}
\label{subsec:overview-secure}

As an example, consider sandboxing an untrusted font rendering library (e.g.,
\libgraphite) as used in a browser like Firefox:
\begin{lstlisting}[style=C, escapeinside=||]
void onPageLoad(int* text) {
  ...
  int* screen = ...; // stored in r12
  int* temp_buf = ...;
  gr_get_pixel_buffer(text, temp_buf);
  memcpy(screen, temp_buf, 100); |\label{code:memcpy}|
  ...
}
\end{lstlisting}
This code calls the \libgraphite \Cinline{gr_get_pixel_buffer} function to
render text into a temporary buffer and then copies the temporary buffer to the
variable \Cinline{screen} to be rendered.

\sloppy
Using SFI to sandbox this library ensures that the browser's memory is isolated
from \texttt{libgraphite}\dash---memory isolation ensures that
\Cinline{gr_get_pixel_buffer} cannot access the memory of \Cinline{onPageLoad}
or any other parts of the browser stack and heap.
Unfortunately, memory isolation alone is not enough:
if transitions are simply function calls, attackers can violate the calling
convention at the application-library boundary (e.g., the
\Cinline{gr_get_pixel_buffer} call and its return) to break isolation.
Below, we describe the different ways a compromised \libgraphite can do this.

\para{Clobbering callee-save registers}
Suppose the \Cinline{screen} variable in the above \Cinline{onPageLoad} snippet
is compiled down to the register \Asminline{r12}.
In the System V calling convention \Asminline{r12}
is a \emph{callee-saved} register~\cite{system-v},
so if \Cinline{gr_get_pixel_buffer} clobbers
\Asminline{r12}, then it is also supposed to restore it to
its original value before returning to \Cinline{onPageLoad}.
A compromised \libgraphite doesn't have to do this; instead, the attacker can
poison the register:
\begin{lstlisting}[style=asm, escapeinside=||]
mov r12, |\Biohazard|
ret
\end{lstlisting}
\noindent Since \Asminline{r12} (\Cinline{screen}) in our hypothetical example
is then used on \coderef{code:memcpy} to \Cinline{memcpy} the
\Cinline{temp_buf} from the sandbox memory, this gives the attacker a write
gadget that they can use to hijack Firefox's control flow.
To prevent such attacks, we need \emph{callee-save register integrity}, i.e.,
we must ensure that sandboxed code restores callee-save registers upon
returning to the application.

\para{Leaking scratch registers}
Dually, \emph{scratch registers} can potentially leak sensitive information
into the sandbox.
Suppose that Firefox keeps a secret (e.g., an encryption key) in a
scratch register.
Memory isolation alone would not prevent an attacker-controlled \libgraphite
from using uninitialized registers, thereby reading this secret.
To prevent such leaks, we need \emph{scratch register confidentiality}.

\para{Reading and corrupting stack frames}
Finally, if the application and sandboxed library share a stack, the attacker
could potentially read and corrupt data (and pointers) stored on the stack.
To prevent such attacks, we need \emph{stack frame encapsulation}, i.e., we
need to ensure that sandboxed code cannot access application stack frames.

\subsection{Heavyweight transitions}
\label{sec:background-heavyweight}

SFI toolchains\dash---from NaCl~\cite{yee_native_2009} to Wasm native compilers
like Lucet~\cite{lucet} and WAMR~\cite{wamr}\dash---use \emph{heavyweight
transitions} to wrap calls and returns and address the aforementioned attacks.
These heavyweight transitions are secure transitions.
They provide:

\para{1. Callee-save register integrity}
The \emph{springboard}\dash---the transition code which wraps calls\dash---saves
callee-save registers to a separate stack stored in protected application
memory.
When returning from the library to the application,
the \emph{trampoline}\dash---the code which wraps returns\dash---restores the
registers.

\para{2. Scratch register confidentiality}
Since any scratch register may contain secrets, the springboard clears
\emph{all} scratch registers before transitioning into the sandbox.

\para{3. Stack frame encapsulation}
Most (but not all) SFI systems provision separate stacks for trusted and
sandboxed code and ensure that the trusted stack is not accessible from the
sandbox.
The springboard and trampoline account for this in three ways.
First, they track the separate stack pointers at each transition in order to
switch stacks.
Second, the springboard copies arguments passed on the stack to the sandbox
stack, since sandboxed code cannot access arguments stored on the
application stack.
Finally, the trampoline tracks the actual return address on transition by
keeping it in the protected memory, so that the sandboxed library cannot tamper
with it.

\para{The cost of wrappers}
Heavyweight springboards and trampolines guarantee secure transitions
but have two significant drawbacks.
First, they impose an overhead on SFI\dash---calls into the sandboxed library become significantly more expensive than
simple application function calls  (\secref{sec:eval}).
Heavyweight transitions conservatively save and clear more state than might be necessary, essentially
reimplementing aspects of an OS process switch and duplicating work done by
well-behaved libraries.
Second, springboards and trampolines must be customized to different platforms,
i.e., different processors and calling conventions, and, in extreme cases such
as in \citet{vahldiek-oberwagner_erim_2019}, even different applications.
Implementation mistakes
can\dash---and have~\cite{nacl-bugs, nacl-bug-2919, nacl-bug-775, nacl-bug-1607,
nacl-bug-1633, bartel2018twenty}\dash---resulted in sandbox escape attacks.

\subsection{Zero-cost transitions}
\label{subsec:overview-zero}

Heavyweight transitions are conservative because they make few assumptions
about the structure (or possible behavior) of the code running in the sandbox.
SFI systems like NaCl and Wasm \emph{do}, however, impose structure on
sandboxed code to enforce memory isolation.
In this section we show that by imposing structure on sandboxed code we can
make transitions less conservative.
Specifically, we describe a set of \emph{zero-cost conditions}
that impose \emph{just enough} internal structure on sandboxed code to ensure
that it will behave like a high-level, compositional language while maintaining
SFI's high performance.
SFI systems that meet these conditions can safely elide almost all the extra
work done by heavyweight springboards and trampolines, thus moving toward the
ideal of SFI transitions as simple, fast, and portable function calls.

\para{Zero-cost conditions}
We assume that the sandboxed library code is split into functions and that each
function has an expected number of arguments.
We \emph{formalize} the internal structure required of library code via a
\emph{safety monitor} that checks the zero-cost conditions, i.e., the local
requirements necessary to ensure that calls-into and returns-from the untrusted
library functions are ``well-behaved'' and, hence, that they satisfy the secure
transition requirements.

\para{1. Callee-save register restoration}
First, our monitor enforces function-call level adherence to callee-save
register conventions: our monitor tracks callee-save state and checks that it
has been correctly restored upon returning.
Importantly, satisfying the monitor means that application calls to a
well-behaved library function do not require a transition which separately saves
and restores callee-save registers, since the function is known to obey the
standard calling convention.

\para{2. Well-bracketed control-flow}
Second, our monitor requires that the library code adheres to well-bracketed
return edges.
Abstractly, calls and returns should be well-bracketed: when \verb+f+ calls
\verb+g+ and then \verb+g+ calls \verb+h+, \verb+h+ ought to return to \verb+g+
and then \verb+g+ ought to return to \verb+f+.
However, untrusted functions may subvert the control stack to implement
arbitrary control flow between functions.
This unrestricted control flow is at odds with compositional reasoning,
preventing \emph{local} verification of functions.
Further, subverting well-bracketing could enable an attacker to cause \verb+h+
to return directly to \verb+f+.
Then, even if \verb+h+ and \verb+f+ both restore their callee-save registers,
those of \verb+g+ would be left unrestored.
Accordingly, we require two properties of the library to ensure that calls and
returns are well-bracketed.
First, each jump must stay within the same function.
This limits inter-function control flow to function calls and returns.
Second, the (specification) monitor maintains a ``logical'' call stack, 
which is used to ensure that returns go only to the preceding caller.

\para{3. Type-directed forward-edge CFI}
Our monitor also requires that library code obeys type-directed forward-edge
CFI.
That is, for every call instruction encountered during execution, the jump
target address is the start of a library function and the arguments
passed match those expected by the called function.
This ensures that each function starts from a (statically) known stack
shape, preventing a class of attacks where a benign function can be tricked into
overwriting other stack frames or hijacking control flow because it is passed
too few (or too many) arguments.
If this were not the case, a locally well-behaved function that was passed too
few arguments could write to a saved register or the saved return address,
expecting that stack slot to be the location of an argument.

\para{4. Local state encapsulation}
Our monitor establishes \emph{local state encapsulation} by checking
that all stack reads and writes are within the current stack frame.
This check allows us to \emph{locally}, i.e., by checking each function in
isolation, ensure that a library function correctly saves and restores
callee-save registers upon entry and exit.
To see why local state encapsulation is needed, consider
the following idealized assembly function \Asminline{library_func}:
\begin{lstlisting}[style=asm, escapeinside=||, morekeywords={library_func:, library_helper:}]
library_func:         library_helper:
  push r12              store sp - 1 := |\Biohazard|
  mov r12 <- 1          ret
  load r1 <- sp - 1
  add r1 <- r12
  call library_helper
  pop r12
  ret
\end{lstlisting}
If \Asminline{library_helper} is called it will overwrite the stack slot where
\Asminline{library_func} saved \Asminline{r12}, and \Asminline{library_func}
will then ``restore'' \Asminline{r12} to the attacker's desired value.
Our monitor prohibits such cross-function tampering, thus ensuring that
all subsequent reasoning about callee-save integrity can be carried out
locally in each function.

\para{5. Confidentiality}
Finally, our monitor uses dynamic information flow control (IFC) tracking to
define the confidentiality of scratch registers.
The monitor tracks how (secret application) values stored in scratch registers
flow through the sandboxed code, and checks that the library code does not leak
this information.
Concretely, our implementations enforce this by ensuring that, within each
function's localized control flow, all register and local stack
variables are initialized before use.

The individual properties making up our zero-cost conditions are well-known to be beneficial to software security, and their enforcement in low-level code has been extensively studied~(\secref{sec:related}): our insight\dash---made manifest in the monitor
soundness proofs of \sectionref{sec:overlay:secure}\dash---is that in
conjunction these conditions are \emph{sufficient} to eliminate heavyweight
transitions in SFI systems, which can currently be a source of significant overhead when sandboxing arbitrary code.
Indeed, in \sectionref{sec:web-assembly-secure} we show that the Wasm
type system is strict enough to ensure that a Wasm compiler generates native
code that already meets these conditions.
To increase the trustworthiness of this zero-cost compatible Wasm, in
\sectionref{sec:wasm-verifier} we describe a verifier that statically checks
that compiled Wasm code meets the zero-cost conditions.
In \sectionref{sec:wasm-proof} we describe our proof of soundness for the
verifier, proving that the verifier's checks ensure monitor safety and therefore
zero-cost security.
Further, in \sectionref{sec:segments-secure} we demonstrate how the zero-cost
conditions can be used to design a new SFI scheme by combining hardware-backed
memory isolation with existing LLVM compiler passes.

\section{A Gated Assembly Language}
\label{sec:model}

\begin{figure}[t]
  \begin{small}
  \begin{tabular}{@{}>{$}r<{$} >{$}c<{$} >{$}r<{$} >{$}c<{$} >{$}l<{$}}
    & & n & \in & \nats \\
    \privs & \ni & p & \bnfdef & \trusted \bnfalt \untrusted \\
    \vals & \ni & v & \bnfdef & n \\
    \regs & \ni & r & \bnfdef & \mathtt{r}_{n} \bnfalt sp \bnfalt pc \\
    \regions & \ni & k & \in & \nats \rightharpoonup \nats \\

    \expressions & \ni & e & \bnfdef & r \bnfalt v \bnfalt e \oplus e \\

    \commands & \ni & c & \bnfdef & \cpop{r}{p} \bnfalt \cpush{p}{e} \bnfalt \cjmp{k}{e} \bnfalt \cload{r}{k}{e} \bnfalt \cstore{k}{e}{e} \bnfalt \\
              &     &   &         & \cgatecall{n}{e} \bnfalt \cgateret \bnfalt \cmov{r}{e} \bnfalt \ccall{k}{e} \bnfalt \cret{k} \bnfalt \\
              &     &   &         & \cmovlabel{r}{p} \bnfalt \cstorelabel{p}{e} \\
    \codes    & \ni & C & \bnfdef & \nats \rightharpoonup \privs \times \commands \\
    \regvals  & \ni & R & \bnfdef & \regs \rightarrow \vals \\
    \memories & \ni & M & \bnfdef & \nats \rightarrow \vals \\
    \states   & \ni & \Psi & \bnfdef & \error \bnfalt \{ pc \bnftypes \nats, sp \bnftypes \nats, R \bnftypes \regvals, M \bnftypes \memories, C \bnftypes \codes \}
  \end{tabular}
  \end{small}
  \caption{Syntax}
  \label{fig:formalism:syntax:lang}
\end{figure}

We formalize zero-cost transitions via an assembly language, \langname{}, that
captures key notions of an application interacting with a sandboxed library,
focusing on capturing properties of the transitions between the application and
sandboxed library.

\para{Code}
\figref{fig:formalism:syntax:lang}
summarizes the syntax of \langname{}:
a \textsc{Risc}-style language with natural
numbers ($\nats$) as the sole data type.
Code ($C$) and data ($M$) memory are separated, and, to capture the separation
of application code from sandboxed library code, $C$ is an (immutable) partial
map from $\nats$ to pairs of a privilege ($p$) ($\trusted$ or $\untrusted$) and
a command ($c$), where $\trusted$ and $\untrusted$ are our \emph{security
domains}.

\para{States}
Memory is a (total) map from $\nats$
to values ($v$).
We assume that the memory is subdivided
into disjoint regions ($M_p$) so that
the application and library have separate memory.
Each of these regions is further divided
into a disjoint heap $H_p$ and stack $S_p$.
We write $\Psi$ to denote the states
or machine configurations, which comprise
code, memory, and a fixed, finite set of
registers mapping register names $(r_n)$
to values, with a distinguished stack
pointer ($sp$) and program counter ($pc$)
register.
We write $\currentcom{\Psi}{c}{p}$
for $\Psi.C(\Psi.pc) = (p, c)$,
that is that the current instruction
is $c$ in security domain $p$.
We write $\Psi_0 \in \programs$
to mean that $\Psi_0$ is a valid
initial program state.
The definition of validity varies between different SFI techniques (e.g.,
heavyweight transitions make assumptions about the initial state of the
separate stack).

\para{Gated calls and returns}
We capture the transitions between the application and the library by defining a
pair of instructions $\cgatecall{n}{e}$ and $\cgateret$, that serve as the
\emph{only} way to switch between the two security domains (that is,
$\cinst{call}$ and $\cinst{ret}$ check that the target is in the same security
domain).
The first, $\cgatecall{n}{e}$,
represents a call from the
application into the sandbox
or a callback from the sandbox
to the application with the $n$
annotation representing the number
of arguments to be passed.
The second, $\cgateret$, represents
the corresponding return from sandbox
to application or vice-versa.
We leave the reduction rule for both
\emph{implementation specific} in order
to capture the details of a given SFI
system's trampolines and springboards.

\para{Memory isolation}
\langname{} provides abstract mechanisms
for enforcing SFI memory isolation by
equipping the standard $\mathtt{load}$,
$\mathtt{store}$, $\mathtt{push}$, and
$\mathtt{pop}$ with (optional) statically
annotated checks.
To capture different styles of enforcement we model these checks as partial
functions that map a pointer to its new value or are undefined when a particular
address is invalid.
This lets us, for instance, capture NaCl's
coarse grained, dynamically enforced
isolation (sandboxed code may read
and write anywhere in the sandbox memory)
by requiring that all loads and stores
are annotated with the check $k(n)|_{n \in M_{\untrusted}} = n$.
This captures that NaCl's memory isolation does not remap addresses but traps
when an address is outside the sandbox memory region
($M_{\untrusted}$).\footnote{NaCl implements memory protection differently on
different platforms. The 32-bit implementation traps whereas the 64-bit
implementation masks addresses. We focus on the former.}
The rule for $\cinst{load}$ below demonstrates the use of these region
annotations in the semantics.

\para{Control-flow integrity}
\langname{} also provides abstract
control-flow integrity enforcement
via annotations on $\mathtt{jmp}$,
$\mathtt{call}$, and $\mathtt{ret}$.
These are also enforced dynamically.
However, we require that the standard
control flow operations remain within
their own security domain so that
$\mathtt{gatecall}$ and $\mathtt{gateret}$
remain the only way to switch
security domains.

\para{Operational semantics}
We capture the dynamic behavior via 
a deterministic small-step operational
semantics ($\Psi \step \Psi'$).
The rules are standard; we show the
rule for $\mathtt{load}$ here:
\begin{small}
\begin{mathpar}
  \inferrule
  {
    addr = \immval{\Psi}{e}
    \\ addr' = k(addr)
    \\\\ v = \Psi.M(addr')
    \\ R' = \Psi.R[r \mapsto v]
  }
  {\currentop{\Psi}{\cload{r}{k}{e}} \step \pcinc{\Psi}[R \assign R']}
\end{mathpar}
\end{small}
$\immval{\Psi}{e}$ evaluates the expression based on the register file and
$\pcinc{\Psi}$ increments $pc$, checking that it remains within the same
security domain and returning an error otherwise.
If the function $k(addr)$ is undefined ($addr$ is not within bounds), the
program will step to a distinguished, terminal state $\error$.
$\currentop{\Psi}{c}$ is simply shorthand for $\currentcom{\Psi}{c}{p}$ when we
do not care about the security domain.
Lastly, we do not include a specific halt command, instead halting when $pc$ is
not in the domain of $C$.

\subsection{Secure transitions}
\label{sec:assembly:security}

Next, we use \langname{} to \emph{declaratively} specify high-level properties
that capture the intended security goals of transition systems.
This lets us use \langname{} both as a setting to study zero-cost
transitions and to explore the correctness of implementations of
springboards and trampolines.
As a demonstrative example we prove that NaCl-style heavyweight transitions
satisfy the high-level properties (\iftechreport{\secref{appendix:nacl}}{see the
technical appendix~\cite{kolosick2021isolation}}).

\begin{figure}[t]
  \begin{small}
  \begin{mathpar}
    \inferrule
    {
      \Psi_1 \step \Psi_2
      \\ \currentcom{\Psi_1}{c_1}{p_1}
      \\\\ \currentcom{\Psi_2}{c_2}{p_2}
      \\ p_1 = p_2 = p
    }
    {\Psi_1 \stepp{p} \Psi_2}

    \inferrule
    {\Psi \stepp{p} \Psi'}
    {\Psi \stepbox \Psi'}

    \inferrule
    {\Psi \stepwb \Psi'}
    {\Psi \stepbox \Psi'}

    \inferrule
    {
      \Psi \step \Psi_1 \stepboxstar \Psi_2 \step \Psi'
      \\\\ \currentop{\Psi}{\cgatecall{n}{i}}
      \\ \currentop{\Psi_2}{\cgateret}
    }
    {\Psi \stepwb \Psi'}
  \end{mathpar}
  \end{small}
  \caption{Well-Bracketed Transitions}
  \label{fig:well-bracketed-transition}
\end{figure}

\para{Well-bracketed gated calls}
SFI systems may allow arbitrary \emph{nesting}
of calls into and callbacks out of the sandbox.
Thus, it is insufficient to define that callee-save registers have been properly
restored by simply equating register state upon entry to the sandbox and the
following exit.
Instead we make the notion of an entry and its \emph{corresponding} exit
precise, by using \langname{}'s $\mathtt{gatecall}$ and $\mathtt{gateret}$ to
define a notion of \emph{well-bracketed gated calls} that serve as the backbone
of transition integrity properties.
A well-bracketed gated call, which we write $\Psi \stepwb \Psi'$
(\figref{fig:well-bracketed-transition}), captures the idea that $\Psi$ is a
gated call from one security domain to another, followed by running in the new
security domain, and then $\Psi'$ is the result of a gated return that balances
the gated call from $\Psi$.
This can include potentially recursive but properly bracketed gated calls.
Well-bracketed gated calls let us relate the state before a gated call with the
state after the \emph{corresponding} gated return, capturing when the library
has fully returned to the application.

\para{Integrity}
Relations between the states before calling into the sandbox and then after the
corresponding return capture SFI transition system \emph{integrity} properties.
We identify two key integrity properties that SFI transitions must maintain:

\emph{1. Callee-save register integrity}
requires that callee-save registers are restored after returning from a gated
call into the library.
This ensures that an attacker cannot unexpectedly modify the private state of
an application function.

\emph{2. Return address integrity}
requires that the sandbox
\begin{enumerate*}
\item returns to the instruction after the $\mathtt{gatecall}$,
\item does not tamper with the stack pointer, and
\item does not modify the call stack itself.
\end{enumerate*}
Together these ensure that an attacker cannot tamper with the application
control flow.

These integrity properties are crucial to ensure that the sandboxed library cannot
break application invariants.
To capture them formally, we first define an abstract notion of integrity
across a well-bracketed gated call.
This not only allows us to cleanly define the above properties, but also
provides a general framework that can capture integrity properties for different
architectures.

Specifically, we define an integrity property by a predicate $\mathcal{I} :
\mathit{Trace} \times \mathit{State} \times \mathit{State} \rightarrow \prop$
that captures when integrity is preserved across a call ($\prop$ is the type of
propositions).
The first argument is a trace, a sequence of steps that our program has taken
before making the gated call.
The next two arguments are the states before and after the well-bracketed gated
call.
$\mathcal{I}$ defines when these two states are properly related.
This leads to the following definition of $\mathcal{I}$-Integrity:

\begin{definition}[$\mathcal{I}$-Integrity]
  Let $\mathcal{I} : \mathit{Trace} \times \mathit{State} \times \mathit{State} \rightarrow \prop$.
  We say that an SFI transition system has $\mathcal{I}$-integrity if
  $\Psi_0 \in \programs$, $\pi = \Psi_0 \stepstar \Psi_1$,
  $\currentcom{\Psi_1}{\_}{\trusted}$, and $\Psi_1 \stepwb \Psi_2$ imply that
  $\mathcal{I}(\pi, \Psi_1, \Psi_2)$.
\end{definition}

\noindent
We instantiate this to define our two integrity properties:

\para{Callee-save register integrity}
We define callee-save register integrity as an $\mathcal{I}$-integrity property
that requires the callee-save registers' values to be equal in both states:
\begin{definition}[Callee-Save Register Integrity]
  Let $\mathbb{CSR}$ be the callee-save registers and define
  $\mathcal{CSR}(\_, \Psi_1, \Psi_2) \triangleq \Psi_2.R(\mathbb{CSR}) = \Psi_1.R(\mathbb{CSR})$.
  If an SFI transition system has $\mathcal{CSR}$-integrity then we say it has callee-save register integrity.
\end{definition}

\para{Return address integrity}
We specify that the library returns to the expected instruction as a relation
between $\Psi_1$ and $\Psi_2$, namely that $\Psi_2.pc = \Psi_1.pc + 1$.
Restoration of the stack pointer is similarly specified as
$\Psi_2.sp = \Psi_1.sp$.
Specifying call stack integrity is more involved as $\Psi_1$ lacks information
on where return addresses are saved: they look like any other stack data.
Instead, return addresses are defined by the history of calls and returns
leading up to $\Psi_1$, which we capture with the trace argument $\pi$.
We thus define a function $\oname{return-address}(\pi)$ (\iftechreport{see
\figref{fig:appendix:return-address} in the appendix}{see the technical
appendix~\cite{kolosick2021isolation}}) that computes the locations of
return addresses from a trace.
The third clause of return address integrity
is then that these locations' values are
preserved from $\Psi_1$ to $\Psi_2$, yielding:

\begin{definition}[Return Address Integrity]
  \begin{align*}
    \mathcal{RA}(\pi, \Psi_1, \Psi_2) &\triangleq \Psi_2.pc = \Psi_1.pc + 1 \wedge \Psi_2.sp = \Psi_1.sp \\
    & \wedge \Psi_2.M(\oname{return-address}(\pi)) = \Psi_1.M(\oname{return-address}(\pi))
  \end{align*}
  If an SFI transition system has $\mathcal{RA}$-integrity then we say the system has return address integrity.
\end{definition}

\para{Confidentiality}
SFI systems must ensure that secrets cannot be leaked to the untrusted
library, i.e., they must provide \emph{confidentiality}.
We specify confidentiality as noninterference, which informally states that
``changing secret inputs should not affect public outputs.''
In the context of library sandboxing, application data is secret whereas
library data is non-secret (public).\footnote{This could also be extended to
a setting with mutually distrusting components.}
To capture this formally, we pair programs with a confidentiality policy,
$\mathbb{C} \in \states \rightharpoonup (\nats \mathrel{+} \regs \rightarrow
\privs)$, that labels all memory and registers as $\trusted$ or $\untrusted$ at
each gated call into the library.
These labels form a lattice:
$\untrusted \lesstrusted \trusted$ (non-secret \emph{can} ``flow to'' secret)
and $\trusted \nlesstrusted \untrusted$ (secret \emph{cannot} ``flow to''
non-secret).\footnote{Details can be found in
\iftechreport{\sectionref{sec:appendix:confidentiality}}{the technical
appendix~\cite{kolosick2021isolation}}.}

To prove noninterference, that changing secret data does not affect public (or
non-secret) outputs, we need to define public outputs.
We over-approximate public outputs as the set of values \emph{exposed} to the
application.
This includes all arguments to a $\mathtt{gatecall}$ callback, the return value
when returning to the application via $\mathtt{gateret}$, and all values
stored in the sandboxed library's heap ($H_{\untrusted}$) (which may be
referenced by other returned values).

Alas, this is not enough: in a callback, the application may choose to
declassify secret data.
For instance, a sandboxed image decoding library might, after parsing the file
header, make a callback requesting the data to decode the rest of the image.
This application callback will then transfer that data (which was previously
confidential) to the sandbox, declassifying it in the transfer.

To account for such intentional declassifications, we follow
\citet{matos_declassification_2005} and define confidentiality as
\emph{disjoint noninterference}.
We use $\Psi =_{\mathbb{C}} \Psi'$ to mean that $\Psi$ and $\Psi'$ agree on all values
labeled $\untrusted$ by the confidentiality policy, capturing varying secret inputs.
We further write $\Psi =_{\mathtt{call}\ m} \Psi'$ when $\Psi$ and $\Psi'$
agree on all sandboxed heap values, the program counter, and the $m$ arguments
passed to a callback and $\Psi =_{\mathtt{ret}} \Psi'$ when $\Psi$ and
$\Psi'$ agree on all sandboxed heap values, the program counter, and the value
in the return register (written $r_{ret}$).\footnote{Full definitions are in
\iftechreport{\appref{appendix:language-defs} \defref{appendix:call-equivalence} and
\defref{appendix:ret-equivalence}}{the technical appendix~\cite{kolosick2021isolation}}.}
This lets us formally define noninterference as follows:

\begin{definition}[\StrongNI{}]{~}

  We say that an SFI transition system has the \strongni{} property if,
  for all initial configurations and their confidentiality properties $(\Psi_0, \mathbb{C}) \in \programs$,
  traces $\Psi_0 \stepstar \Psi_1 \step \Psi_2 \steplowstar \Psi_3 \step \Psi_4$,
  where $\Psi_1$ is a gated call into the library
  ($\currentcom{\Psi_1}{\cgatecall{n}{e}}{\trusted}$),
  and $\Psi_3 \step \Psi_4$ leaves the library and reenters the application
  ($\currentcom{\Psi_4}{\_}{\trusted}$),
  and, for all $\Psi_1'$ such that $\Psi_1 =_{\mathbb{C}}
  \Psi_1'$, we have that $\Psi_1' \step \Psi_2' \steplowstar \Psi_3' \step \Psi_4'$,
  $\currentcom{\Psi_4'}{\_}{\trusted}$, $\Psi_4.pc = \Psi_4'.pc$, and either
  \begin{enumerate*}
  \item $\Psi_3$ is a gated call to the application ($\currentop{\Psi_3}{\cgatecall{m}{e}}$ and $\currentop{\Psi_3'}{\cgatecall{m}{e}}$) and $\Psi_4 =_{\mathtt{call}\ m} \Psi_4'$ or
  \item $\Psi_3$ is a gated return to the application ($\currentop{\Psi_3}{\cgateret}$ and $\currentop{\Psi_3'}{\cgateret}$) and $\Psi_4 =_{\mathtt{ret}} \Psi_4'$.
  \end{enumerate*}
\end{definition}
This definition captures that, for any sequence of executing
within the library then returning control to the application, varying
confidential inputs does not influence the public outputs and the library
returns control to the application in the same number of steps.
Thus, an SFI system that satisfies \StrongNI{} is guaranteed
to not leak data while running within the sandbox.

We formalize NaCl style heavyweight transitions in
\iftechreport{\appref{appendix:nacl}}{the technical
appendix~\cite{kolosick2021isolation}} and prove that they meet the above secure
transition properties.
We discuss our proof that zero-cost Wasm meets the above secure transition
properties in \sectionref{sec:wasm-proof}.

\section{Zero-Cost Transition Conditions}
\label{sec:formal-conditions}

\begin{figure}[t]
  \begin{small}
  \begingroup
  \setlength{\tabcolsep}{1mm}
  \setlength{\arraycolsep}{1mm}
  \begin{tabular}{>{$}r<{$} >{$}c<{$} >{$}r<{$} >{$}c<{$} >{$}l<{$}}
    \vals & \ni & \oc{v} & \bnfdef & \val{n}{p}
    \\

    \frames & \ni & \oc{\SF} & \bnfdef &
        \{
        \mathit{base} \bnftypes \nats;
        \mathit{ret\mbox{-}addr\mbox{-}loc} \bnftypes \nats;
        \mathit{csr\mbox{-}vals} \bnftypes \powerset{\regs \times \nats} \}
    \\

    \functions & \ni & \oc{F} & \bnfdef &
        \{
        \mathit{instrs} \bnftypes \nats \rightharpoonup \commands;
        \mathit{entry} \bnftypes \nats;
        \mathit{type} \bnftypes \nats \}
    \\

    \ostates & \ni & \oc{\Phi} & \bnfdef & \oerror \bnfalt
        \{
        \Psi \bnftypes \states;
        \mathit{funcs} \bnftypes \nats \rightharpoonup \functions;
        \mathit{stack} \bnftypes [\frames] \}
  \end{tabular}
  \endgroup
  \end{small}
  \caption{\olangname{} Extended Syntax}
  \label{fig:overlay:syntax}
\end{figure}

Having laid out the security properties required of an SFI transition
system, we turn to formally defining the set of zero-cost conditions on sandboxed
code such that they sufficiently capture when we may securely elide springboards
or trampolines.
To this end we define our zero-cost conditions as a safety monitor via the
language \olangname{} overlaid on top of \langname{}.
\olangname{} extends \langname{} with additional structure and dynamic type
checks that ensure the invariants needed for zero-cost transitions are
maintained upon returning from library functions, providing both an inductive
structure for proofs of security for zero-cost implementations and providing a
top-level guarantee that our integrity and confidentiality properties are
maintained.
In \sectionref{sec:overlay:secure} we outline the proofs of overlay soundness,
showing that \olangname{} captures when a system is zero-cost secure.

\para{Syntax of \olangname{}}
\figref{fig:overlay:syntax} shows the extended syntax of \olangname{}.
Values ($\oc{v}$) are extended with a security label $p$.
Overlay state, written $\oc{\Phi}$, wraps the state of \langname{}, extending it
with two extra pieces of data.
First, \olangname{} requires the sandboxed code be organized into functions
($\oc{\Phi}.\mathit{funcs}$).
$\oc{\Phi}.\mathit{funcs}$ maps each command in the sandboxed library to its
parent function.
Functions ($\oc{F}$) also store the code indices of their commands as the field
$\oc{F}.\mathit{instrs}$, store the entry point ($\oc{F}.\mathit{entry}$),
and track the number of arguments the function expects ($\oc{F}.\mathit{type}$).
This partitioning of sandboxed code into functions is static.
Second, the overlay state dynamically tracks a list of overlay stack frames
($\oc{\Phi}.\mathit{stack}$).
These stack frames ($\oc{\SF}$) are solely logical and inaccessible to instructions.
They instead serve as bookkeeping to implement the dynamic type checks of
\olangname{} by tracking the base address of each stack frame
($\oc{\SF}.\mathit{base}$), the stack location of the return address
($\oc{\SF}.\mathit{ret\mbox{-}addr}$), and the values of the callee
save registers upon entry to the function ($\oc{\SF}.\mathit{csr\mbox{-}vals}$).
We are concerned with the behavior of the untrusted library, so the logical
stack does not finely track application stack frames, but keeps a single large
``stack frame'' for all nested application stack frames.

When code fails the overlay's dynamic checks it will result in the state
$\oerror$.
Our definition of monitor safety, which will ensure that zero-cost transitions
are secure, is then simply that a program does not step to an $\oerror$.

\subsection{Overlay monitor} \label{subsec:overlay-monitor-def}

\begin{figure}[t]
  \begin{small}
  \begin{mathpar}
    \inferrule[\defredOcall]
    {
      \val{n}{\untrusted} = \oimmval{\oc{\Phi}}{e}
      \\ n' = k(n)
      \\ sp' = \oc{\Phi}.sp + 1
      \\\\ M' = \oc{\Phi}.M[sp' \mapsto \oc{\Phi}.pc + 1]
      \\ \mathit{stack}' = [\oc{\SF}] \concat \oc{\Phi}.\mathit{stack}
      \\ \oc{\SF} = \oname{new-frame}(\oc{\Phi}, n', sp')
      \\ \oname{typechecks}(\oc{\Phi}, n', sp')
      \\ \oc{\Phi'} = \oc{\Phi}[\mathit{stack} \assign \mathit{stack}', pc \assign n', sp \assign sp', M \assign M']
    }
    {\currentcom{\oc{\Phi}}{\ccall{k}{e}}{\untrusted} \ostep \oc{\Phi'}}

    \inferrule[\defredOret]
    {
      \oname{is-ret-addr}(\oc{\Phi}, \oc{\Phi}.sp)
      \\ \natval{n} = \oc{\Phi}.M(\oc{\Phi}.sp)
      \\ n' = k(n)
      \\\\ \oname{csr-restored}(\oc{\Phi})
      \\ \oc{\Phi'} = \oname{pop-frame}(\oc{\Phi})
    }
    {\currentcom{\oc{\Phi}}{\cret{k}}{\untrusted} \ostep \oc{\Phi'}[pc \assign n', sp \assign \oc{\Phi}.sp - 1]}

    \inferrule[\defredOjmp]
    {
      \val{n}{\untrusted} = \oimmval{\oc{\Phi}}{e}
      \\ n' = k(n)
      \\\\ \oname{in-same-func}(\oc{\Phi}, \oc{\Phi}.pc, n')
    }
    {\currentcom{\oc{\Phi}}{\cjmp{k}{e}}{\untrusted} \ostep \oc{\Phi}[pc \assign n']}

    \inferrule[\defredOstore]
    {
      \val{n}{\untrusted} = \oimmval{\oc{\Phi}}{e}
      \\ \oc{v} = \val{\_}{p_{e'}} = \oimmval{\oc{\Phi}}{e'}
      \\ M' = \oc{\Phi}.M[n' \mapsto \oc{v}]
      \\\\ \oname{writeable}(\oc{\Phi}, n')
      \\ n' = k(n)
      \\ p_{e'} = \trusted \Longrightarrow n' \notin H_{\untrusted}
    }
    {\currentcom{\oc{\Phi}}{\cstore{k}{e}{e'}}{\untrusted} \ostep \pcinc{\oc{\Phi}}[M \assign M']}
  \end{mathpar}
  \end{small}
  \caption{\olangname Operational Semantics Excerpt}
  \label{fig:overlay:operational-excerpt}
\end{figure}

\begin{figure}[t]
  \begin{small}
  \begin{mathpar}
    \inferrule
    {
      \oc{F} = \oc{\Phi}.\mathit{funcs}(\mathit{target})
      \\ \oc{F}.\mathit{entry} = \mathit{target}
      \\ sp \in S_p
      \\\\ [\oc{\SF}] \concat \_ = \oc{\Phi}.\mathit{stack}
      \\ sp \geq \oc{\SF}.\mathit{ret\mbox{-}addr} + \oc{F}.\mathit{type}
    }
    {\oname{typechecks}(\oc{\Phi}, \mathit{target}, sp)}

    \inferrule
    {
      [\oc{\SF}] \concat \_ = \oc{\Phi}.\mathit{stack}
      \\\\ \mathit{ret\mbox{-}addr} = \oc{\SF}.\mathit{ret\mbox{-}addr}
    }
    {\oname{is-ret-addr}(\oc{\Phi}, \mathit{ret\mbox{-}addr})}

    \inferrule
    {
      \oc{F} \in \cod{\oc{\Phi}.\mathit{funcs}}
      \\\\ n, n' \in \oc{F}.\mathit{instrs}
    }
    {\oname{in-same-func}(\oc{\Phi}, n, n')}

    \inferrule
    {
      [\oc{\SF}] \concat \_ = \oc{\Phi}.\mathit{stack}
      \\\\ \forall (r, n) \in \oc{\SF}.\mathit{csr\mbox{-}vals}.~ \oc{\Phi}.R(r) = n
    }
    {\oname{csr-restored}(\oc{\Phi})}

    \inferrule
    {
      [\oc{\SF}] \concat \_ = \oc{\Phi}.\mathit{stack}
      \\\\ n \in S_p \Longrightarrow
      n \geq \oc{\SF}.\mathit{base} \wedge n \neq \oc{\SF}.\mathit{ret\mbox{-}addr}
    }
    {\oname{writeable}(\oc{\Phi}, n)}
  \end{mathpar}
  \end{small}
  \caption{\olangname Semantics Auxiliary Predicates}
  \label{fig:overlay:aux-preds}
\end{figure}

\olangname{} enforces our zero-cost conditions by extending the operational
semantics of \langname{} with additional checks in the overlay's small
step operational semantics, written $\oc{\Phi} \ostep \oc{\Phi'}$.
Each of these steps is a refinement of the underlying \langname{} step, that is
$\oc{\Phi}.\Psi \step \oc{\Phi'}.\Psi$ whenever $\oc{\Phi'}$ is not $\oerror$.
\figref{fig:overlay:operational-excerpt} (with auxiliary definitions shown in
\figref{fig:overlay:aux-preds}) shows an excerpt of the checks, which we
describe below.
Full definitions can be found in \iftechreport{\appref{appendix:overlay}}{the
technical appendix~\cite{kolosick2021isolation}}.
The checks are similar in nature to the defensive semantics of \citet{sfi-as-ai}
though they account for confidentiality and define a more flexible notion of
protecting stack frames.

\para{Call}
In the overlay, the reduction rule for library $\mathtt{call}$ instructions
(\explainredOcall{}) checks type-safe execution with $\oname{typechecks}$, a
predicate over the state ($\oc{\Phi}$), call target ($\mathit{target}$), and
stack pointer ($sp$) that checks that
\begin{enumerate*}
\item the address we are jumping to is the entry instruction of one of the
functions,

\item the stack pointer remains within the stack ($sp \in S_p$), and

\item the number of arguments expected by the callee have been pushed
to the stack.
\end{enumerate*}
On top of this, $\mathtt{call}$ also creates a new logical stack frame
recording the base of the new frame, location of the return address, and the
current callee-save register values, pushing the new frame onto the overlay
stack.
To ensure IFC, we require that $i$ has the label $\untrusted$ to ensure that
control flow is not influenced by confidential values; a similar check is done
when jumping within library code, obviating the need for a program counter
label.
Further, because the overlay captures zero-cost transitions, $\mathtt{gatecall}$
behaves in the exact same way except for an additional IFC check that the
arguments are not influenced by confidential values.

\para{Jmp}
Our zero-cost conditions rely on preventing invariants internal to a function
from being interfered with by other functions.
A key protection enabling this is illustrated by the reduction
for $\mathtt{jmp}$ (\explainredOjmp), which enforces that the
only inter-function control flow is via $\mathtt{call}$
and $\mathtt{ret}$: the $\oname{in-same-func}$ predicate
checks that the current ($n$) and target ($n'$)
instructions are within the same overlay function.
The same check is added to the program counter increment operation,
$\pcinc{\oc{\Phi}}$.
These checks ensure that the logical call stack corresponds to the actual
control flow of the program, enabling the overlay stack's use in maintaining
invariants at the level of function calls.

\para{Store}
The reduction rule for $\mathtt{store}$ (\explainredOstore) demonstrates the
other key protection enabling function local reasoning, with the check that the
address ($n$) is $\oname{writeable}$ given the current state of the overlay stack.
The predicate $\oname{writeable}$ guarantees that, if the operation is writing
to the stack, then that write must be within the current frame and cannot be the
location of the stored return address.
This allows reasoning to be localized to each function: they do not need to
worry about their callees tampering with their local variables.
Protecting the stored return address is crucial for ensuring well-bracketing,
which guarantees that each function returns to its caller.

To guarantee IFC, \redOstore{} first requires that the pointer have the label
$\untrusted$, ensuring that the location we write to is not based on confidential
data.
Second, the check $p_{i'} = \trusted \Longrightarrow n' \notin H_{\untrusted}$
enforces that confidential values cannot be written to the library heap.
Similar checks, based on standard IFC techniques, are implemented for all other
instructions.

\para{Ret}
With control flow checks and memory write checks in place, we guarantee that,
when we reach a $\mathtt{ret}$ instruction, the logical call frame will
correspond to the ``actual'' call frame.
$\mathtt{ret}$ is then responsible for guaranteeing well-bracketing and ensuring
callee-save registers are restored.
This is handled by two extra conditions on $\mathtt{ret}$ instructions:
$\oname{is-ret-addr}$ and $\oname{csr-restored}$.
$\oname{csr-restored}$ checks that callee-save registers have been
properly restored by comparing against the values that were saved in the
logical stack frame by $\mathtt{call}$.
$\oname{is-ret-addr}$ checks that the value pointed to by the stack pointer
($\mathit{ret\mbox{-}addr}$) corresponds to the location of the return address
saved in the logical stack frame.
Memory writes were checked to enforce that the return address cannot be
overwritten, so this guarantees the function will return to the expected program
location.

\subsection{Overlay Semantics Enforce Security}
\label{sec:overlay:secure}

The goal of the overlay semantics and our zero-cost conditions is to capture the
essential behavior necessary to ensure that individual, well-behaved library
functions can be composed together into a sandboxed library call that enforces
SFI integrity and confidentiality properties.
Thus, library code that is well-behaved under the dynamic overlay type system
will behave equivalently to library code with springboard and trampoline
wrappers, and therefore well-behaved library code can safely elide those
wrappers and their overhead.
We prove that the overlay semantics is sound with respect to each of our
security properties:
\begin{theorem}[Overlay Integrity Soundness] \label{thm:overlay-integrity-soundness}
  If $\oc{\Phi_0} \in \programs$, $\oc{\Phi_0} \ostepn{n} \oc{\Phi_1}$,
  $\currentcom{\oc{\Phi_1}}{\_}{\trusted}$, and
  $\oc{\Phi_1} \ostepstar \oc{\Phi_2}$ such that $\oc{\Phi_1}.\Psi \stepwb
  \oc{\Phi_2}.\Psi$ with $\pi = \oc{\Phi_0}.\Psi \stepn{n} \oc{\Phi_1}.\Psi$, then
  \begin{enumerate*}
  \item $\mathcal{CSR}(\pi, \oc{\Phi_1}.\Psi, \oc{\Phi_2}.\Psi)$ and
  \item $\mathcal{RA}(\pi, \oc{\Phi_1}.\Psi, \oc{\Phi_2}.\Psi)$.
  \end{enumerate*}
\end{theorem}

\begin{theorem}[Overlay Confidentiality Soundness] \label{thm:overlay-confidentiality-soundness}
  If $\oc{\Phi_0} \in \programs$, $\currentcom{\oc{\Phi_1}}{\_}{\untrusted}$,
  $\currentcom{\oc{\Phi_3}}{\_}{\trusted}$, $\oc{\Phi_0}.\Psi \stepstar
  \oc{\Phi_1}.\Psi \steplown{n} \oc{\Phi_2}.\Psi \step \oc{\Phi_3}.\Psi$,
  $\oc{\Phi_1} \ostepn{n + 1} \oc{\Phi_3}$, and $\oc{\Phi_1} =_{\untrusted}
  \oc{\Phi_1'}$,
  then $\oc{\Phi_1'}.\Psi \steplown{n} \oc{\Phi_2'}.\Psi \step \oc{\Phi_3'}.\Psi$,
  $\oc{\Phi_1'} \ostepn{n + 1} \oc{\Phi_3'}$
  $\currentcom{\oc{\Phi_3'}}{\_}{\trusted}$, $\oc{\Phi_3}.pc = \oc{\Phi_3'}.pc$,
  and
  \begin{enumerate*}
  \item $\currentop{\oc{\Phi_2}}{\cgatecall{n'}{e}}$, $\currentop{\oc{\Phi_2'}}{\cgatecall{n'}{e}}$, and $\oc{\Phi_3} =_{\mathtt{call}\ n'} \oc{\Phi_3'}$ or
  \item $\currentop{\oc{\Phi_2}}{\cgateret}$, $\currentop{\oc{\Phi_2'}}{\cgateret}$, and $\oc{\Phi_3} =_{\mathtt{ret}} \oc{\Phi_3'}$.
  \end{enumerate*}
\end{theorem}
\section{Instantiating Zero-Cost}
\label{sec:implementation-security}

We describe two isolation systems that securely support zero-cost transitions:
they meet the overlay monitor zero-cost conditions.
The first (\secref{sec:web-assembly-secure}) is an SFI system using
WebAssembly as an IR before compiling to native code using the Lucet
toolchain~\cite{lucet}.
Here we rely on the language-level invariants of Wasm to satisfy our
zero-cost requirements.
To ensure that these invariants are maintained, in
\sectionref{sec:wasm-verifier} we describe a verifier, \verifname, that checks
that compiled binaries meet the zero-cost conditions.
In \sectionref{sec:wasm-proof} we outline our proof that the verifier guarantees
that compiled Wasm can safely elide springboards and trampolines.

The second system, \trsegmentsfi, is our novel SFI system combining
the x86 segmented memory model for memory isolation with several
security-hardening LLVM compiler passes to enforce our zero-cost conditions.
While WebAssembly meets the zero-cost conditions, it imposes additional
restrictions that lead to unrelated slowdowns.
\trsegmentsfi thus serves as a platform for evaluating the potential cost of
enforcing the zero-cost conditions directly as well as a proof-of-concept
SFI implementation designed using the zero-cost framework.

\subsection{WebAssembly}
\label{sec:web-assembly-secure}

WebAssembly (Wasm) is a low-level bytecode with a sound, static type system.
Wasm's abstract state includes global variables 
and heap memory, which are zero-initialized at start-up.
All heap accesses are explicitly bounds checked, 
meaning that compiled Wasm programs inherently 
implement heap isolation.
Beyond this, Wasm programs enjoy several language-level properties, which ensure
compiled binaries satisfying the zero-cost conditions.
We describe these below.

\para{Control flow}
There are no arbitrary jump instructions in Wasm, only structured intra-function
control flow.
Functions may only be entered through a call instruction, and may only be exited
by executing a return instruction.
Functions also have an associated type; direct calls are type-checked at
compile time while indirect calls are subject to a runtime type check.
This ensures that compiled Wasm meets our type-directed forward-edge CFI condition.

\para{Protecting the stack}
A Wasm function's type precisely describes the space required to allocate the
function's stack frame (including spilled registers).
All accesses to local variables and arguments are performed through statically
known offsets from the current stack base.
It is therefore impossible for a Wasm 
operation to access other stack frames or alter the 
saved return address.
This ensures that compiled Wasm meets our local state encapsulation condition,
and, in combination with type-checking function calls, guarantees that Wasm's
control-flow is well-bracketed.
We therefore know that compiled Wasm functions will
always execute the register-saving preamble and, upon
termination, will execute the register-restoring epilogue.
Further, the function body will not alter the values of any registers saved to
the stack, thereby ensuring restoration of callee-save registers.

\para{Confidentiality}
Wasm code may store values into function-local variables or a function-local
``value stack'' similar to that of the Java Virtual Machine~\cite{jvm}.
The Wasm spec requires that compilers initialize function-local variables either
with a function argument or with a default value.
Further, accesses to the Wasm value stack are governed by a coarse-grained
data-flow type system, with explicit annotations at control flow joins.
These are used to check at compile-time that an instruction cannot pop a value
from the stack unless a corresponding value was pushed earlier in the same
function.
This guarantees that local variable and value stack accesses can be compiled to
register accesses or accesses to a statically-known offset in the stack frame.

When executing a compiled Wasm function without heavyweight transitions,
confidential values from prior computations may linger in these spilled registers or
parts of the stack.
However, the above checks ensure that these locations will only be read if they
have been previously overwritten during execution of the same function by a
low-confidentiality Wasm library value.

\subsection{SegmentZero32}
\label{sec:segments-secure}

To demonstrate that zero-cost conditions can be applied outside of highly structured languages such as Wasm, we demonstrate their enforcement in our novel SFI system for C code called \trsegmentsfi.
As we mention in \secref{subsec:overview-zero}, our zero-cost conditions amalgamate a number of individual conditions which separately have well-studied enforcement mechanisms, and so we are able to compose a series of off-the-shelf Clang/LLVM security-hardening passes to form the core of \trsegmentsfi.
The memory bounds checks are performed using the x86 segmented memory model~\cite{intel-manual} (Similar to NaCl~\cite{yee_native_2009}, however we use an additional segment to separate the sandboxed heap and stack).

Since \trsegmentsfi directly enforces the structure required for zero-cost transitions on C code (rather than relying on Wasm as an IR), it allows us to investigate the intrinsic cost of enforcing zero-cost (See Section \ref{subsec:eval-zerocostsfi}), without suffering from irrelevant Wasm overheads.
We additionally compare \trsegmentsfi against NaCl's 32-bit SFI scheme for the
x86 architecture, which we believe is the fastest production-quality SFI
toolchain currently available.
Below we discuss specific details \trsegmentsfi zero-cost condition enforcement.

\para{Protecting the stack}
We apply the SafeStack~\cite{kuznetsov_code-pointer_2014, safestack-llvm} compiler pass to further split the sandboxed stack into a safe and unsafe stack.
The safe stack contains only data that the compiler can statically verify is
always accessed safely, e.g., return addresses, spilled registers, and
allocations that are only accessed locally using verifiably safe offsets within
the function that allocates them.\footnote{We also use LLVM's
stack-heap clash detection (\textsf{-fstack-clash-protection}) to prevent
the stack growing into the heap.}
All other stack values are moved to the heap segment.
This ensures that pointer manipulation of unsafe stack references cannot be used to corrupt the return address and saved context of the current call.
We write a small LLVM pass to add additional support for tracking whether an access must be made through the heap segment or the stack segment, ensuring correct code generation.

These transformations ensure that malicious code cannot programmatically access
anything stored in the stack segment, except through offsets statically
determined to be safe by the SafeStack pass.
This protects the stored callee-save registers and return address, guaranteeing
the restoration of callee-save registers and well-bracketing \emph{iff forward
control flow is enforced}.

\para{Control flow}
Fortunately, enforcing forward-edge CFI has been widely studied~\cite{burow_control-flow_2017}.
We use a CFI pass as implemented in Clang/LLVM~\cite{cfi-llvm,
DBLP:conf/uss/TiceRCCELP14} including flags to dynamically protect indirect
function calls, ensuring forward control flow integrity.
Further, \trsegmentsfi conservatively bans non-local
control flow (e.g. \texttt{setjmp/longjmp}) in the C source code.
A more permissive approach is possible, but we leave this for future work.

\para{Confidentiality}
To guarantee confidentiality we implement a small change in Clang to zero initialize all stack variables.\footnote{We can't use Clang's existing pass
for variable initialization~\cite{stack-var-init-llvm} as it zero initializes data on the unsafe stack leading to poor performance}
This ensures that scratch registers cannot leak secrets as all sandbox values
are semantically written before use.
In practice, many of these writes are statically known to be dead and therefore optimised away.

\section{Verifying compiled WebAssembly}
\label{sec:wasm-verifier}

Instead of trusting the Wasm compiler, we build a \emph{zero-cost verifier},
\verifname{}, to check that the native, compiled output meets the zero-cost
conditions and is thus safe to run without springboards and trampolines.
\verifname{} is a static x86 assembly analyzer that takes as input potentially
untrusted native programs and verifies a series of local properties via abstract
interpretation.
Together these local properties guarantee that the monitor checks defined in
\olangname{} are met; we discuss the proof of soundness in
\sectionref{sec:wasm-proof}.

\verifname{} extends the VeriWasm SFI verifier~\cite{veriwasm}.
Both operate over WebAssembly modules compiled by the Lucet Wasm
compiler~\cite{lucet}, first disassembling the native x86 code before computing
a control-flow graph (CFG) for each function in the binary.
The disassembled code is then lifted to a subset of \langname{}, which serves as
the first abstract domain in our analysis.
Unfortunately, the properties checked by VeriWasm, while sufficient to guarantee
SFI security, are insufficient to guarantee zero-cost security.
Below we will describe how \verifname{} extends VeriWasm to guarantee the
stronger zero-cost conditions are met.

\begin{figure}[t]
\begin{lstlisting}[style=asm, escapeinside=||, morekeywords={bad_func:, good_func:}]
bad_func: [] >> rax                           good_func: [rdi] >> rax
  push r12                                      mov rax <- rdi |\label{line:verifier:sample:store}|
  ; TRACK: stack[0] = initial r12 value         ret
  mov r12 <- 1                                                 |\label{line:verifier:sample:init}|
  ; TRACK: r12 initialized
  mov r11 <- r13 + r12                                         |\label{line:verifier:sample:uninit}|
  ; TRACK: r11 uninitialized
  mov rdi <- 2                                                 |\label{line:verifier:sample:arginit}|
  ; TRACK: rdi initialized
  ; ASSERT: good_func arguments initialized
  call good_func                                               |\label{line:verifier:sample:call}|
  ; TRACK: good_func return value initialized
  pop r12                                                      |\label{line:verifier:sample:restore}|
  ; TRACK: r12 = initial r12 value
  ; ASSERT: callee-save registers restored
  gateret                                                      |\label{line:verifier:sample:ret}|
\end{lstlisting}
\caption{Disassembled and lifted WebAssembly functions}
\label{fig:verifier:sample}
\end{figure}

\para{The \verifname{} analyzers}
\verifname{} adds two new analyses to VeriWasm.
The first extends VeriWasm's CFI analysis, which only captures coarse
grained control-flow (i.e., that all calls target valid sandboxed functions),
to also extract type information.
Extracting type information from the binary code is possible without any complex
type inference because Lucet leaves the type signatures in the compiled output
(though we do not need to trust Lucet to get these type signatures correct 
since \verifname{} would catch any deviations at the binary level).
For direct calls, \verifname{} simply extracts the WebAssembly type stored in the
binary.
For indirect calls we extend the VeriWasm indirect call analysis to track the
type of each indirect call table entry, enabling us to resolve each indirect call to a
statically known type.
These types correspond to the input registers and stack slots, and the output registers
(if any) used by a function.
For example, in \figref{fig:verifier:sample} \Asminline{bad_func} takes no input
and outputs to \Asminline{rax} and \Asminline{good_func} takes \Asminline{rdi}
as input and outputs to \Asminline{rax}.

The second analysis tracks dataflow in local variables, i.e., in registers and stack slots.
Continuing with \Asminline{bad_func} as our example this analysis captures that:
in \coderef{line:verifier:sample:store} stack slot 0 now holds the initial value
of \Asminline{r12}, in \coderef{line:verifier:sample:init} \Asminline{r12}
holds an initialized (and therefore public) value, in
\coderef{line:verifier:sample:uninit} \Asminline{r13} has not been
initialized and therefore potentially contains confidential data so
\Asminline{r11} may also contain confidential data, etc.
This analysis is used to check confidentiality, callee-save register
restoration, local state encapsulation, and is combined with the previous
analysis to check type-directed CFI.

\para{The dataflow abstract domain}
To track local variable dataflow, \verifname{} uses an abstract domain with
three elements: $\vcuninit$ which represents an uninitialized, potentially
confidential value; $\vcinit$ which represents an initialized, public value; and
$\vccallee{r}$ which represents a potentially confidential value which
corresponds to the original value of the callee-save register $r$.
The domain forms a meet-semilattice with $\vcuninit$ the least element and all
other elements incomparable.

From here, analysis is straightforward, with a function's argument registers and
stack slots initialized to $\vcinit$, each callee-save register $r$ initialized
to $\vccallee{r}$, and everything else $\vcuninit$.
Instructions are interpreted as expected, e.g., $\cinst{mov}$ simply copies the
abstract value of its source into the target, operations return the meet of
their operands, and all constants and reads from the heap are treated as
initialized.
Across calls we assume that callee-save register conventions are followed (as we
will be checking this), preserving the value of all callee-save registers and
clearing all other registers' values.
We extract the type information from the extended CFI analysis to
determine the return register that is initialized after a function call.

\para{Checking the zero-cost conditions}
The above two analyses, along with additional information from VeriWasm's existing analyses enable us to check the zero-cost conditions.
\begin{CompactEnumerate}
\item

  \emph{Callee-save register restoration:}
  The $\vccallee{r}$ value enables straightforward checking that callee-save
  registers have been restored by checking that, at each $\cinst{ret}$
  instruction, each callee-save register $r$ has the abstract value
  $\vccallee{r}$.

\item

  \emph{Well-bracketed control-flow:}
  VeriWasm already implements a stack checker that guarantees that all writes to
  the stack are to local variables, ensuring that the saved return address on the stack
  cannot be tampered with.
  Further, it checks that the stack pointer is restored to its original
  location at the end of every function, ensuring the saved return address is used.

\item

  \emph{Type-directed forward-edge CFI:}
  The dataflow analysis gives us the registers that are initialized when we
  reach a \cinst{call} instruction, enabling us to check that the input arguments
  of the target have been initialized.
  For example, when we reach \coderef{line:verifier:sample:call} we know that
  \Asminline{rdi} has the value $\vcinit$.
  The type-based CFI analysis tells us that \Asminline{good_func}
  expects \Asminline{rdi} as an input, so this call is marked as safe.

\item

  \emph{Local state encapsulation:}
  To ensure SFI security, VeriWasm checks that no writes are
  below the current stack frame, ensuring that verified Wasm functions cannot
  tamper with other frames.

\item

  \emph{Confidentiality:}
  We check confidentiality using the information obtained in our dataflow
  analysis, where the value $\vcinit$ ensures that a value is initialized with a
  public, non-confidential value.
  This enables us to check each of the confidentiality checks encoded in
  \olangname{} are met: for instance the type-safe forward-edge CFI check
  described above already ensures each argument is initialized.
  In \figref{fig:verifier:sample}, the confidentiality checker will flag
  \coderef{line:verifier:sample:uninit} as unsafe because \Asminline{r13} still has
  the value $\vccallee{\mathtt{r13}}$, which potentially contains confidential information
  leaked from the application.
\end{CompactEnumerate}

\subsection{Proving Wasm secure}
\label{sec:wasm-proof}
We prove that compiled and verified Wasm libraries can safely elide springboards and
trampolines while maintaining integrity and confidentiality, by showing that the
verified code would not violate the safety monitor.
Formally, this amounts to showing that Wasm code verified by \verifname{} never
reaches an $\oerror$ state.
This allows us to apply \thmref{thm:overlay-integrity-soundness} and
\thmref{thm:overlay-confidentiality-soundness}.
It is relatively straightforward (with one exception) to prove that the abstract
interpretation as described guarantees the necessary safety conditions.

The crucial exception in the soundness proof is when a function calls to other Wasm
functions.
We must inductively assume that the called function is safe, i.e.,
doesn't change any variables in our stack frame, restores callee-save registers,
etc.
Unfortunately, a naive attempt does not lead to an inductively well-founded
argument.
Instead, we use the overlay monitor's notion of a well-behaved function to
define a step-indexed logical relation (detailed in
\iftechreport{\appref{appendix:lr}}{the technical
appendix~\cite{kolosick2021isolation}}) that captures a semantic notion of
well-behaved functions (as a relation $\Frel$), and then lift this to a relation
over an entire Wasm library (as a relation $\Lrel$).
This gives a basis for an inductively well-founded argument where we can prove
that, locally, the abstract interpretation gives that each Wasm function is
semantically well-behaved (is in $\Frel$) and then use this to prove the
standard fundamental theorem of a logical relation for a whole Wasm library:
\begin{theorem}[Fundamental Theorem for Wasm Libraries] \label{thm:wasm-in-lrel}
  For any number of steps $n \in \nats$ and compiled Wasm library $L$,
  $(n, L) \in \Lrel$.
\end{theorem}
\noindent
This theorem states that every function in a compiled Wasm library,
when making calls to other Wasm functions or application callbacks, is
well-behaved with respect to the zero-cost conditions.
The number of steps is a technical detail related to step-indexing.
Zero-cost security then follows by adequacy of the logical relation,
\thmref{thm:overlay-integrity-soundness}, and
\thmref{thm:overlay-confidentiality-soundness}:
\begin{theorem}[Adequacy of Wasm Logical Relation] \label{thm:lrel-adequacy}
  For any number of steps $n \in \nats$, library $L$ such that $(n, L) \in
  \Lrel$, program $\oc{\Phi_0} \in \programs$ using $L$, and $n' \leq n$, if
  $\oc{\Phi_0} \ostepn{n'} \oc{\Phi'}$ then $\oc{\Phi'} \neq \oerror$.
\end{theorem}
\noindent
Details of the logical relation and proofs are in
\iftechreport{\appref{appendix:webassembly}}{the technical
appendix~\cite{kolosick2021isolation}}.

\section{Evaluation}
\label{sec:eval}
We evaluate our zero-cost model by asking four questions:
\begin{CompactItemize}
\item \textbf{Q1}: What is the cost of a context switch? (\secref{subsec:eval-transitions})
\item \textbf{Q2}: What is end-to-end performance gain of Wasm-based SFI due to zero-cost transitions? (\secref{subsec:eval-wasm})
\item \textbf{Q3}: What is the performance overhead of purpose-built zero-cost SFI enforcement? (\secref{subsec:eval-zerocostsfi})
\item \textbf{Q4}: Is the \verifname verifier effective? (\secref{subsec:verifier-eval})
\end{CompactItemize}

Since our zero-cost condition enforcement does incur some runtime overhead, \textbf{Q2} and \textbf{Q3} are heavily workload-dependent.
The benefit a workload receives from the zero-cost approach will be in direct proportion to the frequency with which it performs domain transitions.

\para{Systems}
To investigate the first three questions, we consider two groups of SFI systems.
The first group compares a number of different transition models for Wasm-based SFI for 64-bit binaries, built on top of the Lucet compiler~\cite{lucet}.
All of these will have identical runtime overhead, meaning that the only variance between them will be due to transition overhead.
The \trlucet build uses the original heavyweight springboards and trampolines
shipped with the Lucet runtime written in Rust.
\trfullswitch adopts techniques from NaCl and uses optimized
assembly to save and restore application context during transitions.
\trfast implements our zero-cost transition system, meaning transitions are
simple function calls.
To understand the overhead of register saving/restoring and stack
switching, we also evaluate a \trregsave build which saves/restores registers
like \trfullswitch, but shares the library and application stack like
\trfast.

The second group compares optimized SFI techniques for 32-bit binaries.
Wasm-based SFI imposes overheads far beyond what is strictly necessary to
enforce our zero-cost conditions, both because of the immaturity of the Lucet
compiler in comparison to more established compilers such as Clang, and because
Wasm inherently enforces additional restrictions on compiled code (e.g.,
structured intra-function control flow).
We design \trsegmentsfi~(\secref{sec:segments-secure}) to enforce only our zero-cost-conditions and nothing more, aiming to benchmark it against the Native Client 32-bit isolation scheme (\trnacl)~\cite{yee_native_2009}, arguably the fastest production SFI system available, which requires heavyweight transitions.
Both systems make use of memory segmentation, a 32-bit x86-only feature for fast memory isolation.
Unfortunately, we cannot make a uniform comparison between \trnacl, \trsegmentsfi, and \trfast since Lucet only supports a 64-bit target.

Each group additionally uses unsandboxed, insecure native execution (\texttt{Vanilla}) as a baseline.
To represent the best possible performance of schemes relying on heavyweight
transitions, we also benchmark \tridealheavy and \tridealheavysixfour,
ideal hardware isolation schemes, which incur no runtime
overhead but require heavyweight transitions.
To simulate the performance of these ideal schemes, we simply measure the performance
of native code with heavyweight trampolines.

We integrate all of the above SFI schemes into Firefox using the RLBox
framework~\cite{rlbox}.
Since RLBox already provides plugins for the \trlucet and \trnacl builds, we
only implement the plugins for the remaining system builds.

\para{Benchmarks}
We use a micro-benchmark to evaluate the cost of a single transition for our
different transition models, using unsandboxed native calls as a baseline
(\textbf{Q1}).

We answer questions \textbf{Q2}--\textbf{Q3} by measuring the end-to-end
performance of font and image rendering in Firefox, using a sandboxed
\libgraphite and \libjpeg, respectively.
We use these libraries because they have many cross-sandbox transitions, which
\citet{rlbox} previously observed to affect the overall browser performance.
To evaluate the performance of \libgraphite, we use Kew's
benchmark\footnote{Available at
\url{https://jfkthame.github.io/test/udhr_urd.html}}, which reflows the text on
a page ten times, adjusting the size of the fonts each time to negate the
effects of font caches.
When calling \libgraphite, Firefox makes a number of calls into the sandbox
roughly proportional to the number of glyphs on the page.
We run this benchmark 100 times and report the median execution
time below (all values have standard deviations within 1\%).

To evaluate the performance of \libjpeg, we measure the overhead of rendering
images of varying complexity and size.
Since the work done by the sandboxed \libjpeg, per call, is proportional to the
width of the image\dash---Firefox executes the library in \emph{streaming
mode}, one row at a time\dash---we consider images of different widths,
keeping the image height fixed.
This allows us to understand the benefits and limitations of zero-cost
transitions, since the proportion of execution time spent context-switching decreases
as the image width increases.
We do this for three images, of varying complexity: a simple image consisting
of a single color (\simplejpeg), a stock image from the Image Compression
benchmark suite\footnote{Online:
\url{https://imagecompression.info/test_images/}.  Visited Dec 9, 2020.}
(\stockjpeg), and an image of random pixels (\randomjpeg).
We render each image 500 times and report the median time (standard
deviations are all under 1\%).

Finally, we use \SPECOhSix to partly evaluate the sandboxing overhead of our
purpose-built \trsegmentsfi SFI system (\textbf{Q3}), and to measure
\verifname's verification speed (\textbf{Q4}).

\para{Machine and software setup}
We run all but the verification benchmarks on an \Intel
Core\textsuperscript{TM} i7-6700K machine with four 4GHz cores, 64GB RAM,
running Ubuntu 20.04.1 LTS (kernel version 5.4.0-58).
We run benchmarks with a shielded isolated cpuset~\cite{cpu-shielding}
consisting of one core with hyperthreading disabled and the clock frequency
pinned to 2.2GHz.
We generate Wasm sandboxed code in two steps: First, we compile C/\C++
to Wasm using Clang-11, and then compile Wasm to native code using the 
fork of the Lucet used by RLBox (snapshot from Dec 9, 2020).
We generate NaCl sandboxed code using a modified version of Clang-4.
We compile all other C/\C++ source code, including \trsegmentsfi sandboxed code and
benchmarks using Clang-11.
We implement our Firefox benchmarks on top of Firefox Nightly (from August 22,
2020).

\para{Summary of results}
We find that the performance of Wasm-based isolation
can be significantly improved by adopting zero-cost transitions, but that
Lucet-compiled WebAssembly's runtime overhead means that it does not outperform
more optimised isolation schemes in end-to-end tests.
The low performance overhead of \trsegmentsfi demonstrates that these runtime
overheads are not inherent to the zero-cost approach, and that an optimised
zero-cost SFI system can significantly outperform more traditional schemes,
especially for workloads with a large number of transitions.
Finally, we find that we can efficiently check zero-cost conditions at the
binary level, for Lucet compiled code, with no false positives.

%%%%%%%%%%%%%%%%%%%%%%%%%%%%%%%%%%%%%%%%%%%%%%

\begin{figure}
\footnotesize

\begin{tabular}{p{2.5cm}|cccc}
    \toprule
    \textbf{Build}
  & \textbf{Direct call}
  & \textbf{Indirect call}
  & \textbf{Callback}
  & \textbf{Syscall}
  \\
  \toprule
  \trnative (in C) &
  1ns & 56ns & 56ns & 24ns
  \\
  \trlucet &
  --- & 1137ns & --- & ---
  \\
  \trfullswitch &
  120ns & 209ns & 172ns & 192ns
  \\
  \trregsave &
  120ns & 210ns & 172ns & 192ns
  \\
  \textbf{\trfast} &
  \bf 7ns & \bf 66ns & \bf 67ns & \bf 60ns
  \\
  \midrule
  \trnative (in C, 32-bit) &
  1ns & 74ns & 74ns & 37ns
  \\
  \trnacl &
  --- & 714ns & 373ns & 356ns
  \\
  \textbf{\trsegmentsfi} &
  \bf 24ns & \bf 108ns & \bf 80ns & \bf 88ns
  \\
  \bottomrule
\end{tabular}

\caption{
Costs of transitions in different isolation models.
Zero-cost transitions are shown in \textbf{boldface}.
\trnative is the performance of an unsandboxed C function call, to serve as a baseline.
}
\label{fig:transition-overheads}
\end{figure}

%%%%%%%%%%%%%%%%%%%%%%%%%%%%%%%%%%%%%%%%%%%%%%

\subsection{The cost of transitions}
\label{subsec:eval-transitions}
We measure the cost of different cross-domain transitions\dash---direct and
indirect calls into the sandbox, callbacks from the sandbox, and syscall
invocations from the sandbox\dash---for the different system builds
described above.
To expose overheads fully, we choose extremely fast payloads---either a
function that just adds two numbers or the \gettimeofday syscall,
which relies on Linux's vDSO to avoid CPU ring changes.
The results are shown in \figref{fig:transition-overheads}.
All numbers are averages of one million repetitions, and repeated runs have
negligible standard deviation.\footnote{
Lucet and NaCl don't support direct sandbox
calls; Lucet further does not support custom callbacks or syscall invocations.
}

We make several observations.
First, among Wasm-based SFI schemes, zero-cost transitions (\trfast) are
significantly faster than even optimized heavyweight transitions
(\trfullswitch).
Lucet's existing indirect calls written in Rust (\trlucet) are significantly
slower than both.
Second, the cost of stack switching (the difference of \trfullswitch and
\trregsave) is surprisingly negligible.
Third, the performance of \trnative and \trfast should be identical but is not.
This is \emph{not} because our transitions have a hidden cost. Rather, it's
because we are comparing code produced by two different compilers:
\trnative is native code produced by Clang,  while \trfast is code produced by
Lucet, and Lucet's code generation is not yet highly
optimized~\cite{cranelift-speedup}.
For example, in the benchmark that adds two numbers, Clang eliminates
the function prologue and epilogue that save and restore the frame
pointer, while Lucet does not.
We observe similar trends for hardware-based isolation.
For example, we find that \trsegmentsfi transitions are much faster than
\tridealheavy and \trnacl transitions and only \tranSegzeroNativeFuncDiff
slower than \trnative for direct calls.
Finally, we observe that \trsegmentsfi is slower than \trfast: hardware
isolation schemes like \trsegmentsfi and \trnacl execute instructions to enable
or disable the hardware based memory isolation in their transitions.

\subsection{End-to-end performance improvements of zero-cost transitions for Wasm}
\label{subsec:eval-wasm}

%%%%%%%%%%%%%%%%%%%%%%%%%%%%%%%%%%%%%%%%%%%%%%

\begin{figure*}
  
  \begin{subfigure}{0.32\textwidth}
    \includegraphics[width=4.5cm]{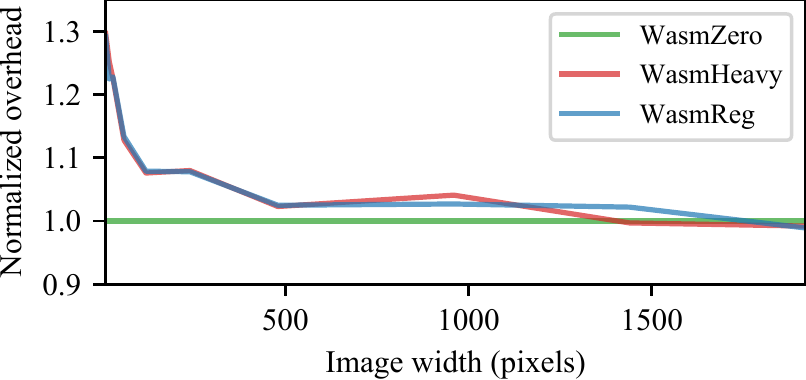}
    \caption{\simplejpeg}
    \label{fig:jpeg-simpleimg}
  \end{subfigure}
  \begin{subfigure}{0.32\textwidth}
    \includegraphics[width=4.5cm]{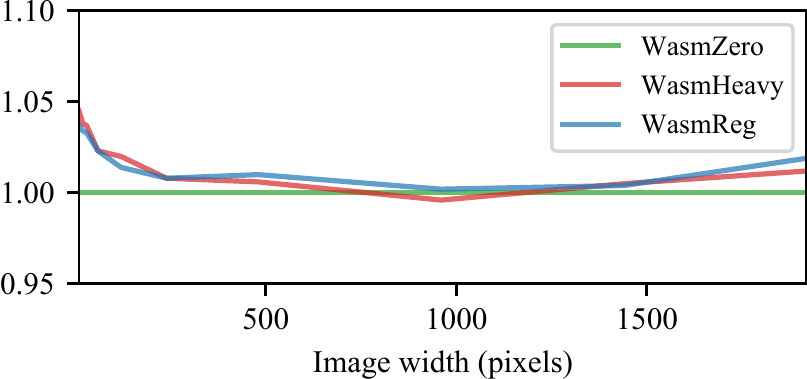}
    \caption{\stockjpeg}
    \label{fig:jpeg-stockimg}
  \end{subfigure}
  \begin{subfigure}{0.32\textwidth}
    \includegraphics[width=4.5cm]{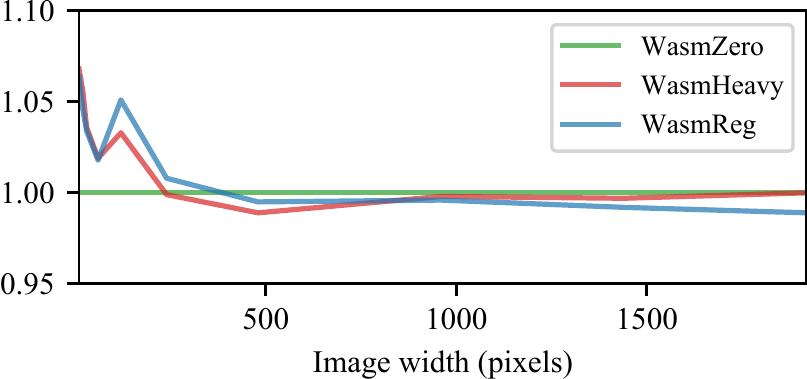}
    \caption{\randomjpeg}
    \label{fig:jpeg-randomimg}
  \end{subfigure}
  
  \caption{
    Performance of different Wasm transitions on rendering of (a)~a simple
    image with one color, (b)~a stock image, and (c)~a complex image with
    random pixels, normalized to \trfast.
    \trfast transitions outperform other transitions. The difference
    diminishes with width, but narrower images are more common on
    the web.
  }
  \label{fig:jpeg-img}
\end{figure*}

%%%%%%%%%%%%%%%%%%%%%%%%%%%%%%%%%%%%%%%%%%%%%%

We evaluate the end-to-end performance impact of the different transition
models on Wasm-sandboxed font and image rendering as used in Firefox (see
\secref{sec:eval}).

\para{Font rendering}
We report the performance of \libgraphite isolated with Wasm-based schemes on
Kew's benchmark below:
 
\begin{center}
\footnotesize

\begin{tabular}{p{1.65cm}|p{1.3cm}p{1.4cm}p{1cm}p{1.4cm}p{1cm}p{1.6cm}}
      \toprule
    & \textbf{\trlucet}
    & \textbf{\trfullswitch}
    & \textbf{\trregsave}
    & \textbf{\trfast}
    & \textbf{\trnative}
    & \textbf{\tridealheavysixfour}
    \\
\toprule
    \textbf{Font render}
    & 8173ms & 2246ms & 2230ms & 2032ms & 1116ms & 1563ms  \\
\bottomrule
\end{tabular}
\end{center}

\noindent
As expected, Wasm with zero-cost transitions (\trfast) outperforms the
other Wasm-based SFI transition models.
Compared to \trfast, Lucet's existing transitions slow down rendering 
by over \ffMaxFontSlowdownWasmLucetZero.\footnote{
  This overhead is smaller than the 8$\times$ overhead reported by
  \citet{rlbox}; we attribute this difference to the different compilers\dash---we
  use a more recent, and faster, version of Lucet.
}
But, even the optimized heavyweight transitions (\trfullswitch) impose a
\ffMaxFontSlowdownWasmHeavyZero performance tax.
This overhead is due to register saving/restoring; stack switching
only accounts for \ffMaxFontStackSwitchWasmOverhead overhead.

While these results show that existing Wasm-based isolation schemes can benefit
from switching to zero-cost transitions\dash---and indeed the speed-up due to
zero-cost transitions allowed Mozilla to ship the Wasm-sandboxed
\libgraphite\dash---they also show that Lucet-compiled Wasm is slow
($\sim$80$\%$ slower than Vanilla).
This, unfortunately, means that the transition cost savings alone are not
enough to beat \tridealheavysixfour, even for a workload with many transitions.
To compete with this ideal SFI scheme with heavyweight transitions, we would
need to reduce the runtime overhead to $\sim$40$\%$.
\citet{not-so-fast} report the average runtime overhead of Mozilla SpiderMonkey JIT-compiled WebAssembly compared to native as $\sim$45$\%$ in a different set of benchmarks, while noting many correctable inefficiencies in the generated assembly code, suggesting that there is a lot of room for Lucet to be further optimised.

\para{Image rendering}
\figref{fig:jpeg-img} report the overhead of Wasm-based sandboxing on
image rendering, normalized to \trfast to highlight the relative overheads
of different transitions as compared to our zero-cost transitions.
We report results of \trlucet separately, in \iftechreport{\appref{appendix:img}
(\figref{fig:jpeg-img-lucet})}{the technical
appendix~\cite{kolosick2021isolation}} because the rendering times are up to
\ffMaxImgSlowdownWasmLucetZero longer than the other builds.
Here, we instead focus on evaluating the overheads of optimized
heavy transitions.

As expected, \trfast significantly outperforms other transitions when images 
are narrower and simpler.
On \simplejpeg, \trfullswitch and \trlucet can take as much as 
\ffMaxImgSimpleSlowdownWasmHeavyZero and \ffMaxImgSimpleSlowdownWasmLucetZero
longer to render the image as with \trfast transitions.
However, this performance gap diminishes as image width increases (and the
relative cost of context switching decreases).
For \stockjpeg and \randomjpeg, the \trfullswitch trends are similar, but
the rendering time differences start at about 
\ffMaxImgStockRandomSlowdownWasmHeavyZero.
Lucet's existing transitions (\trlucet) are still significantly slower 
than zero-cost transitions (\trfast) even on wide images.

%%%%%%%%%%%%%%%%%%%%%%%%%%%%%%%%%%%%%%%%%%%%%%

\begin{figure}
  \includegraphics[width=6.5cm]{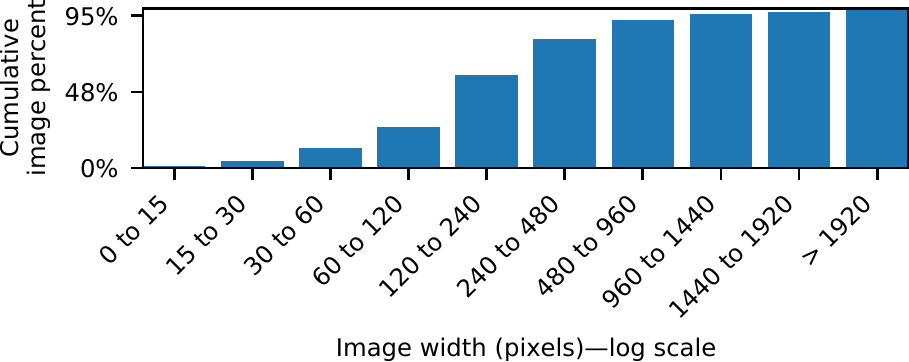}
  \caption{
    Cumulative distribution of image widths on the landing pages of the Alexa 
    top 500 websites.
    Over 80\% of the images have widths under 480 pixels.
    Narrower images have a higher transition rate, and thus higher relative
    overheads when using expensive transitions.
  }
  \label{fig:image-sizes}
\end{figure}

%%%%%%%%%%%%%%%%%%%%%%%%%%%%%%%%%%%%%%%%%%%%%%

Though the differences between the transitions are smaller as the image width
increases, many images on the Web are narrow.
\figref{fig:image-sizes} shows the distribution of images on the landing
pages of the Alexa top 500 websites. Of the 10.6K images, 8.6K (over 80\%) have
widths between 1 and 480 pixels, a range in which zero-cost transitions
noticeably outperform the other kinds of transitions.

Like font rendering, we measure the target runtime overhead Lucet should
achieve to beat \tridealheavysixfour end-to-end for rendering images.
We report our results in \iftechreport{\figref{fig:jpeg-img-ideal64} in
\appref{appendix:img}}{the technical appendix~\cite{kolosick2021isolation}}.
For the small simple image, we observe this to be 94\%\dash---this is approximately the overhead of Lucet that we see already today.
For the small stock image, we observe this to be 15\%\dash---this is much smaller than the overhead of Lucet today, but lower overheads have been demonstrated on some benchmarks by the prototype Wasm compiler of~\citet{sledge}.

\subsection{Performance overhead of purpose-built zero-cost SFI enforcement}

\label{subsec:eval-zerocostsfi}
%%%%%%%%%%%%%%%%%%%%%%%%%%%%%%%%%%%%%%%%%%%%%%

\begin{figure*}
  
  \begin{subfigure}{0.32\textwidth}
    \includegraphics[width=4.5cm]{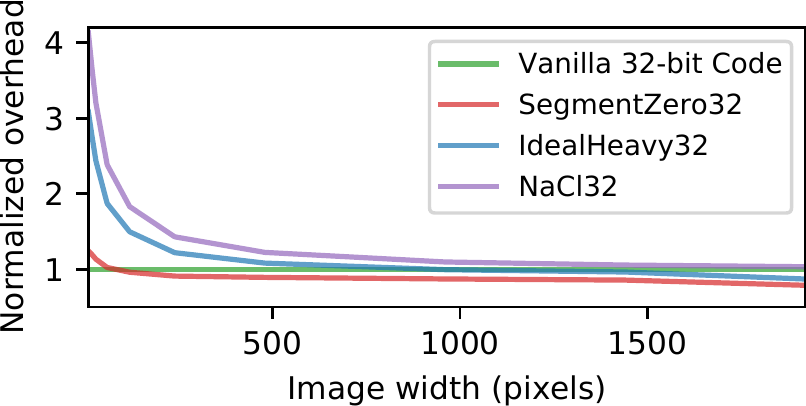}
    \caption{\simplejpeg}
    \label{fig:jpeg-simpleimg-hw}
  \end{subfigure}
  \begin{subfigure}{0.32\textwidth}
    \includegraphics[width=4.5cm]{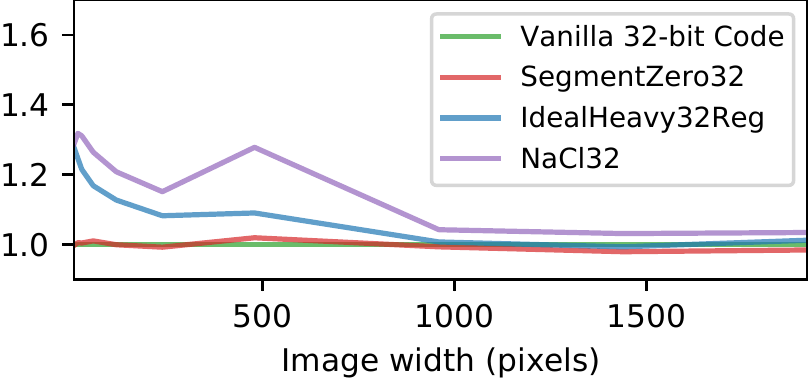}
    \caption{\stockjpeg}
    \label{fig:jpeg-stockimg-hw}
  \end{subfigure}
  \begin{subfigure}{0.32\textwidth}
    \includegraphics[width=4.5cm]{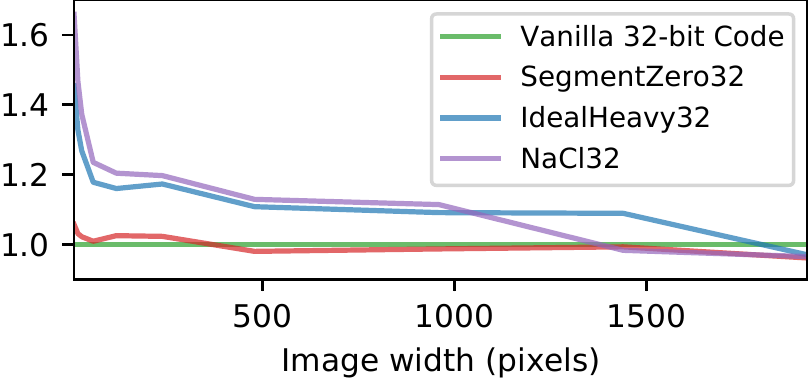}
    \caption{\randomjpeg}
    \label{fig:jpeg-randomimg-hw}
  \end{subfigure}
  
  \caption{
    Performance of image rendering with libjpeg sandboxed with
    \trsegmentsfi and \trnacl and \tridealheavy.
    Times are relative to unsandboxed code.
    \trnacl and \tridealheavy relative overheads are as high as 312\% and 208\% 
    respectively, while \trsegmentsfi relative overheads do not exceed 24\%.  }
  \label{fig:jpeg-img-hw}
  
\end{figure*}

%%%%%%%%%%%%%%%%%%%%%%%%%%%%%%%%%%%%%%%%%%%%%%

In this section, we measure the performance overhead of \trsegmentsfi with
zero-cost transitions.
We compare \trsegmentsfi with NaCl (\trnacl) and \tridealheavy --- a
hypothetical SFI scheme with no isolation enforcement overhead, both of which
rely on heavyweight transitions.
We measure the overhead of these systems on the standard \SPECOhSix benchmark
suite, and the \libgraphite and \libjpeg font and image rendering benchmarks.
Since both \trsegmentsfi and \trnacl use segmentation which is supported only 
in 32-bit mode, we implement these three isolation builds in 32-bit mode and 
compare it to native 32-bit unsandboxed code.
We describe these benchmarks next.

\para{SPEC}
We report the impact of sandboxing on \SPECOhSix in \figref{fig:spec}.
Several benchmarks are not compatible with \trnacl; augmenting \trnacl runtime
and libraries to fix such compatibility issues (e.g., as done
in~\cite{yee_native_2009} for SPEC2000) is beyond the scope of this paper.
The \texttt{gcc} benchmark, on the other hand, is not compatible with
\trsegmentsfi ---
\texttt{gcc} fails (at runtime) because the CFI used by
\trsegmentsfi\dash---Clang's CFI\dash---incorrectly computes a target CFI label.
Clang's CFI implementation is more precise than necessary for our zero-cost
conditions; as with \trnacl, we leave the implementation of a coarse-grain and
more permissive CFI to future work.
On the overlapping benchmarks, \trsegmentsfi's overhead is comparable to
\trnacl's.

\begin{figure*}
\vspace{-1em}
\footnotesize
\begin{center}
\begin{adjustbox}{width=1\textwidth}
\begin{tabular}{c|c|c|c|c|c|c|c|c|c|c|c|c}
  \toprule
  \textbf{\texttt{System}}
  & \texttt{400.perlbench}
  & \texttt{401.bzip2}
  & \texttt{403.gcc}
  & \texttt{429.mcf}
  & \texttt{445.gobmk}
  & \texttt{456.hmmer}
  & \texttt{458.sjeng}
  & \texttt{462.libquantum}
  & \texttt{464.h264ref}
  & \texttt{471.omnetpp}
  & \texttt{473.astar}
  & \texttt{483.xalancbmk}
\\
\toprule
\textbf{\trnacl} & --- & --- & 1.10$\times$ & --- & 1.27$\times$ & 0.97$\times$ & 1.20$\times$ & 1.06$\times$ & 1.34$\times$ & 1.06$\times$ & 1.31$\times$ & --- \\
\textbf{\trsegmentsfi} & 1.20$\times$ & 1.08$\times$ & --- & 1.04$\times$ & 1.25$\times$ & 0.82$\times$ & 1.16$\times$ & 1.02$\times$ & 1.01$\times$ & 1.01$\times$ & 1.10$\times$ & 1.05$\times$ \\
\bottomrule
\end{tabular}
\end{adjustbox}
\end{center}

\caption{Overheads compared to native code on \SPECOhSix (nc), for \trnacl and
\trsegmentsfi.}
\label{fig:spec}
\end{figure*}

\para{Font rendering}
The impact of these isolation schemes on font rendering is
shown below:

\begin{center}
\footnotesize

\begin{tabular}{p{1.65cm}|cccc}
      \toprule
    & \textbf{\texttt{\trnative (32-bit)}}
    & \textbf{\tridealheavy}
    & \textbf{\trnacl}
    & \textbf{\trsegmentsfi}
    \\
\toprule
    \textbf{Font render}
    & 1441ms & 2399ms & 2769ms & 1765ms \\
\bottomrule
\end{tabular}
\end{center}

\noindent
We observe that \trnacl and \tridealheavy impose an overhead of 
\ffMaxFontOverheadNaClNative and \ffMaxFontOverheadIdealNative 
respectively.
In contrast, \trsegmentsfi has a smaller overhead
(\ffMaxFontOverheadSegzeroNative) as it does not have to save and restore
registers or switch stacks.
We attribute the overhead of \trsegmentsfi (over Vanilla) to three
factors: (1) changing segments to enable/disable isolation 
during function calls, (2) using indirect function calls for cross-domain calls
(a choice that simplifies engineering but is not fundamental), and (3)
the structure imposed by our zero-cost condition enforcement.

\para{Image rendering}
We report the impact of sandboxing on image rendering in
\figref{fig:jpeg-img-hw}.
For narrow images (10 pixel width), \trsegmentsfi overheads relative to the
native unsandboxed code are \ffMaxImgSimpleOverheadSegzeroNative, 
\ffMaxImgStockOverheadSegzeroNative, and \ffMaxImgRandomOverheadSegzeroNative 
for \simplejpeg,
\stockjpeg and \randomjpeg, respectively.
These overheads are lower than the corresponding overheads for \trnacl
(\ffMaxImgSimpleOverheadNaClNative, 
\ffMaxImgStockOverheadNaClNative, and 
\ffMaxImgRandomOverheadNaClNative) as well as \tridealheavy 
(\ffMaxImgSimpleOverheadIdealNative, 
\ffMaxImgStockOverheadIdealNative and 
\ffMaxImgRandomOverheadIdealNative).
As in the Wasm measurements, these overheads shrink as image width increases
and the complexity of the image increases (e.g., the overheads for images wider
than 480 pixels are negligible).

\subsection{Effectiveness of the \verifname verifier}
\label{subsec:verifier-eval}
We evaluate \verifname's effectiveness by using it to (1) verify 13
binaries\dash---five third-party libraries shipping (or about to ship) with
Firefox compiled across 3 binaries, and 10 binaries from the \SPECOhSix
benchmarks\dash---and (2) find nine manually introduced bugs, inspired by real
calling convention bugs in previous SFI toolchains~\cite{cranelift-bug-1177,
nacl-bug-775, nacl-bug-2919}.
We measure \verifname's performance verifying the aforementioned 13 binaries.
Finally, we stress test \verifname by running it on random binaries generated
by Csmith~\cite{csmith}.

\para{Experimental setup}
We run all \verifname{} experiments on a 2.1GHz \Intel Xeon® Platinum 8160
machine with 96 cores and 1 TB of RAM running Arch Linux 5.11.12.
All experiments run on a single core and use no more than 2GB of RAM.
We compile the SPEC binaries used using the Lucet toolkit used in
\sectionref{subsec:eval-wasm}.
We verify three of the Firefox libraries from Firefox Nightly;  we compile the
other two from the patches that are in the process of being integrated into
Firefox.

\para{Effectiveness and performance results}
\verifname{} successfully verifies the 13 Firefox and \SPECOhSix binaries.
These binaries vary in size from 150 functions (the \texttt{lbm} benchmark from
\SPECOhSix) to 4094 functions (the binary consisting of the Firefox Nightly
libraries \libogg, \libgraphite, and \hunspell).
It took \verifname between 1.77 seconds and 19.28 seconds to verify these
binaries, with an average of 9.2 seconds and median of 5.93 seconds.
\verifname{}'s performance is on par with the original VeriWasm's performance:
on the 10 \SPECOhSix benchmarks evaluated in the VeriWasm
paper~\cite{veriwasm} \verifname{} is slightly (15\%) faster, despite
checking zero-cost conditions in addition to all of VeriWasm's original checks.
This is due to various engineering improvements that were made to VeriWasm in
the course of developing \verifname{}.

\verifname also successfully found bugs injected into nine binaries.
These bugs tested all the zero-cost properties that \verifname{} was
designed to check, and when possible they were based on real bugs (like those in
Cranelift~\cite{cranelift-bug-1177}).
\verifname{} successfully detected all nine of these bugs, giving us confidence
that \verifname{} is capable of finding violations of the zero-cost conditions.

\para{Fuzzing results}
We fuzzed \verifname{} to both search for potential bugs in Lucet, as well as
to ensure \verifname{} does not incorrectly declare safe programs unsafe.
The fuzzing pipeline works in four stages: first, we use Csmith~\cite{csmith} to
generate random C and \C++ programs, next we use Clang to compile the generated
C/\C++ program to WebAssembly, followed by compiling the Wasm file to native code
using Lucet, and finally we verify the generated binary with \verifname{}.
As these programs were compiled by Lucet, we expect them to adhere to the
zero-cost conditions, and any binaries flagged by \verifname{} are either bugs
in Lucet or are spurious errors in \verifname{}.

While we did not find bugs in Lucet, fuzzing did find cases where
\verifname{} triggered spurious errors.
After fixing these errors, we verified 100,000 randomly generated programs
with no false positives.

\section{Limitations}
\label{sec:limitations}
Our Wasm SFI scheme is capable of sandboxing any C/\C++-like language (with arbitrary
intra-function control flow, arbitrary type casting, arbitrary pointer aliasing,
arithmetic etc.) that can compile to Wasm, so long as it does not use
features which Wasm must emulate through JavaScript\footnote{See
\url{https://emscripten.org/}}\dash---most prominently \C++-style exceptions,
\texttt{setjmp/longjmp}, and multithreading.
These limitations are not inherent to our zero-cost conditions, and Wasm is in
the process of being extended with support for all of the above
features~\cite{10.1145/3360559,wasmeh}.

Our \trsegmentsfi scheme is built as a proof-of-concept, using mostly existing
LLVM passes to sandbox C programs compiled to 32-bit x86, as an approach to understanding the overhead of zero-cost conditions on native code.
As such, our \trsegmentsfi implementation does not support, for instance,
\texttt{setjmp/longjmp} or multithreading (similar to Wasm).
In addition it does not support user-defined variadic function arguments or position
independent code.
However, these are engineering limitations and not fundamental.
For example, user-defined variadic function arguments could be supported by extending the SafeStack LLVM pass to move the
variadic argument buffer into the unsafe stack, and position independent code could be supported
through minor generalisations of existing compiler primitives.

Both Wasm and \trsegmentsfi rely on a type-directed forward-edge CFI which requires us to statically infer a limited amount of information about arguments to functions.\footnote{We do not need to infer any information about the heap or unsafe stack. Variadic functions, for example, can pass a dynamic number of arguments on the unsafe stack.}
This information includes the number of arguments, their width, and the calling convention.
In practice, this information is readily available as part of compilation and
does not require any complex control flow inference (unlike more precise fine
grain CFI schemes), so this is only a limitation when analyzing certain
binary programs.
Like most SFI schemes, both Wasm and our \trsegmentsfi do not currently
support JIT code compilation within the isolated component; adding this would
require engineering work, but can be done following the approaches of
\cite{nacl-jit, vahldiek-oberwagner_erim_2019}.
Finally, side channels are out of scope for this paper.

\section{Related work}
\label{sec:related}

A considerable amount of research has gone into efficient implementations of
memory isolation and CFI techniques to provide SFI across many
platforms~\cite{sehr_adapting_2010,mccamant_evaluating_2006,
goonasekera_libvm_2015, lucet, omniware, omniware-pldi, vino, herder2009fault,
rockjit, robusta, vx32, gang-sfi, payer2011fine, wedge, lwc, shreds, bgi}.
However, these systems either implement or require the user to implement
heavyweight springboards and trampolines to guarantee security.

\para{SFI systems}
\citet{wahbe_efficient_1993} suggest two ways to optimize
transitions: (1) partitioning the registers used by the application and the
sandboxed component and (2) performing link time optimizations (LTO)
that conservatively eliminates register saves that are never used in the entire
sandboxed component (not just the callee).
Register partitioning would cause slowdowns due to increased spilling.
Native Client~\cite{yee_native_2009} optimized 
transitions by clearing and saving contexts using 
machine specific mechanisms to, e.g., clear floating point 
state and SIMD registers in bulk.
However, we show (\secref{sec:eval}) that, even with those optimizations, the
software model imposes significant transition overheads.
While CPU makers continue to add optimized context switching
instructions, such instructions do not yet eliminate all
overhead.

\citet{zeng-tan2011} combine an SFI scheme with a rich CFI scheme
enforcing structure on executing code.
While a similar approach, their goal is to safely perform optimizations
to elide SFI and CFI bounds checks, and they do not impose sufficient structure to
enforce well-bracketing, a necessary property for zero-cost transitions.
XFI~\cite{xfi} also combines an SFI scheme with a rich CFI scheme and adopts a
safe stack model.
While meeting many of the zero-cost conditions, it does not prevent reading
uninitialized scratch registers and therefore cannot ensure confidentiality
without heavyweight springboards that clear scratch registers.
They also do not specify the CFG granularity, so it is not clear if is strong
enough to satisfy the zero-cost type-safe CFI requirement.

\para{WebAssembly based isolation}
WasmBoxC~\cite{wasmboxc} sandboxes C code by compiling to Wasm
followed by (de)compiling back to C, ensuring that the sandboxed code will
inherit isolation properties from Wasm.
The sandboxed library code can be safely linked with C applications, enabling a
form of zero-cost transition.
The zero-cost Wasm SFI system described by this paper was designed and released
prior to and independently of WasmBoxC, as the creators of WasmBoxC acknowledge
(citation elided for DBR).
Moreover, we believe that the theory developed in this paper provides a foundation for
analyzing and proving the security of WasmBoxC though such analysis would need
to account for possible undefined behavior introduced in compiling to C.

Sledge~\cite{sledge} describes a Wasm runtime for edge computing, that
relies on Wasm properties to enable efficient isolation of serverless
components.
However, Sledge focuses on function
scheduling including preempting running Wasm programs,
so its needs for context saving differ from library sandboxing as
contexts must be saved even in the middle of function calls.

\para{SFI Verification}
Previous work on SFI
(e.g.,~\cite{mccamant_evaluating_2006,yee_native_2009, xfi, veriwasm}) uses a
\emph{verifier} or a theorem prover~\cite{armor, compcert-sfi} to validate the
relevant SFI properties of compiled sandbox code.
However, unlike \verifname{}, none of these verifiers establish sufficient
properties for zero-cost transitions.

\para{Hardware based isolation}
Hardware features such as memory protection
keys~\cite{vahldiek-oberwagner_erim_2019, hodor}, extended page
tables~\cite{qiang_libsec_2017}, virtualization
instructions~\cite{qiang_libsec_2017, dune}, or even dedicated hardware
designs~\cite{donky} can be used to speed up
memory isolation.
These works focus on the efficiency of memory isolation as well as switching
between protected memory domains; however these approaches also use a single memory region that contain both the stack and heap making them incompatible with zero-cost conditions, i.e. they require heavyweight transitions.
\tridealheavy and \tridealheavysixfour in \sectionref{sec:eval} studies an idealized version of such a scheme.

\para{Capabilities}
\citet{karger} and \citet{skorstengaard_stktokens_2019} look at protecting
interacting components on systems that provide hardware-enforced capabilities.
\citet{karger} specifically looks at how register saving and restoration can be
optimized based on different levels of trust between components, however their
analysis does not offer formal security guarantees.
\citet{skorstengaard_stktokens_2019} investigate a calling-convention based on
capabilities (\`a la CHERI~\cite{cheri}) that allow safe sharing of a
stack between distrusting components.
Their definition of well-bracketed control flow and local state encapsulation
via an overlay inspired our work, and our logical relation is also based
on their work.
However, their technique does not yet ensure an equivalent notion to our
confidentiality property, and further is tied to machine support for hardware
capabilities.

\para{Type safety for isolation}
There has also been work on using strongly-typed languages to provide similar
security benefits.
SingularityOS~\cite{aiken2006deconstructing, hunt2007singularity,
fahndrich2006language}, explored using Sing\# to build an OS with cheap
transitions between mutually untrusting processes.
Unlike the work on SFI techniques that zero-cost transitions extend, tools like
SingularityOS require engineering effort to rewrite unsafe components in new
safe languages.

At a lower level, Typed Assembly Language (TAL)~\cite{morrisett_system_1999,
morrisett_stack-based_2002, morrisett_talx86_1999} is a type-safe compilation
target for high-level type-safe languages.
Its type system enables proofs that assembly programs follow
calling conventions, and enables an elegant definition of stack safety through
polymorphism.
Unfortunately, SFI is designed with unsafe code in mind, so cannot generally be
compiled to meet TAL's static checks.
To handle this our zero-cost and security conditions instead capture the \emph{behavior} that
TAL's type system is designed to ensure.

\begin{acks}
  We thank the reviewers for their suggestions and insightful comments.
  Many thanks to Bobby Holley, Mike Hommey, Chris Fallin, Tom Ritter, Till
  Schneidereit, Andy Wortman, and Alon Zakai for fruitful discussions.
  We thank Chris Fallin, Pat Hickey, and Andy Wortman for working with us to
  integrate \verifname{} into the Lucet compiler.
  This work was supported in part by gifts from Cisco, Fastly, Google, and
  Intel; by the NSF under Grant Number CNS-1514435, CNS-2120642, CCF-1918573,
  CAREER CNS-2048262; and, by the CONIX Research Center, one of six centers in
  JUMP, a Semiconductor Research Corporation (SRC) program sponsored by DARPA.
  Conrad Watt was supported by the EPSRC grant REMS: \textit{Rigorous
  Engineering for Mainstream Systems} (EP/K008528/1), a Google PhD Fellowship in
  Programming Technology and Software Engineering, and a Research Fellowship from
  Peterhouse, University of Cambridge.
\end{acks}

\bibliography{local.bib}

\iftechreport{%
\appendix

\section{Language Definitions}
\label{appendix:language-defs}

\figref{fig:appendix:syntax:lang} presents the syntax of our sandbox language model.
For all programs we define the regions $M$, $M_p$, $H_p$, $S_p$, $C_p$, and $I$.
$M = \nats$ and represents the whole memory space.
$M_p$, $H_p$, and $S_p$ are the memory, heap, and stack of the application or library where the heap and stack sit disjointly inside the memory.
$C_p$ is set of instruction indices such that $C(n) = (p, \_)$.
$I$ is the set of import indices, the beginnings of application functions that the library is allowed to jump to.

A note on calling convention: arguments are passed on the stack and the return
address is placed above the arguments.

\begin{center}
  \begin{tabular}{>{$}r<{$} >{$}c<{$} >{$}r<{$} >{$}c<{$} >{$}l<{$}}
    & & pc, sp, n, \ell & \in & \nats \\
    \privs & \ni & p & \bnfdef & \trusted \bnfalt \untrusted \\
    \vals & \ni & v & \bnfdef & n \\
    \regs & \ni & r & \bnfdef & \mathtt{r}_{n} \bnfalt sp \bnfalt pc \\
    \checks & \ni & k & \in & \nats \rightharpoonup \nats \\

    \expressions & \ni & e & \bnfdef & r \bnfalt v \bnfalt e \oplus e \\

    \commands & \ni & c & \bnfdef & \cpop{r}{p} \\
    & & & \bnfalt & \cpush{p}{e} \\
    & & & \bnfalt & \cload{r}{k}{e} \\
    & & & \bnfalt & \cstore{k}{e}{e} \\
    & & & \bnfalt & \cmov{r}{e} \\
    & & & \bnfalt & \ccall{k}{e} \\
    & & & \bnfalt & \cret{k} \\
    & & & \bnfalt & \cjmp{k}{e} \\
    & & & \bnfalt & \cgatecall{n}{e} \\
    & & & \bnfalt & \cgateret \\

    \codes & \ni & C & \bnfdef & \nats \rightharpoonup \privs \times \commands \\
    \regvals & \ni & R & \bnfdef & \regs \rightarrow \vals \\
    \memories & \ni & M & \bnfdef & \nats \rightarrow \vals \\
    \states & \ni & \Psi & \bnfdef & \error \\
    & & & \bnfalt &
      \begin{array}[t]{lllll}
        \{
        & pc & \bnftypes & \nats \\
        & sp & \bnftypes & \nats \\
        & R & \bnftypes & \regvals \\
        & M & \bnftypes & \memories \\
        & C & \bnftypes & \codes & \}
      \end{array}
  \end{tabular}
  \captionof{figure}{Syntax}
  \label{fig:appendix:syntax:lang}
\end{center}

\figref{fig:appendix:operational} and
\figref{fig:appendix:operational:auxiliary} define the base small-step
operational semantics.
We separate this into transitions $\currentop{\Psi}{c} \step \Psi'$ and error transitions $\currentop{\Psi}{c} \step \error$.

\begin{center}
  \judgmentHead{}{\currentop{\Psi}{c} \step \Psi'}
  \begin{mathpar}
    \inferrule
    {
      \Psi.sp \in S_{p_s}
      \\ p_s \lesstrusted p
      \\\\ v = \Psi.M(\Psi.sp)
      \\ R' = R[r \mapsto v]
    }
    {\currentop{\Psi}{\cpop{r}{p}} \step \pcinc{\Psi}[sp \assign \Psi.sp - 1, R \assign R']}

    \inferrule
    {
      v = \immval{\Psi}{e}
      \\ sp' = \Psi.sp + 1
      \\\\ M' = \Psi.M[sp' \mapsto v]
      \\ sp' \in S_{p_s}
      \\ p_s \lesstrusted p
    }
    {\currentop{\Psi}{\cpush{p}{e}} \step \pcinc{\Psi}[sp \assign sp', M \assign M']}

    \inferrule
    {
      n = \immval{\Psi}{e}
      \\n' = k(n)
      \\\\ v = \Psi.M(n')
      \\ R' = \Psi.R[r \mapsto v]
    }
    {\currentop{\Psi}{\cload{r}{k}{e}} \step \pcinc{\Psi}[R \assign R']}

    \inferrule
    {
      n = \immval{\Psi}{e}
      \\ v = \immval{\Psi}{e'}
      \\\\ n' = k(n)
      \\ M' = \Psi.M[n' \mapsto v]
    }
    {\currentop{\Psi}{\cstore{k}{e}{e'}} \step \pcinc{\Psi}[M \assign M']}

    \inferrule
    {
      v = \immval{\Psi}{e}
      \\ R' = \Psi.R[r \mapsto v]
    }
    {\currentop{\Psi}{\cmov{r}{e}} \step \pcinc{\Psi}[R \assign R']}

    \inferrule
    {
      n = \immval{\Psi}{e}
      \\ n' = k(n)
      \\ sp' = \Psi.sp + 1
      \\\\ M' = \Psi.M[sp' \mapsto \Psi.pc + 1]
      \\ sp' \in S_{p_s}
    }
    {\currentop{\Psi}{\ccall{k}{e}} \step \Psi[pc \assign n', sp \assign sp', M \assign M']}

    \inferrule
    {
      n = \immval{\Psi}{e}
      \\ n' = k(n)
    }
    {\currentop{\Psi}{\cjmp{k}{e}} \step \Psi[pc \assign n']}

    \inferrule
    {
      n = \Psi.M(\Psi.sp)
      \\\\ n' = k(n)
      \\ \Psi.sp \in S_{p_s}
    }
    {\currentop{\Psi}{\cret{k}} \step \Psi[pc \assign n', sp \assign \Psi.sp - 1]}

    \inferrule
    {
      v = \immval{\Psi}{e}
    }
    {\currentop{\Psi}{\cmov{sp}{e}} \step \pcinc{\Psi}[sp \assign v]}
  \end{mathpar}
  \captionof{figure}{Operational Semantics}
  \label{fig:appendix:operational}
\end{center}

\begin{center}
  \judgmentHead{}{\currentop{\Psi}{c} \step \error}
  \begin{mathpar}
    \inferrule
    {
      \Psi.sp \in S_{p_s}
      \\ p_s \nlesstrusted p
    }
    {\currentop{\Psi}{\cpop{r}{p}} \step \error}

    \inferrule
    {
      \Psi.sp + 1 \in S_{p_s}
      \\ p_s \nlesstrusted p
    }
    {\currentop{\Psi}{\cpush{p}{e}} \step \error}

    \inferrule
    {
      n = \immval{\Psi}{e}
      \\ k(n) \text{ undefined}
    }
    {\currentop{\Psi}{\cload{r}{k}{e}} \step \error}

    \inferrule
    {
      n = \immval{\Psi}{e}
      \\ k(n) \text{ undefined}
    }
    {\currentop{\Psi}{\cstore{k}{e}{e'}} \step \error}

    \inferrule
    {
      \\ \Psi.sp + 1 \notin S_{p_s}
    }
    {\currentop{\Psi}{\ccall{k}{e}} \step \error}

    \inferrule
    {
      n = \immval{\Psi}{e}
      \\ k(n) \text{ undefined}
    }
    {\currentop{\Psi}{\ccall{k}{e}} \step \error}

    \inferrule
    {
      \\ \Psi.sp \notin S_{p_s}
    }
    {\currentop{\Psi}{\cret{k}} \step \error}

    \inferrule
    {
      n = \Psi.M(\Psi.sp)
      \\ k(n) \text{ undefined}
    }
    {\currentop{\Psi}{\cret{k}} \step \error}

    \inferrule
    {
      n = \immval{\Psi}{e}
      \\ k(n) \text{ undefined}
    }
    {\currentop{\Psi}{\cjmp{k}{e}} \step \error}
  \end{mathpar}
  \captionof{figure}{Operational Semantics}
  \label{fig:appendix:operational:errors}
\end{center}

\begin{center}
  \begin{mathpar}
    \inferrule
    {C(\Psi.pc) = (\_, c)}
    {\currentop{\Psi}{c}}

    \inferrule
    {C(\Psi.pc) = (p, c)}
    {\currentcom{\Psi}{c}{p}}

    \inferrule
    {
      \Psi \stepstar \Psi'
      \\\\ \neg \exists \Psi''.~ \Psi' \step \Psi''
      \\ \Psi' \neq \error
    }
    {\Psi \evalsto \Psi'}

    \untrusted \lesstrusted \trusted
  \end{mathpar}

  \[\begin{array}{rcl}
    \immval{\Psi}{v} & \triangleq & v \\
    \immval{\Psi}{r} & \triangleq & \Psi.R(r) \\
    \immval{\Psi}{sp} & \triangleq & \Psi.sp \\
    \immval{\Psi}{pc} & \triangleq & \Psi.pc \\
    \immval{\Psi}{e \oplus e'} & \triangleq & \immval{\Psi}{e} \oplus \immval{\Psi}{e'} \\
    \\
    \pcinc{\Psi} & \triangleq &
    \begin{cases}
      \Psi[pc \assign \Psi.pc + 1] & \text{when } \pi_1(\Psi.C(\Psi.pc)) = \pi_1(\Psi.C(\Psi.pc + 1)) \\
      \error & \text{otherwise}
    \end{cases}
  \end{array}\]
  
  \captionof{figure}{Operational Semantics: Auxiliary Definitions}
  \label{fig:appendix:operational:auxiliary}
\end{center}

\figref{fig:appendix:operational:derived-forms} defines unguarded derived forms for memory operations.

\begin{center}
  \[\begin{array}{rcl}
    \cpop{r}{} & \triangleq & \cpop{r}{\top} \\
    \cpush{}{e} & \triangleq & \cpush{\top}{e} \\
    \cload{r}{}{e} & \triangleq & \cload{r}{id}{e} \\
    \cstore{}{e}{e'} & \triangleq & \cstore{id}{e}{e'} \\
    \cjmp{}{e} & \triangleq & \cjmp{id}{e} \\
    \ccall{}{e} & \triangleq & \ccall{id}{e} \\
    \cret{} & \triangleq & \cret{id} \\
  \end{array}\]
  \captionof{figure}{Derived Forms}
  \label{fig:appendix:operational:derived-forms}
\end{center}

\subsection{Sandbox Properties}

\begin{center}
  \begin{mathpar}
    \inferrule
    {
      \Psi_1 \step \Psi_2
      \\ \currentcom{\Psi_1}{c_1}{p_1}
      \\\\ \currentcom{\Psi_2}{c_2}{p_2}
      \\ p_1 = p_2 = p
    }
    {\Psi_1 \stepp{p} \Psi_2}

    \inferrule
    {\Psi \stepp{p} \Psi'}
    {\Psi \stepbox \Psi'}

    \inferrule
    {\Psi \stepwb \Psi'}
    {\Psi \stepbox \Psi'}

    \inferrule
    {
      \Psi \step \Psi_1 \stepboxstar \Psi_2 \step \Psi'
      \\\\ \currentop{\Psi}{\cgatecall{n}{e}}
      \\ \currentop{\Psi_2}{\cgateret}
    }
    {\Psi \stepwb \Psi'}
  \end{mathpar}
  \captionof{figure}{Well-Bracketed Transitions}
  \label{fig:appendix:well-bracketed-step}
\end{center}

\subsubsection{Integrity}

Integrity is all about maintaining application invariants across calls into the sandbox.
These invariants vary significantly from program to program, so to capture this
generality we define $\mathcal{I}$-Integrity and then instantiate it in several
specific instances.

\begin{definition}[$\mathcal{I}$-Integrity]{~}
  Let $\mathcal{I} : \mathit{Trace} \times \mathit{State} \times \mathit{State} \rightarrow \prop$.
  We say that an SFI transition systems has $\mathcal{I}$-integrity if
  $\Psi_0 \in \programs$, $\pi = \Psi_0 \stepstar \Psi_1$,
  $\currentcom{\Psi_1}{\_}{\trusted}$, and $\Psi_1 \stepwb \Psi_2$ imply that
  $\mathcal{I}(\pi, \Psi_1, \Psi_2)$.
\end{definition}

Informally, callee-save register integrity says that the values of callee-save registers are restored by gated calls into the sandbox:
\begin{definition}[Callee-Save Register Integrity]{~}

  Let $\mathbb{CSR}$ be the list of callee-save registers and define
  \[
    \mathcal{CSR}(\_, \Psi_1, \Psi_2) \triangleq \Psi_1.R(\mathbb{CSR}) = \Psi_2.R(\mathbb{CSR}).
  \]
  If an SFI transition system has $\mathcal{CSR}$-integrity then we say the system has callee-save register integrity.
\end{definition}

\begin{center}
  \begin{align*}
    \operatorname{return-address}_p & : \mathit{Trace} \rightarrow \powerset{\nats}
    \\
    \operatorname{return-address}_p(\Psi_0 \stepstar \currentcom{\Psi}{\ccall{k}{e}}{p} \step \Psi') & \triangleq \operatorname{return-address}_p(\Psi_0 \stepstar \Psi) \cup \{\Psi.sp + 1\}
    \\
    \operatorname{return-address}_p(\Psi_0 \stepstar \currentcom{\Psi}{\cret{k}}{p} \step \Psi') & \triangleq \operatorname{return-address}_p(\Psi_0 \stepstar \Psi) - \{\Psi.sp\}
    \\
    \operatorname{return-address}_p(\Psi_0 \stepstar \currentcom{\Psi}{\cgatecall{n}{e}}{p} \step \Psi') & \triangleq \operatorname{return-address}_p(\Psi_0 \stepstar \Psi) \cup \{\Psi.sp + 1\}
    \\
    \operatorname{return-address}_p(\Psi_0 \stepstar \currentcom{\Psi}{\cgateret}{p} \step \Psi') & \triangleq \operatorname{return-address}_p(\Psi_0 \stepstar \Psi) - \{\Psi.sp\}
    \\
    \operatorname{return-address}_p(\Psi_0 \stepstar \currentop{\Psi}{c} \step \Psi') & \triangleq \operatorname{return-address}_p(\Psi_0 \stepstar \Psi)
    \\
    \operatorname{return-address}_p(\Psi_0 \stepn{0} \Psi_0) & \triangleq \emptyset
  \end{align*}
  \captionof{figure}{Call stack return address calculation}
  \label{fig:appendix:return-address}
\end{center}

\begin{definition}[Return Address Integrity]
  Define
  \begin{align*}
    \mathcal{RA}(\pi, \Psi_1, \Psi_2) \triangleq (\Psi_1.M(\operatorname{return-address}_{\trusted}(\pi)) &= \Psi_2.M(\operatorname{return-address}_{\trusted}(\pi))) \\
    & {}\wedge (\Psi_2.sp = \Psi_1.sp)  \wedge (\Psi_2.pc = \Psi_1.pc + 1)
  \end{align*} 
  \[
  \]
  If an SFI transition system has $\mathcal{RA}$-integrity then we say the system has return address integrity.
\end{definition}

\subsubsection{Confidentiality}
\label{sec:appendix:confidentiality}

A confidentiality policy is defined by a partial function $\mathbb{C} \in \states \rightharpoonup (\nats \mathrel{+} \regs \rightarrow \privs)$.
The domain of $\mathbb{C}$ must include all $\Psi$ such that $\Psi$ is a program
state where the application is making a gated call into the library, that is,
where $\Psi_0.C(\Psi.pc) = (\trusted, \cgatecall{n'}{e})$.
$\mathbb{C}$ captures which registers and memory slots are labelled confidential
($\trusted$) or public ($\untrusted$) at a gated call into the sandbox.
In the following we use $f|_{X}$ for the restriction of the function $f : A
\rightarrow B$ to a subset $X \subseteq A$.

We say $\Psi =_{\mathbb{C}} \Psi'$ when
\begin{enumerate}
\item $\Psi.pc = \Psi'.pc$
\item $\Psi.sp = \Psi'.sp$
\item $\currentcom{\Psi}{\cgatecall{n}{e}}{\trusted}$ and $\currentcom{\Psi'}{\cgatecall{n}{e}}{\trusted}$
\item $\Psi.R|_{\{r \mid \mathbb{C}(\Psi)(r) = \untrusted\}} = \Psi'.R|_{\{r \mid \mathbb{C}(\Psi')(r) = \untrusted\}}$
\item $\Psi.M|_{\{n \mid \mathbb{C}(\Psi)(n) = \untrusted\}} = \Psi'.M|_{\{n \mid \mathbb{C}(\Psi')(n) = \untrusted\}}$
\end{enumerate}

We then define two notions of observational equivalence.
\begin{definition} \label{appendix:call-equivalence}
  We say $\Psi =_{\mathtt{call}\ n} \Psi'$ if
  \begin{enumerate}
  \item $\Psi.M(H_{\untrusted}) = \Psi'.M(H_{\untrusted})$

  \item $\Psi.pc = \Psi'.pc$

  \item $\Psi.sp = \Psi'.sp$

  \item For all $i \in [1, n]$, there exists some $n'$ such that $n' = \Psi.M(\Psi.sp - i) = \Psi'.M(\Psi.sp - i)$.
  \end{enumerate}
\end{definition}

Second, let $r_{ret}$ be the calling convention return register.
\begin{definition} \label{appendix:ret-equivalence}
  We say $\Psi =_{\mathtt{ret}} \Psi'$ if
  \begin{enumerate}
  \item $\Psi.M(H_{\untrusted}) = \Psi'.M(H_{\untrusted})$

  \item $\Psi.pc = \Psi'.pc$

  \item There exists some $n$ such that $n = \Psi.R(r_{ret}) = \Psi'.R(r_{ret})$.
  \end{enumerate}
\end{definition}

\begin{definition}[\StrongNI{}]{~}

  We say that an SFI transition system has the \strongni{} property if,
  for all initial configurations and their confidentiality properties $(\Psi_0, \mathbb{C}) \in \programs$,
  traces $\Psi_0 \stepstar \Psi_1 \step \Psi_2 \steplowstar \Psi_3 \step \Psi_4$,
  where $\Psi_1$ is a gated call into the library
  $\currentcom{\Psi_1}{\cgatecall{n}{e}}{\trusted}$,
  and $\Psi_3 \step \Psi_4$ leaves the library and reenters the application
  ($\currentcom{\Psi_4}{\_}{\trusted}$),
  and, for all $\Psi_1'$ such that $\Psi_1 =_{\mathbb{C}}
  \Psi_1'$, we have that $\Psi_1' \step \Psi_2' \steplowstar \Psi_3' \step \Psi_4'$,
  $\currentcom{\Psi_4'}{\_}{\trusted}$, $\Psi_4.pc = \Psi_4'.pc$, and
  \begin{enumerate}
  \item $\Psi_3$ is a gated call to the application ($\currentop{\Psi_3}{\cgatecall{m}{e}}$ and $\currentop{\Psi_3'}{\cgatecall{m}{e}}$) and $\Psi_4 =_{\mathtt{call}\ m} \Psi_4'$ or
  \item $\Psi_3$ is a gated return to the application ($\currentop{\Psi_3}{\cgateret}$ and $\currentop{\Psi_3'}{\cgateret}$) and $\Psi_4 =_{\mathtt{ret}} \Psi_4'$.
  \end{enumerate}
\end{definition}
\section{NaCl}
\label{appendix:nacl}

\begin{center}
  \begin{tabular}{>{$} c <{$} | c}
    \multicolumn{2}{c}{NaCl context in application} \\
    \hline
    ctx - 0 & library stack pointer \\
    \ldots & \\
    ctx^{\ast} & $ctx$ \\

    \multicolumn{2}{c}{} \\

    \multicolumn{2}{c}{NaCl context in library} \\
    \hline
    ctx - 0 & application stack pointer \\
    ctx - 1 & $\mathbb{CSR}_0$ \\
    ctx - 2 & $\mathbb{CSR}_1$ \\
    \ldots & \ldots \\
    ctx - \operatorname{len}(\mathbb{CSR}) & $\mathbb{CSR}_{\operatorname{len}(\mathbb{CSR}) - 1}$ \\
    \ldots \\
    ctx^{\ast} & $ctx$ \\
  \end{tabular}
  \captionof{figure}{Transition Context Layout}
  \label{fig:appendix:nacl:context-layout}
\end{center}

\begin{center}
  \begin{small}
  \begin{tabular}{>{$} r <{$} | >{$} l <{$} @{\hspace{\tabcolsep}} l}
    \multicolumn{2}{l}{$\operatorname{nacl-springboard}(n, e) \triangleq$} \\
    \hline
    & \cload{r_0}{}{ctx^{\ast}}
    \\
    & \cload{r_1}{}{r_0}
    & // $r_1$ holds the library stack pointer
    \\
    j \in (\operatorname{len}(\mathbb{CSR}), 0]
    & \overline{
      \cstore{}{r_0}{\mathbb{CSR}_j};
      \cmov{r_0}{r_0 + 1}
    }
    & // save callee-save registers
    \\
    & \cmov{r_1}{r_1 + n}
    & // set $r_1$ to the new top of the library stack
    \\
    & \cmov{sp}{sp - 1}
    & // move the stack pointer to the first argument
    \\
    j \in [0, n)
    & \overline{
      \cpop{r_2}{};
      \cstore{M_{\untrusted}}{r_1}{r_2};
      \cmov{r_1}{r_1 - 1}
    }
    & // copy arguments to library stack
    \\
    & \cmov{r_2}{sp + (n + 1)}; \cstore{}{r_0}{r_2}
    & // save stack pointer
    \\
    & \cmov{sp}{r_1 + n}
    & // set new stack pointer
    \\
    & \cstore{}{ctx^{\ast}}{r_0}
    & // update $ctx$
    \\
    r \in \mathbb{R}
    & \overline{\cmov{r}{0}}
    & // clear registers
    \\
    & \cjmp{}{e}
  \end{tabular}
  \end{small}
\end{center}

\begin{center}
  \begin{small}
  \begin{tabular}{>{$} r <{$} | >{$} l <{$} l}
    \multicolumn{2}{l}{$\operatorname{nacl-trampoline} \triangleq$} \\
    \hline
    & \cload{r_0}{}{ctx^{\ast}}
    \\
    j \in [0, \operatorname{len}(\mathbb{CSR}))
    & \overline{
      \cmov{r_0}{r_0 - 1};
      \cload{\mathbb{CSR}_j}{}{r_0}
    }
    & // restore callee-save registers
    \\
    & \cload{r_0}{}{ctx^{\ast}}
    \\
    & \cload{r_1}{}{r_0}
    & // $r_1$ holds the application stack pointer
    \\
    & \cmov{r_0}{r_0 - \operatorname{len}(\mathbb{CSR})}; \cstore{}{r_0}{sp}
    & // save library stack pointer
    \\
    & \cstore{}{ctx^{\ast}}{r_0}
    & // update $ctx$
    \\
    & \cmov{sp}{r_1}
    & // switch to application stack
    \\
    & \cret{}
  \end{tabular}
  \end{small}
\end{center}

\begin{center}
  \begin{small}
  \begin{tabular}{>{$} r <{$} | >{$} l <{$} l}
    \multicolumn{2}{l}{$\operatorname{nacl-cb-springboard}(n, e) \triangleq$} \\
    \hline
    & \cload{r_0}{}{ctx^{\ast}}
    \\
    & \cload{r_1}{}{r_0}
    & // $r_1$ holds the application stack pointer
    \\
    & \cmov{r_0}{r_0 + 1}; \cstore{}{r_0}{sp}
    & // save stack pointer
    \\
    & \cstore{}{ctx^{\ast}}{r_0}
    & // update $ctx$
    \\
    & \cmov{sp}{sp - 1}
    & // move the stack pointer to the first argument
    \\
    & \cmov{r_1}{r_1 + n}
    & // set $r_1$ to the new top of the library stack
    \\
    j \in [0, n)
    & \overline{
      \cpop{r_2}{M_{\untrusted}};
      \cstore{}{r_1}{r_2};
      \cmov{r_1}{r_1 - 1}
      }
    & // copy arguments to application stack
    \\
    & \cmov{sp}{r_1 + n}
    & // set new stack pointer
    \\
    & \cjmp{I}{e}
  \end{tabular}
  \end{small}
\end{center}

\begin{center}
  \begin{small}
  \begin{tabular}{>{$} r <{$} | >{$} l <{$} l}
    \multicolumn{2}{l}{$\operatorname{nacl-cb-trampoline} \triangleq$} \\
    \hline
    & \cload{r_0}{}{ctx^{\ast}}
    \\
    & \cload{r_1}{}{r_0}
    & // $r_1$ holds the library stack pointer
    \\
    & \cmov{r_0}{r_0 - 1}; \cstore{}{ctx^{\ast}}{r_0}
    & // update $ctx$
    \\
    & \cmov{sp}{r_1}
    & // switch to library stack
    \\
    r \in \mathbb{R}
    & \overline{\cmov{r}{0}}
    & // clear registers
    \\
    & \cret{C_{\untrusted}}
  \end{tabular}
  \end{small}
\end{center}

\subsection{Programs}

A NaCl program $\Psi$ is defined by the following conditions:

\begin{enumerate}
\item
  All memory operations in the sandboxed library are guarded:
  \begin{align*}
    \forall n.
    & \Psi.C(n) = (\untrusted, \cpop{r}{p}) \Longrightarrow p = \untrusted \\
    & \Psi.C(n) = (\untrusted, \cpush{p}{e}) \Longrightarrow p = \untrusted \\
    & \Psi.C(n) = (\untrusted, \cload{r}{k}{e}) \Longrightarrow \cod{k} \subseteq M_{\untrusted} \\
    & \Psi.C(n) = (\untrusted, \cstore{k}{e}{e'}) \Longrightarrow \cod{k} \subseteq M_{\untrusted}.
  \end{align*}

\item
  The application does not write $\trusted$ data to the sandbox memory.

\item
  Gated calls are the only way to move between application and library code:
  \begin{align*}
    \forall n.
    & \Psi.C(n) = (p, \ccall{k}{e}) \Longrightarrow \cod{k} \subseteq C_{p} \\
    & \Psi.C(n) = (p, \cret{k}) \Longrightarrow \cod{k} \subseteq C_{p} \\
    & \Psi.C(n) = (p, \cjmp{k}{e}) \Longrightarrow \cod{k} \subseteq C_{p} \\
  \end{align*}

\item
  The program starts in the application: $\Psi.pc = 0$ and $\Psi.C(0) = (\trusted, \_)$.

\item
  $ctx^{\ast}$ and $ctx$ start initialized to the library stack:
  \begin{align*}
    & ctx^{\ast} \in H_{\trusted} \\
    & ctx = \Psi.M(ctx^{\ast}) \in H_{\trusted} \\
    & \Psi.M(ctx) = S_{\untrusted}[0] - 1.
  \end{align*}

\item
  When calling into the library NaCl considers all registers except the stack
  pointer and program counter confidential.
  It further labels sandbox memory and the $n'$ arguments in the application
  stack as public with the remainder of application memory labeled as
  confidential.

  \begin{align*}
    \mathbb{C}(\Psi)(r) &=
      \begin{cases}
        \untrusted & \text{when } r = sp \\
        \untrusted & \text{when } r = pc \\
        \trusted & \text{otherwise}
      \end{cases} \\
    \mathbb{C}(\Psi)(n) &=
      \begin{cases}
        \untrusted & \text{when } n \in M_{\untrusted} \\
        \untrusted & \text{when } n \in (\Psi.sp - n', \Psi.sp] \text{ where } \currentcom{\Psi}{\cgatecall{\mathit{n'}}{e}}{\trusted} \\
        \trusted & \text{when } n \in M_{\trusted} \wedge n \notin (\Psi.sp - n', \Psi.sp] \text{ where } \currentcom{\Psi}{\cgatecall{\mathit{n'}}{e}}{\trusted}
      \end{cases}
  \end{align*}
\end{enumerate}

\subsection{Properties}

Throughout the following we will use the shorthand $ctx_{\Psi} \triangleq \Psi.M(ctx^{\ast})$.

\begin{proposition}
  NaCl has the \strongni{} property.
\end{proposition}
\begin{proof}
  Follows immediately from the fact that all reads and writes are guarded to be within $M_{\untrusted}$, all values in $M_{\untrusted}$ have label $\untrusted$, and all jumps remain within the library code.
\end{proof}

\begin{lemma} \label{lemma:appendix:nacl:context-in-trusted}
  The trampoline context is in $H_{\trusted}$.
\end{lemma}

\begin{lemma} \label{lemma:appendix:nacl:ctx-positive}
  $ctx \geq ctx_0$.
\end{lemma}

\begin{lemma} \label{lemma:appendix:nacl:lib-p-steps-preserve-app-memory}
  If $\currentcom{\Psi_1}{c}{\untrusted}$ and $\Psi_1 \steppstar{p} \Psi_2$, then $\Psi''.M_{\trusted} = \Psi'.M_{\trusted}$.
\end{lemma}
\begin{proof}
  There are two cases for $p$: $p = \trusted$ and $p = \untrusted$.
  If $p = \trusted$, then $\Psi_2 = \Psi_1$ and therefore trivially $\Psi_2.M_{\trusted} = \Psi_1.M_{\trusted}$.
  If $p = \untrusted$, then for all $\Psi$ such that $\Psi_1 \stepstar \Psi \stepn{+} \Psi_2$, $\Psi.C(\Psi.pc) = (\untrusted, \_)$.
  By the structure of a NaCl program, this ensures that $\Psi_2.M_{\trusted} = \Psi_1.M_{\trusted}$.
\end{proof}

\begin{lemma} \label{lemma:appendix:nacl:pstep-nogate}
  If $\Psi \steppstar{p} \Psi'$ then $\Psi.p = \Psi'.p$.
\end{lemma}

\begin{lemma} \label{lemma:appendix:nacl:wb-flips-priv}
  If $\Psi \stepwb \Psi'$ where $\Psi \stepstar \Psi'' \step \Psi'$, then $\Psi''.p \neq \Psi.p$ and $\Psi'.p = \Psi.p$.
\end{lemma}
\begin{proof}
  We proceed by simultaneous induction on the well-bracketed transition $\Psi \stepwb \Psi'$ and the length of $\Psi_0 \stepstar \Psi \stepwb \Psi'$.

  \begin{itemize}
  \pfcase{No callbacks}

      We have that $\currentop{\Psi}{\cgatecall{n}{e}}$ and there exist $\Psi_1$, $\Psi_2$, and $p$ such that $\Psi \step \Psi_1 \steppstar{p} \Psi_2 \step \Psi'$ where $\currentop{\Psi_2}{\cgateret}$.
      Here $\Psi'' = \Psi_2$.
      By inspection of the reduction for $\cgatecall{n}{e}$ we know that $\Psi_1.p \neq \Psi.p$ and therefore by \lemref{lemma:appendix:nacl:pstep-nogate} $\Psi''.p \neq \Psi.p$.
      By inspection of the reduction for $\cgateret$, $\Psi'.p = \Psi.p$.

  \pfcase{Callbacks}

      We have that $\currentop{\Psi}{\cgatecall{n}{e}}$ and there exist $\Psi_1$, $\Psi_2$  such that $\Psi \step \Psi_1 \stepboxstar \Psi_2 \step \Psi'$ where $\currentop{\Psi_2}{\cgateret}$.
      Here $\Psi'' = \Psi_2$.
      By inspection of the reduction for $\cgatecall{n}{e}$ we know that $\Psi_1.p \neq \Psi.p$.
      We now show, by induction on $\Psi_1 \stepboxstar \Psi_2$, that $\Psi_2.p = \Psi_1.p \neq \Psi.p$.

      \begin{subproof}
        If there are no steps then clearly $\Psi_2.p = \Psi_1.p \neq \Psi.p$.
        There are two possible cases for $\Psi_1 \stepboxstar \Psi_3 \stepbox \Psi_4$.
        When $\Psi_3 \stepp{p} \Psi_4$, \lemref{lemma:appendix:nacl:pstep-nogate} gives us that $\Psi_4.p = \Psi_3.p = \Psi_1.p \neq \Psi.p$.
        When $\Psi_3 \stepwb \Psi_4$, our outer inductive hypothesis gives us that $\Psi_4.p = \Psi_3.p = \Psi_1.p \neq \Psi.p$.
      \end{subproof}

      Lastly, by inspection of the reduction for $\cgateret$, $\Psi'.p = \Psi.p$.
  \end{itemize}
\end{proof}

\begin{lemma} \label{lemma:appendix:nacl:stepbox-preserves-p}
  If $\Psi \stepboxstar \Psi'$, then $\Psi'.p = \Psi.p$.
\end{lemma}
\begin{proof}
  By induction, \lemref{lemma:appendix:nacl:pstep-nogate}, and \lemref{lemma:appendix:nacl:wb-flips-priv}.
\end{proof}

\begin{lemma}[Context Integrity] \label{lemma:appendix:nacl:context-integrity}
  Let $\Psi_0 \in \programs$, $\Psi_0 \stepstar \Psi$, and $\Psi \stepwb \Psi'$.
  Then
  \begin{enumerate}
  \item if $\Psi.p = \trusted$, then $\Psi.M([ctx_{\Psi_0}, ctx_{\Psi})) = \Psi'.M([ctx_{\Psi_0}, ctx_{\Psi'}))$ and $ctx_{\Psi} = ctx_{\Psi'}$,
  \item if $\Psi.p = \untrusted$, then $\Psi.M([ctx_{\Psi_0}, ctx_{\Psi}]) = \Psi'.M([ctx_{\Psi_0}, ctx_{\Psi'}])$ and $ctx_{\Psi} = ctx_{\Psi'}$.
  \end{enumerate}
\end{lemma}
\begin{proof}
  We proceed by mutual, simultaneous induction on the well-bracketed transition $\Psi_1 \stepwb \Psi_2$ and the length of $\Psi_0 \stepstar \Psi \stepwb \Psi'$.

  First we consider the case where $\currentcom{\Psi}{c}{\trusted}$.
  \begin{itemize}
  \pfcase{No callbacks}

    We have that $\currentop{\Psi}{\cgatecall{n}{e}}$ and there exist $\Psi_1$, $\Psi_2$, and $p$ such that $\Psi \step \Psi_1 \steppstar{p} \Psi_2 \step \Psi'$ where $\currentop{\Psi_2}{\cgateret}$.
    By \lemref{lemma:appendix:nacl:lib-p-steps-preserve-app-memory} we have that $\Psi_1.M_{\trusted} = \Psi_2.M_{\trusted}$.
    By assumption we have that $\currentcom{\Psi}{\cgatecall{n}{e}}{\trusted}$ and therefore by \lemref{lemma:appendix:nacl:wb-flips-priv} we have that $\currentcom{\Psi_2}{\cgateret}{\untrusted}$.
    By \lemref{lemma:appendix:nacl:context-in-trusted}, \lemref{lemma:appendix:nacl:ctx-positive}, and inspection of the reduction rules for $\cgatecall{n}{e}$ and $\cgateret$ we have that $\Psi.M([ctx_{\Psi_0}, ctx_{\Psi})) = \Psi'.M([ctx_{\Psi_0}, ctx_{\Psi'}))$ and $ctx_{\Psi} = ctx_{\Psi'}$.

  \pfcase{Callbacks}

    We have that $\currentop{\Psi}{\cgatecall{n}{e}}$ and there exist $\Psi_1$, $\Psi_2$  such that $\Psi \step \Psi_1 \stepboxstar \Psi_2 \step \Psi'$ where $\currentop{\Psi_2}{\cgateret}$.
    By inspection of the reduction rule for $\cgatecall{n}{e}$ we have that $ctx_{\Psi_1} = ctx_{\Psi} + \operatorname{len}(\mathbb{CSR})$.
    We now show, by induction on $\Psi_1 \stepboxstar \Psi_2$, that $ctx_{\Psi_2} = ctx_{\Psi_1}$ and $\Psi_1.M([ctx_{\Psi_0}, ctx_{\Psi_1}]) = \Psi_2.M([ctx_{\Psi_0}, ctx_{\Psi_2}])$.

    \begin{subproof}
      If there are no steps then clearly $ctx_{\Psi_2} = ctx_{\Psi_1}$ and all of $\Psi_1.M_{\trusted} = \Psi_2.M_{\trusted}$.
      There are two possible cases for $\Psi_1 \stepboxstar \Psi_3 \stepbox \Psi_4$, and notice that in both $\Psi_3.p = \Psi_1.p = \untrusted$ (by \lemref{lemma:appendix:nacl:stepbox-preserves-p}).
      When $\Psi_3 \stepp{p} \Psi_4$, \lemref{lemma:appendix:nacl:lib-p-steps-preserve-app-memory} gives us that $\Psi_3.M_{\trusted} = \Psi_4.M_{\trusted}$ and then \lemref{lemma:appendix:nacl:context-in-trusted} gives us that $ctx_{\Psi_4} = ctx_{\Psi_3} = ctx_{\Psi_1}$.
      When $\Psi_3 \stepwb \Psi_4$, case 2 of our inductive hypothesis gives us that $\Psi_1.M([ctx_{\Psi_0}, ctx_{\Psi_1}]) = \Psi_3.M([ctx_{\Psi_0}, ctx_{\Psi_3}]) = \Psi_4.M([ctx_{\Psi_0}, ctx_{\Psi_4}])$ and $ctx_{\Psi_4} = ctx_{\Psi_3} = ctx_{\Psi_1}$.
    \end{subproof}

    Finally, by \lemref{lemma:appendix:nacl:stepbox-preserves-p} we get that $\Psi_2 = \untrusted$ and then inspection of the reduction rule for $\cgateret$ gives us that $\Psi.M([ctx_{\Psi_0}, ctx_{\Psi})) = \Psi'.M([ctx_{\Psi_0}, ctx_{\Psi'}))$ and $ctx_{\Psi} = ctx_{\Psi'}$.
  \end{itemize}

  Second we consider the case where $\currentcom{\Psi}{c}{\untrusted}$.
  \begin{itemize}
  \pfcase{No callbacks}

    We have that $\currentop{\Psi}{\cgatecall{n}{e}}$ and there exist $\Psi_1$, $\Psi_2$, and $p$ such that $\Psi \step \Psi_1 \steppstar{p} \Psi_2 \step \Psi'$ where $\currentop{\Psi_2}{\cgateret}$.
    By the structure of a NaCl program we have that $ctx_{\Psi_1} = ctx_{\Psi_2}$ and $\Psi_1.M([ctx_{\Psi_0}, ctx_{\Psi_1}]) = \Psi_2.M([ctx_{\Psi_0}, ctx_{\Psi_2}])$.
    By assumption we have that $\currentcom{\Psi}{\cgatecall{n}{e}}{\untrusted}$ and therefore by \lemref{lemma:appendix:nacl:wb-flips-priv} we have that $\currentcom{\Psi_2}{\cgateret}{\trusted}$.
    By \lemref{lemma:appendix:nacl:ctx-positive} and inspection of the reduction rules for $\cgatecall{n}{e}$ and $\cgateret$ we have that $\Psi.M([ctx_{\Psi_0}, ctx_{\Psi})) = \Psi'.M([ctx_{\Psi_0}, ctx_{\Psi'}))$ and $ctx_{\Psi} = ctx_{\Psi'}$.

  \pfcase{Callbacks}

    We have that $\currentop{\Psi}{\cgatecall{n}{e}}$ and there exist $\Psi_1$, $\Psi_2$  such that $\Psi \step \Psi_1 \stepboxstar \Psi_2 \step \Psi'$ where $\currentop{\Psi_2}{\cgateret}$.
    By inspection of the reduction rule for $\cgatecall{n}{e}$ we have that $ctx_{\Psi_1} = ctx_{\Psi} + 1$.
    We now show, by induction on $\Psi_1 \stepboxstar \Psi_2$, that $ctx_{\Psi_2} = ctx_{\Psi_1}$ and $\Psi_1.M([ctx_{\Psi_0}, ctx_{\Psi_1})) = \Psi_2.M([ctx_{\Psi_0}, ctx_{\Psi_2}))$.

    \begin{subproof}
      If there are no steps then clearly $ctx_{\Psi_2} = ctx_{\Psi_1}$ and all of $\Psi_1.M_{\trusted} = \Psi_2.M_{\trusted}$.
      There are two possible cases for $\Psi_1 \stepboxstar \Psi_3 \stepbox \Psi_4$, and notice that in both $\Psi_3.p = \Psi_1.p = \trusted$ (by \lemref{lemma:appendix:nacl:stepbox-preserves-p}).
      When $\Psi_3 \stepp{p} \Psi_4$, the structure of a NaCl program gives us that $ctx_{\Psi_1} = ctx_{\Psi_2}$ and $\Psi_1.M([ctx_{\Psi_0}, ctx_{\Psi_1}]) = \Psi_2.M([ctx_{\Psi_0}, ctx_{\Psi_2}])$.
      When $\Psi_3 \stepwb \Psi_4$, case 1 of our inductive hypothesis gives us that $\Psi_1.M([ctx_{\Psi_0}, ctx_{\Psi_1})) = \Psi_3.M([ctx_{\Psi_0}, ctx_{\Psi_3})) = \Psi_4.M([ctx_{\Psi_0}, ctx_{\Psi_4}))$ and $ctx_{\Psi_4} = ctx_{\Psi_3} = ctx_{\Psi_1}$.
    \end{subproof}

    Finally, by \lemref{lemma:appendix:nacl:stepbox-preserves-p} we get that $\Psi_2 = \trusted$ and then inspection of the reduction rule for $\cgateret$ gives us that $\Psi.M([ctx_{\Psi_0}, ctx_{\Psi}]) = \Psi'.M([ctx_{\Psi_0}, ctx_{\Psi'}])$ and $ctx_{\Psi} = ctx_{\Psi'}$.
  \end{itemize}
\end{proof}

\begin{proposition}
  NaCl has callee-save register integrity.
\end{proposition}
\begin{proof}
  Follows from \lemref{lemma:appendix:nacl:context-integrity}, \lemref{lemma:appendix:nacl:wb-flips-priv}, and inspection of the reduction rules for $\cgatecall{n}{e}$ and $\cgateret$.
\end{proof}

\begin{lemma} \label{lemma:appendix:nacl:return-address-integrity}
  Let $\Psi_0 \in \programs$, $\pi = \Psi_0 \stepstar \Psi$, and $\Psi \stepwb \Psi'$, then $\Psi'.M(\operatorname{return-address}_{\trusted}(\pi)) = \Psi.M(\operatorname{return-address}_{\trusted}(\pi))$.
\end{lemma}
\begin{proof}
  First we consider the case where $\currentcom{\Psi}{c}{\trusted}$.
  \begin{itemize}
  \pfcase{No callbacks}

    \sloppy
    We have that $\currentop{\Psi}{\cgatecall{n}{e}}$ and there exist $\Psi_1$, $\Psi_2$, and $p$ such that $\Psi \step \Psi_1 \steppstar{p} \Psi_2 \step \Psi'$ where $\currentop{\Psi_2}{\cgateret}$.
    By inspection of the reduction rule for $\cgatecall{n}{e}$ we see that $\Psi_1.M(\Psi.sp + 1) = \Psi.pc + 1$.
    This is adding to the top of the stack, so by the structure of a NaCl program we have that $\Psi_1.M(\operatorname{return-address}_{\trusted}(\pi)) = \Psi.M(\operatorname{return-address}_{\trusted}(\pi))$.
    By \lemref{lemma:appendix:nacl:lib-p-steps-preserve-app-memory} we have that $\Psi_1.M_{\trusted} = \Psi_2.M_{\trusted}$ and therefore \lemref{lemma:appendix:nacl:context-in-trusted} gives us that $ctx_{\Psi_2} = ctx_{\Psi_1}$ and $\Psi_2.M(\operatorname{return-address}_{\trusted}(\pi)) = \Psi_1.M(\operatorname{return-address}_{\trusted}(\pi))$.
    If we inspect the trampoline code we see that, right before we execute the $\cret{}$, we have set $sp$ to $\Psi_2.M(ctx_{\Psi_2}) = \Psi_1.M(ctx_{\Psi_1}) = \Psi.sp + 1$.
    Thus, after returning the only part of the application stack that we modify is $\Psi.sp + 1$.
    This, and the fact that $\Psi_2.M(\operatorname{return-address}_{\trusted}(\pi)) = \Psi.M(\operatorname{return-address}_{\trusted}(\pi))$ gives us that $\Psi.M(\operatorname{return-address}_{\trusted}(\pi)) = \Psi'.M(\operatorname{return-address}_{\trusted}(\pi))$.

  \pfcase{Callbacks}

    We have that $\currentop{\Psi}{\cgatecall{n}{e}}$ and there exist $\Psi_1$, $\Psi_2$  such that $\Psi \step \Psi_1 \stepboxstar \Psi_2 \step \Psi'$ where $\currentop{\Psi_2}{\cgateret}$.
    By inspection of the reduction rule for $\cgatecall{n}{e}$ we see that $\Psi_1.M(\Psi.sp + 1) = \Psi.pc + 1$.
    This is adding to the top of the stack, so by the structure of a NaCl program we have that $\Psi_1.M(\operatorname{return-address}_{\trusted}(\pi)) = \Psi.M(\operatorname{return-address}_{\trusted}(\pi))$.
    We now show, by induction on $\Psi_1 \stepboxstar \Psi_2$ that $\Psi_2.M(ctx_{\Psi_2}) = \Psi.sp + 1$ and $\Psi_2.M(\operatorname{return-address}_{\trusted}(\pi)) = \Psi_1.M(\operatorname{return-address}_{\trusted}(\pi))$.

    \begin{subproof}
      If there are no steps then $\Psi_2 = \Psi_1$ and both goals hold immediately.
      There are two possible cases for $\Psi_1 \stepboxstar \Psi_3 \stepbox \Psi_4$ and notice that in both $\Psi_3.p = \Psi_1.p = \untrusted$ (by \lemref{lemma:appendix:nacl:stepbox-preserves-p}).
      If $\Psi_3 \stepp{p} \Psi_4$ then \lemref{lemma:appendix:nacl:lib-p-steps-preserve-app-memory} gives us that $\Psi_4.M_{\trusted} = \Psi_3.M_{\trusted}$ and our goal holds (as all of $\operatorname{return-address}_{\trusted}(\pi)$ is in $S_{\trusted}$).
      If $\Psi_3 \stepwb \Psi_4$, then \lemref{lemma:appendix:nacl:context-integrity} gives us that $ctx_{\Psi_3} = ctx_{\Psi_4}$ and $\Psi_3.M([ctx_{\Psi_0}, ctx_{\Psi_3}]) = \Psi_4.M([ctx_{\Psi_0}, ctx_{\Psi_4}])$ and therefore that $\Psi_4.M(ctx_{\Psi_4}) = \Psi.sp + 1$.
      $\operatorname{return-address}_{\trusted}(\Psi_0 \stepstar \Psi_3) = \operatorname{return-address}_{\trusted}(\pi) \uplus \Psi.sp + 1$, so our inductive hypothesis gives us that $\Psi_4.M(\operatorname{return-address}_{\trusted}(\pi)) = \Psi_3.M(\operatorname{return-address}_{\trusted}(\pi))$.
    \end{subproof}
  \end{itemize}

  Second we consider the case where $\currentcom{\Psi}{c}{\untrusted}$.
  \begin{itemize}
  \pfcase{No callbacks}

    We have that $\currentop{\Psi}{\cgatecall{n}{e}}$ and there exist $\Psi_1$, $\Psi_2$, and $p$ such that $\Psi \step \Psi_1 \steppstar{p} \Psi_2 \step \Psi'$ where $\currentop{\Psi_2}{\cgateret}$.
    By inspection of the reduction rule for $\cgatecall{n}{e}$ we see that $\Psi_1.M(\operatorname{return-address}_{\trusted}(\pi)) = \Psi.M(\operatorname{return-address}_{\trusted}(\pi))$.
    By the structure of a NaCl program we have that any call stack elements that are added during the callback will be popped before the $\cgateret$.
    Thus, $\Psi_2.M(\operatorname{return-address}_{\trusted}(\pi)) = \Psi_1.M(\operatorname{return-address}_{\trusted}(\pi))$.
    Inspection of the reduction rule for $\cgateret$ then gives us that $\Psi'.M(\operatorname{return-address}_{\trusted}(\pi)) = \Psi_2.M(\operatorname{return-address}_{\trusted}(\pi)) = \Psi.M(\operatorname{return-address}_{\trusted}(\pi))$.

  \pfcase{Callbacks}

    We have that $\currentop{\Psi}{\cgatecall{n}{e}}$ and there exist $\Psi_1$, $\Psi_2$  such that $\Psi \step \Psi_1 \stepboxstar \Psi_2 \step \Psi'$ where $\currentop{\Psi_2}{\cgateret}$.
    By inspection of the reduction rule for $\cgatecall{n}{e}$ we see that $\Psi_1.M(\operatorname{return-address}_{\trusted}(\pi)) = \Psi.M(\operatorname{return-address}_{\trusted}(\pi))$.
    We now show, by induction on $\Psi_1 \stepboxstar \Psi_2$ that $\Psi_2.M(\operatorname{return-address}_{\trusted}(\pi)) = \Psi_1.M(\operatorname{return-address}_{\trusted}(\pi))$.

    \begin{subproof}
      \sloppy
      If there are no steps then $\Psi_2 = \Psi_1$ and the goal holds immediately.
      There are two possible cases for $\Psi_1 \stepboxstar \Psi_3 \stepbox \Psi_4$ and notice that in both $\Psi_3.p = \Psi_1.p = \trusted$ (by \lemref{lemma:appendix:nacl:stepbox-preserves-p}).
      If $\Psi_3 \stepp{p} \Psi_4$ then the structure of a NaCl program gives us that any call stack elements that are added during the callback will be popped before the $\cgateret$, and therefore our inductive invariant is maintained.
      If $\Psi_3 \stepwb \Psi_4$, then notice that $\operatorname{return-address}_{\trusted}(\Psi_0 \stepstar \Psi_3) = \operatorname{return-address}_{\trusted}(\pi)$, so our inductive hypothesis gives us that $\Psi_4.M(\operatorname{return-address}_{\trusted}(\pi)) = \Psi_3.M(\operatorname{return-address}_{\trusted}(\pi))$.
    \end{subproof}

    Inspection of the reduction rule for $\cgateret$ then gives us that $\Psi'.M(\operatorname{return-address}_{\trusted}(\pi)) = \Psi_2.M(\operatorname{return-address}_{\trusted}(\pi)) = \Psi.M(\operatorname{return-address}_{\trusted}(\pi))$.
  \end{itemize}
\end{proof}

\begin{proposition}
  NaCl has return address integrity.
\end{proposition}
\begin{proof}
  We have that $\Psi_0 \in \programs$, $\pi = \Psi_0 \stepstar \Psi$, $\Psi.p = \trusted$, and $\Psi \stepwb \Psi'$ and wish to show that $\Psi.M(\operatorname{return-address}_{\trusted}(\pi)) = \Psi'.M(\operatorname{return-address}_{\trusted}(\pi))$, $\Psi'.sp = \Psi.sp$, and $\Psi'.pc = \Psi.pc + 1$.

  \lemref{lemma:appendix:nacl:return-address-integrity} gives us that $\Psi.M(\operatorname{return-address}_{\trusted}(\pi)) = \Psi'.M(\operatorname{return-address}_{\trusted}(\pi))$.
  We proceed by simultaneous induction on the well-bracketed transition $\Psi \stepwb \Psi'$ and the length of $\Psi_0 \stepstar \Psi \stepwb \Psi'$ to show that $\Psi'.sp = \Psi.sp$, and $\Psi'.pc = \Psi.pc + 1$.

  \begin{itemize}
  \pfcase{No callbacks}

    We have that $\currentop{\Psi}{\cgatecall{n}{e}}$ and there exist $\Psi_1$, $\Psi_2$, and $p$ such that $\Psi \step \Psi_1 \steppstar{p} \Psi_2 \step \Psi'$ where $\currentop{\Psi_2}{\cgateret}$.
    By inspection of the reduction rule for $\cgatecall{n}{e}$ we see that $\Psi_1.M(\Psi.sp + 1) = \Psi.pc + 1$ and $\Psi_1.M(ctx_{\Psi_1}) = \Psi.sp + 1$.
    By \lemref{lemma:appendix:nacl:lib-p-steps-preserve-app-memory} we have that $\Psi_1.M_{\trusted} = \Psi_2.M_{\trusted}$ and therefore \lemref{lemma:appendix:nacl:context-in-trusted} gives us that $ctx_{\Psi_2} = ctx_{\Psi_1}$.
    If we inspect the trampoline code we see that, right before we execute the $\cret{}$, we have set $sp$ to $\Psi_2.M(ctx_{\Psi_2}) = \Psi_1.M(ctx_{\Psi_1}) = \Psi.sp + 1$.
    Thus, after returning we have that $\Psi'.pc = \Psi.pc + 1$, $\Psi'.sp = \Psi.sp$.

  \pfcase{Callbacks}

    We have that $\currentop{\Psi}{\cgatecall{n}{e}}$ and there exist $\Psi_1$, $\Psi_2$  such that $\Psi \step \Psi_1 \stepboxstar \Psi_2 \step \Psi'$ where $\currentop{\Psi_2}{\cgateret}$.
    By inspection of the reduction rule for $\cgatecall{n}{e}$ we see that $\Psi_1.M(\Psi.sp + 1) = \Psi.pc + 1$ and $\Psi_1.M(ctx_{\Psi_1}) = \Psi.sp + 1$.
    We now show, by induction on $\Psi_1 \stepboxstar \Psi_2$ that $\Psi_2.M(\Psi.sp + 1) = \Psi.pc + 1$ and $\Psi_2.M(ctx_{\Psi_2}) = \Psi.sp + 1$.

    \begin{subproof}
      If there are no steps then $\Psi_2 = \Psi_1$ and both hold immediately.
      There are two possible cases for $\Psi_1 \stepboxstar \Psi_3 \stepbox \Psi_4$ and notice that in both $\Psi_3.p = \Psi_1.p = \untrusted$ (by \lemref{lemma:appendix:nacl:stepbox-preserves-p}).
      If $\Psi_3 \stepp{p} \Psi_4$ then \lemref{lemma:appendix:nacl:lib-p-steps-preserve-app-memory} gives us that $\Psi_4.M_{\trusted} = \Psi_3.M_{\trusted}$ and both hold (as the invariants are on $M_{\trusted}$).
      If $\Psi_3 \stepwb \Psi_4$, then \lemref{lemma:appendix:nacl:context-integrity} gives us that $ctx_{\Psi_3} = ctx_{\Psi_4}$ and $\Psi_3.M([ctx_{\Psi_0}, ctx_{\Psi_3}]) = \Psi_4.M([ctx_{\Psi_0}, ctx_{\Psi_4}])$ and therefore that $\Psi_4.M(ctx_{\Psi_4}) = \Psi.sp + 1$.
      $\operatorname{return-address}_{\trusted}(\Psi_0 \stepstar \Psi_3) = \operatorname{return-address}_{\trusted}(\pi) \uplus \Psi.sp + 1$ so \lemref{lemma:appendix:nacl:return-address-integrity} gives us that $\Psi_4.M(\Psi.sp + 1) = \Psi_3.M(\Psi.sp + 1) = \Psi.pc + 1$.
    \end{subproof}

    Finally, if we inspect the trampoline code we see that, right before we execute the $\cret{}$, we have set $sp$ to $\Psi_2.M(ctx_{\Psi_2}) = \Psi.sp + 1$.
    Thus, after returning we have that $\Psi'.pc = \Psi.pc + 1$, $\Psi'.sp = \Psi.sp$.
  \end{itemize}
\end{proof}
\section{Overlay Semantics}
\label{appendix:overlay}

\figref{fig:appendix:overlay:operational},
\figref{fig:appendix:overlay:operational-aux-judgments}, and
\figref{fig:appendix:overlay:operational-aux-definitions} define the overlay
monitor operational semantics.
The semantics are parameterized by a confidentiality policy $\mathbb{C}$.

\begin{center}
  \begin{tabular}{>{$}r<{$} >{$}c<{$} >{$}r<{$} >{$}c<{$} >{$}l<{$}}
    \vals & \ni & \oc{v} & \bnfdef & \val{n}{p} \\

    \frames & \ni & \oc{\SF} & \bnfdef &
      \begin{array}[t]{lllll}
        \{
        & \mathit{base} & \bnftypes & \nats \\
        & \mathit{ret\mbox{-}addr\mbox{-}loc} & \bnftypes & \nats \\
        & \mathit{csr\mbox{-}vals} & \bnftypes & \powerset{\regs \times \nats} & \}
      \end{array} \\

    \functions & \ni & \oc{F} & \bnfdef &
      \begin{array}[t]{lllll}
        \{
        & \mathit{instrs} & \bnftypes & \nats \rightharpoonup \commands \\
        & \mathit{entry} & \bnftypes & \nats \\
        & \mathit{type} & \bnftypes & \nats & \}
      \end{array} \\

    \ostates & \ni & \oc{\Phi} & \bnfdef & \oerror \\
    & & & \bnfalt &
      \begin{array}[t]{lllll}
        \{
        & \Psi & \bnftypes & \states \\
        & \mathit{funcs} & \bnftypes & \nats \rightharpoonup \functions \\
        & \mathit{stack} & \bnftypes & [\frames] & \}
      \end{array}
  \end{tabular}
  \captionof{figure}{Overlay Extended Syntax}
  \label{fig:appendix:syntax:overlay}
\end{center}

\begin{center}
  \judgmentHead{}{\currentop{\Psi}{c} \ostep \Psi'}
  \begin{mathpar}
    \inferrule
    {
      \val{n}{p_e} = \oimmval{\oc{\Phi}}{e}
      \\ n' = k(n)
      \\ sp' = \oc{\Phi}.sp + 1
      \\ M' = \oc{\Phi}.M[sp' \mapsto \oc{\Phi}.pc + 1]
      \\\\ \oname{typechecks}(\oc{\Phi}, n', sp')
      \\ \oc{\SF} = \oname{new-frame}(\oc{\Phi}, n', sp')
      \\ p_e \lesstrusted \untrusted
    }
    {\currentcom{\oc{\Phi}}{\ccall{k}{e}}{\untrusted} \ostep \oc{\Phi}[\mathit{stack} \assign [\oc{\SF}] \concat \oc{\Phi}.\mathit{stack}, pc \assign n', sp \assign sp', M \assign M'] }

    \inferrule
    {
      \oname{is-ret-addr-loc}(\oc{\Phi}, \oc{\Phi}.sp)
      \\ \natval{n} = \oc{\Phi}.M(\oc{\Phi}.sp)
      \\ n' = k(n)
      \\\\ \oname{csr-restored}(\oc{\Phi})
      \\ \oc{\Phi'} = \oname{pop-frame}(\oc{\Phi})
    }
    {\currentcom{\oc{\Phi}}{\cret{k}}{\untrusted} \ostep \oc{\Phi'}[pc \assign n', sp \assign \oc{\Phi}.sp - 1]}

    \inferrule
    {
      \natval{n'} = \oimmval{\oc{\Phi}}{e}
      \\ sp' = \oc{\Phi}.sp + 1
      \\ M' = \oc{\Phi}.M[sp' \mapsto \oc{\Phi}.pc + 1]
      \\\\ n' \in I
      \\ \oname{typechecks}(\oc{\Phi}, n', sp')
      \\ \oname{args-secure}(\oc{\Phi}, sp', n)
      \\ \oc{\SF} = \oname{new-frame}(\oc{\Phi}, n', sp')
    }
    {\currentcom{\oc{\Phi}}{\cgatecall{n}{e}}{\untrusted} \ostep \oc{\Phi}[\mathit{stack} \assign [\oc{\SF}] \concat \oc{\Phi}.\mathit{stack}, pc \assign n', sp \assign sp', M \assign M']}

    \inferrule
    {
      \oc{\Phi'} = \oname{classify}_{\mathbb{C}}(\oc{\Phi})
      \\ \natval{n'} = \oimmval{\oc{\Phi'}}{e}
      \\ sp' = \oc{\Phi'}.sp + 1
      \\\\ M' = \oc{\Phi'}.M[sp' \mapsto \oc{\Phi'}.pc + 1]
      \\ \oc{\SF} = \oname{new-frame}(\oc{\Phi'}, n', sp')
    }
    {\currentcom{\oc{\Phi}}{\cgatecall{n}{e}}{\trusted} \ostep \oc{\Phi'}[\mathit{stack} \assign [\oc{\SF}] \concat \oc{\Phi'}.\mathit{stack}, pc \assign n', sp \assign sp', M \assign M']}

    \inferrule
    {
      \currentop{\oc{\Phi}}{\cret{}} \ostep \oc{\Phi'}
      \\ p' \lesstrusted p
      \\\\ \val{n}{p'} = \oc{\Phi}.R(r_{ret})
    }
    {\currentcom{\oc{\Phi}}{\cgateret}{p} \ostep \oc{\Phi'}}

    \inferrule
    {
      \oc{v} = \oimmval{\oc{\Phi}}{e}
      \\ sp' = \oc{\Phi}.sp + 1
      \\ sp' \in S_p
      \\\\ M' = \oc{\Phi}.M[sp' \mapsto \oc{v}]
      \\ \oname{writeable}(\oc{\Phi}, sp')
    }
    {\currentop{\oc{\Phi}}{\cpush{p}{e}} \ostep \pcinc{\oc{\Phi}}[sp \assign sp', M \assign M']}

    \inferrule
    {
      \val{n}{p_e} = \oimmval{\oc{\Phi}}{e}
      \\ \oc{v} = \val{\_}{p_{e'}} = \oimmval{\oc{\Phi}}{e'}
      \\\\ M' = \oc{\Phi}.M[n' \mapsto \oc{v}]
      \\ \oname{writeable}(\oc{\Phi}, n')
      \\\\ n' = k(n)
      \\ p_e \lesstrusted p
      \\ p_{e'} \nlesstrusted p \Longrightarrow n' \notin H_{\untrusted}
    }
    {\currentcom{\oc{\Phi}}{\cstore{k}{e}{e'}}{p} \ostep \pcinc{\oc{\Phi}}[M \assign M']}

    \inferrule
    {
      \val{n}{p_e} = \oimmval{\oc{\Phi}}{e}
      \\ n' = k(n)
      \\ p_e \lesstrusted p
      \\\\ \oc{v} = \val{\_}{p_{n'}} = \oc{\Phi}.M(n')
      \\ R' = \oc{\Phi}.R[r \mapsto \oc{v}]
    }
    {\currentcom{\oc{\Phi}}{\cload{r}{k}{e}}{p} \ostep \pcinc{\oc{\Phi}}[R \assign R']}

    \inferrule
    {
      \val{v}{p_e} = \oimmval{\oc{\Phi}}{e}
      \\ p_e \lesstrusted p
    }
    {\currentcom{\oc{\Phi}}{\cmov{sp}{e}}{p} \ostep \pcinc{\oc{\Phi}}[sp \assign v]}

    \inferrule
    {
      \val{n}{p_e} = \oimmval{\oc{\Phi}}{e}
      \\ n' = k(n)
      \\\\ p_e \lesstrusted p
      \\ \oname{in-same-func}(\oc{\Phi}, \oc{\Phi}.pc, n')
    }
    {\currentcom{\oc{\Phi}}{\cjmp{k}{e}}{p} \ostep \oc{\Phi}[pc \assign n']}

    \inferrule
    {
      p_r \nlesstrusted p' \Longrightarrow p_r \lesstrusted p
      \\\\ \val{n}{p_r} = \oc{\Phi}.R(r)
      \\ R' = \oc{\Phi}.R[r \assign \val{n}{p'}]
    }
    {\currentcom{\oc{\Phi}}{\cmovlabel{r}{p'}}{p} \ostep \pcinc{\oc{\Phi}}[R \assign R']}

    \inferrule
    {
      \val{n}{p_e} = \oimmval{\oc{\Phi}}{e}
      \\ \val{m}{p_m} = \oc{\Phi}.M(n)
      \\ p_e \lesstrusted p
      \\\\ M' = \oc{\Phi}.M[n \assign \val{m}{p'}]
      \\ p_m \nlesstrusted p' \Longrightarrow p_m \lesstrusted p
    }
    {\currentcom{\oc{\Phi}}{\cstorelabel{p'}{e}}{p} \ostep \pcinc{\oc{\Phi}}[M \assign M']}

    \inferrule
    {
      \currentop{\oc{\Phi}.\Psi}{c} \step \Psi'
      \\\\ \oname{in-same-func}(\oc{\Phi}, \oc{\Phi}.\Psi.pc, \Psi'.pc)
    }
    {\currentop{\oc{\Phi}}{c} \ostep \oc{\Phi}[\Psi \assign \Psi']}
  \end{mathpar}
  \captionof{figure}{Overlay Operational Semantics}
  \label{fig:appendix:overlay:operational}
\end{center}

\begin{center}
  \begin{mathpar}
    \inferrule
    {
      [\oc{\SF}] \concat \_ = \oc{\Phi}.\mathit{stack}
      \\\\ n \in S_p \Longrightarrow
      n \geq \oc{\SF}.\mathit{base} \wedge
      n \neq \oc{\SF}.\mathit{ret\mbox{-}addr\mbox{-}loc}
    }
    {\oname{writeable}(\oc{\Phi}, n)}

    \inferrule
    {
      F = \oc{\Phi}.\mathit{funcs}(\mathit{target})
      \\ F.\mathit{entry} = \mathit{target}
      \\\\ sp \in S_p
      \\ [\oc{\SF}] \concat \_ = \oc{\Phi}.\mathit{stack}
      \\\\ sp \geq \oc{\SF}.\mathit{ret\mbox{-}addr\mbox{-}loc} + \oc{F}.\mathit{type}
    }
    {\oname{typechecks}(\oc{\Phi}, \mathit{target}, sp)}

    \inferrule
    {
      [\oc{\SF}] \concat \_ = \oc{\Phi}.\mathit{stack}
      \\\\ \mathit{ret\mbox{-}addr\mbox{-}loc} = \oc{\SF}.\mathit{ret\mbox{-}addr\mbox{-}loc}
    }
    {\oname{is-ret-addr-loc}(\oc{\Phi}, \mathit{ret\mbox{-}addr\mbox{-}loc})}

    \inferrule
    {
      [\oc{\SF}] \concat \_ = \oc{\Phi}.\mathit{stack}
      \\\\ \forall (r, n) \in \oc{\SF}.\mathit{csr\mbox{-}vals}.~ \oc{\Phi}.R(r) = n
    }
    {\oname{csr-restored}(\oc{\Phi})}

    \inferrule
    {
      \oc{F} \in \cod{\oc{\Phi}.\mathit{funcs}}
      \\ n, n' \in \dom{\oc{F}.\mathit{instrs}}
    }
    {\oname{in-same-func}(\oc{\Phi}, n, n')}

    \inferrule
    {
      \forall \oc{F} \in \cod{\oc{\Phi}.\mathit{funcs}}.~ n \notin \dom{\oc{F}.\mathit{instrs}}
    }
    {\oname{in-same-func}(\oc{\Phi}, n, n')}

    \inferrule
    {\forall i \in [1, n].~ \oc{\Phi}.M(sp - i) = \val{\_}{\untrusted}}
    {\oname{args-secure}(\oc{\Phi}, sp, n)}
  \end{mathpar}
  \captionof{figure}{Overlay Operational Semantics: Auxiliary Judgments}
  \label{fig:appendix:overlay:operational-aux-judgments}
\end{center}

\begin{center}
  \[\begin{array}{rcl}
    \oname{new-frame}(\oc{\Phi}, \mathit{target}, \mathit{ret\mbox{-}addr\mbox{-}loc}) & \triangleq &
      \begin{array}[t]{lllll}
        \{
        & \mathit{base} & = & \mathit{ret\mbox{-}addr\mbox{-}loc} - \oc{\Phi}.\mathit{funcs}(\mathit{target}).\mathit{type} \\
        & \mathit{ret\mbox{-}addr\mbox{-}loc} & = & \mathit{ret\mbox{-}addr\mbox{-}loc} \\
        & \mathit{csr\mbox{-}vals} & = & \{(r, \oc{\Phi}.R(r))\}_{r \in \mathbb{CSR}} & \}
      \end{array} \\
    \\
    \oname{pop-frame}(\oc{\Phi}) & \triangleq & \oc{\Phi}[\mathit{stack} \assign S] \quad \text{where } [\oc{\SF}] \concat S = \oc{\Phi}.\mathit{stack} \\
    \\
    \pcinc{\oc{\Phi}} & \triangleq &
      \begin{cases}
        \oc{\Phi}[\Psi \assign \pcinc{\Psi}] & \oname{in-same-func}(\oc{\Phi}, \oc{\Phi}.pc, \oc{\Phi}.pc + 1) \\
        \oerror & \text{otherwise}
      \end{cases} \\
      \\
      \oname{classify}_{\mathbb{C}}(\oc{\Phi}) & \triangleq &
      \oc{\Phi}[M \assign M', R \assign R'] \\
      & & \text{where } M'(n) = \val{\natval{\oc{\Phi}.M(n)}}{\mathbb{C}(\oc{\Phi}.\Psi)(n)} \\
      & & \phantom{\text{where }} R'(r) = \val{\natval{\oc{\Phi}.R(r)}}{\mathbb{C}(\oc{\Phi}.\Psi)(r)} \\
      \\
      \oimmval{\Psi}{\oc{v}} & \triangleq & \oc{v} \\
      \oimmval{\Psi}{r} & \triangleq & \Psi.R(r) \\
      \oimmval{\Psi}{sp} & \triangleq & \val{\Psi.sp}{\untrusted} \\
      \oimmval{\Psi}{pc} & \triangleq & \val{\Psi.pc}{\untrusted} \\
      \oimmval{\Psi}{e \oplus e'} & \triangleq & \val{v \oplus v'}{p \sqcup p'} \\
      & & \text{where } \val{v}{p} = \oimmval{\Psi}{e} \\
      & & \phantom{\text{where }} \val{v'}{p'} = \oimmval{\Psi}{e'} \\
      \\
      \natval{n} & \triangleq & \val{n}{\_} \\
      \\
      \trusted \sqcup p & \triangleq & \trusted \\
      p \sqcup \trusted & \triangleq & \trusted \\
      \untrusted \sqcup \untrusted & \triangleq & \untrusted \\
  \end{array}\]
  \captionof{figure}{Overlay Operational Semantics: Auxiliary Definitions}
  \label{fig:appendix:overlay:operational-aux-definitions}
\end{center}

\begin{lemma}[Overlay is a refinement] \label{appendix:overlay:refinement}
  For any $\oPhi \ostep \oc{\Phi'}$, if $\oc{\Phi'} \neq \oerror$, then $\oPhi.\Psi \step \oc{\Phi'}.\Psi$
\end{lemma}

\begin{lemma}[Overlay is equivalent on application reduction] \label{appendix:overlay:application-equivalent}
  For any $\oPhi$, if $\oPhi.\Psi \stephigh \Psi'$, then $\oPhi \ostep \{\Psi \assign \Psi', \mathit{funcs} \assign \oPhi.\mathit{funcs}, \mathit{stack} \assign \oPhi.\mathit{stack}\}$.
\end{lemma}

\begin{theorem}[Overlay Integrity Soundness] \label{thm:appendix:overlay-integrity-soundness}
  If $\oc{\Phi_0} \in \programs$, $\oc{\Phi_0} \ostepn{n} \oc{\Phi_1}$,
  $\currentcom{\oc{\Phi_1}}{\_}{\trusted}$, and
  $\oc{\Phi_1} \ostepstar \oc{\Phi_2}$ such that $\oc{\Phi_1}.\Psi \stepwb
  \oc{\Phi_2}.\Psi$ with $\pi = \oc{\Phi_0}.\Psi \stepn{n} \oc{\Phi_1}.\Psi$, then
  \begin{enumerate}
  \item $\mathcal{CSR}(\pi, \oc{\Phi_1}.\Psi, \oc{\Phi_2}.\Psi)$
  \item $\mathcal{RA}(\pi, \oc{\Phi_1}.\Psi, \oc{\Phi_2}.\Psi)$
  \end{enumerate}
\end{theorem}
\begin{proof}
  By induction over the definition of a well-bracketed step and nested induction
  over the logical call stack.
  The last step follows by the fact that $\oc{\Phi_2} \neq \oerror$, and
  therefore the restoration checks in the overlay monitor passed.
\end{proof}

\begin{theorem}[Overlay Confidentiality Soundness] \label{thm:appendix:overlay-confidentiality-soundness}
  If $\oc{\Phi_0} \in \programs$, $\currentcom{\oc{\Phi_1}}{\_}{\untrusted}$,
  $\currentcom{\oc{\Phi_3}}{\_}{\trusted}$, $\oc{\Phi_0}.\Psi \stepstar
  \oc{\Phi_1}.\Psi \steplown{n} \oc{\Phi_2}.\Psi \step \oc{\Phi_3}.\Psi$,
  $\oc{\Phi_1} \ostepn{n + 1} \oc{\Phi_3}$, and $\oc{\Phi_1} =_{\untrusted}
  \oc{\Phi_1'}$,
  then $\oc{\Phi_1'}.\Psi \steplown{n} \oc{\Phi_2'}.\Psi \step \oc{\Phi_3'}.\Psi$,
  $\oc{\Phi_1'} \ostepn{n + 1} \oc{\Phi_3'}$
  $\currentcom{\oc{\Phi_3'}}{\_}{\trusted}$, $\oc{\Phi_3}.pc = \oc{\Phi_3'}.pc$,
  and
  \begin{enumerate}
  \item $\currentop{\oc{\Phi_2}}{\cgatecall{n'}{e}}$, $\currentop{\oc{\Phi_2'}}{\cgatecall{n'}{e}}$, and $\oc{\Phi_3} =_{\mathtt{call}\ n'} \oc{\Phi_3'}$ or
  \item $\currentop{\oc{\Phi_2}}{\cgateret}$, $\currentop{\oc{\Phi_2'}}{\cgateret}$, and $\oc{\Phi_3} =_{\mathtt{ret}} \oc{\Phi_3'}$,
  \end{enumerate}
\end{theorem}
\begin{proof}
  Proof is standard for an IFC enforcement system.
\end{proof}
\section{WebAssembly}
\label{appendix:webassembly}

\begin{mathpar}
  \inferrule
  {
    \currentop{\Psi}{\ccall{}{i}} \step \Psi'
  }
  {\currentcom{\Psi}{\cgatecall{n}{i}}{\trusted} \step \Psi'[p \assign \untrusted]}

  \inferrule
  {
    \currentop{\Psi}{\ccall{I}{i}} \step \Psi'
  }
  {\currentcom{\Psi}{\cgatecall{n}{i}}{\untrusted} \step \Psi'[p \assign \trusted]}

  \inferrule
  {
    \currentop{\Psi}{\cret{}} \step \Psi'
    \\\\ p' = \untrusted \Leftrightarrow p = \trusted
  }
  {\currentcom{\Psi}{\cgateret}{p} \step \Psi'[p \assign \trusted]}
\end{mathpar}
\captionof{figure}{WebAssembly Trampoline and Springboard}
\label{fig:appendix:wasm:operational}

\begin{figure}[t]
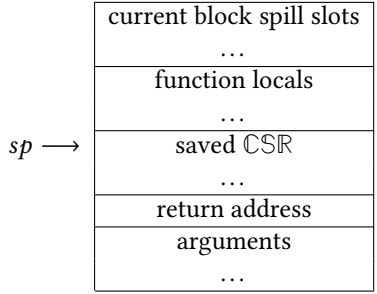

  \begin{center}
    \begin{tabular}{>{$} r <{$} | c |}
      \cline{2-2}
      & current block spill slots \\
      & \ldots \\
      \cline{2-2}
      & function locals \\
      & \ldots \\
      \cline{2-2}
      sp \longrightarrow & saved $\mathbb{CSR}$ \\
      & \ldots \\
      \cline{2-2}
      & return address \\
      \cline{2-2}
      & arguments \\
      & \ldots \\
      \cline{2-2}
    \end{tabular}
  \end{center}
  \caption{WebAssembly Stack Frame}
  \label{fig:appendix:wasm:stack-frame}
\end{figure}

\begin{center}
  \begin{tabular}{>{$}r<{$} >{$}r<{$} >{$}c<{$} >{$}l<{$}}
    \textit{WebAssemblyFunction} & \textit{WF} & \bnfdef &
    \begin{array}[t]{lllll}
      \{
      & |args| & \bnftypes & \nats \\
      & |locals| & \bnftypes & \nats \\
      & entry & \bnftypes & \textit{Block} \\
      & exits & \bnftypes & \powerset{\textit{Block}} \\
      & blocks & \bnftypes & \powerset{\textit{Block}} & \}
    \end{array} \\

    \textit{Block} & B & \bnfdef &
    \begin{array}[t]{lllll}
      \{
      & start & \bnftypes & \nats \\
      & end & \bnftypes & \nats \\
      & |slots| & \bnftypes & \nats \\
      & inputs & \bnftypes & \powerset{\nats + \regs} \\
      & indirects & \bnftypes & \powerset{\textit{IB}} & \}
    \end{array} \\

    \textit{Indirect Block} & \textit{IB} & \bnfdef &
    \begin{array}[t]{lllll}
      \{
      & start & \bnftypes & \nats \\
      & end & \bnftypes & \nats \\
      & parent & \bnftypes & \textit{Block} & \}
    \end{array}
  \end{tabular}

  \captionof{figure}{WebAssembly Structure}
  \label{fig:appendix:wasm:structure}
\end{center}

\subsection{Logical Relation}
\label{appendix:lr}

\begin{center}
  \begin{tabular}{>{$}r<{$} >{$}c<{$} >{$}l<{$}}
    L & \bnfdef &
      \begin{array}[t]{lllll}
        \{
        & \mathit{interface} & \bnftypes & \nats \rightharpoonup \nats \\
        & \mathit{library} & \bnftypes & \nats \rightharpoonup \mathit{WasmFunction} \\
        & \mathit{code} & \bnftypes & \codes & \}
      \end{array}
  \end{tabular}
\end{center}

\[\begin{array}{rcl}
  \mathrm{World} & \triangleq & \nats \times (\nats \rightharpoonup \nats) \times (\nats \rightharpoonup \nats) \\
  \\
  \later & : & \mathrm{World} \rightarrow \mathrm{World} \\
  \later (0, \overline{f_i}, \overline{f_l}) & \triangleq & (0, \overline{f_i}, \overline{f_l}) \\
  \later (n, \overline{f_i}, \overline{f_l}) & \triangleq & (n - 1, \overline{f_i}, \overline{f_l}) \\
  \\
  (C, \overline{\oc{F_2}}, \overline{\oc{F_3}}) :_W & \triangleq & \text{for } i \in \{1, 2\}: \\
  & &
  \begin{array}[t]{l}
    \overline{\oc{F_i}.\mathit{entry}} = \dom{\pi_i(W)} \\
    \forall \oc{F_i} \in \overline{\oc{F_i}}.~ \oc{F_i}.\mathit{type} = \pi_i(W)(\oc{F_i}.\mathit{entry}) \\
    \forall \oc{F_i} \in \overline{\oc{F_i}}.~ (\later W, \oc{F_i}, C|_{\oc{F_i}.\mathit{instrs}}) \in \Frel
  \end{array}
\end{array}\]

\begin{center}
  \[\begin{array}{rcl}
    \Frel & \triangleq &
    \left\{
      (W, \oc{F}, c)
      \middle|
      \begin{array}{l}
        \forall \rho \in [\frames], \mathit{ret\mbox{-}addr}, A \in [\nats], sp, R, M, C, \overline{\oc{F_i}}, \overline{\oc{F_l}}. \\
        \quad \text{let } \rho' = \rho \concat \{\mathit{base} \assign sp - |A|, \mathit{ret\mbox{-}addr\mbox{-}loc} \assign sp, \mathit{csr\mbox{-}vals} \assign R(\mathbb{CSR})\} \\
        \quad |A| = \oc{F}.\mathit{type} \\
        \quad sp > \oname{top}(\rho).\mathit{ret\mbox{-}addr\mbox{-}loc} + |A| \\
        \quad [sp \mapsto \mathit{ret\mbox{-}addr}, (sp - |A| + i \mapsto A_i)_{i \in [0, |A|)}] \leq M  \\
        \quad (C, \overline{\oc{F_i}}, \overline{\oc{F_l}}) :_W \\
        \quad c \leq C \\
        \quad \dom{c} = \oc{F}.\mathit{instrs} \\
        \quad \forall \ell \in \dom{M} \cap M_{\untrusted}.~ \val{n}{\untrusted} = M(\ell) \\
        \quad \text{let } \Psi = \{pc \assign \oc{F}.\mathit{entry}, sp \assign sp, R \assign R, M \assign M, C \assign C \} \\
        \quad \text{let } \oPhi = \{\Psi \assign \Psi, \mathit{stack} \assign \rho', \mathit{funcs} \assign [\oc{F'}.\mathit{entry} \mapsto \oc{F'}]_{\oc{F'} \in \overline{\oc{F_i}} \uplus \overline{\oc{F_l}}}\} \\

        \Longrightarrow \\
        \quad \forall n' \leq W.n.~ \oPhi \osteprhon{\rho'}{n'} \oc{\Phi'} \Longrightarrow \oc{\Phi'} \neq \oerror
        \\
        \quad \text{or } \exists n' \leq W.n.~ \oPhi \osteprhon{\rho'}{n' - 1} \oc{\Phi'} \ostep \oc{\Phi''}
        \\
        \qquad \text{where } (\currentop{\oc{\Phi'}}{\cret{}} \vee \currentop{\oc{\Phi'}}{\cgateret}) \wedge \oc{\Phi'}.\mathit{stack} = \rho' \wedge \oc{\Phi''} \neq \oerror
      \end{array}
    \right\} \\
    \\
    \Lrel & \triangleq &
    \left\{
      (n, L)
      \middle|~
      \begin{array}{l}
        \forall i \in \dom{L.\mathit{library}}. \\
        \quad \text{let } \textit{WF} = L.\mathit{library}(i) \\
        \quad \text{let } W = (n, L.\mathit{interface}, \lambda i \rightarrow L.\mathit{library}(i).|args|) \\
        \quad \text{let } \mathit{instrs} = \biguplus_{B \in \textit{WF}.\mathit{blocks}} [B.start, B.end] \\
        \quad \text{let } \oc{F} = \{ \mathit{instrs} \assign \mathit{instrs}, \mathit{entry} \assign \textit{WF}.\mathit{entry}.\mathit{start}, \mathit{type} \assign \textit{WF}.|args|\} \\
        \quad (W, \oc{F}, L.\mathit{code}|_{\mathit{instrs}}) \in \Frel
      \end{array}
    \right\}
  \end{array}\]
  \captionof{figure}{Function and Library Relations}
\end{center}

\begin{lemma}[FTLR for functions] \label{appendix:wasm:lr:ftlr-functions}
  Let $W \in \mathrm{World}$ and $c$ be the code for a compiled WebAssembly function
  $\textit{WF}$ such that $\textit{WF}$ expects application functions in the
  interface with locations and types $\pi_2(W)$ and in the library with locations
  and types $\pi_3(W)$.
  Further let $\mathit{instrs} = \biguplus_{B \in \textit{WF}.\mathit{blocks}}
  [B.start, B.end]$ and $\oc{F} = \{ \mathit{instrs} \assign \mathit{instrs},
  \mathit{entry} \assign \textit{WF}.\mathit{entry}.\mathit{start}, \mathit{type}
  \assign \textit{WF}.|args|\}$.
  Then $(W, \oc{F}, c) \in \Frel$.
\end{lemma}
\begin{proof}
  We first unroll the assumptions of $(W, \oc{F}, c) \in \Frel$ reusing the
  variable names defined there.
  We will maintain that any steps do not step to $\oerror$ so WOLOG we will
  continually assume $n' \leq W.n$ such that $n'$ greater than the number of steps
  we have taken, otherwise the case $\oc{\Phi} \osteprhon{\rho'}{n'} \oc{\Phi'}
  \Longrightarrow \oc{\Phi'} \neq \oerror$ holds.

  By the structure of a compiled WebAssembly function and assumption we have
  that the stack and stack pointer represent $\textit{WF}.|args|$ arguments.
  The abstract interpretation ensures that if we write to $H_{\untrusted}$ then
  that value has label $\untrusted$ so the checks pass.
  We are further assured that we do not read or below the stack frame.
  The structure of a compiled WebAssembly block then lets us proceed until we
  reach one of
  \begin{enumerate*}
  \item a function call to a library function $\textit{WF}'$ such that $\textit{WF}'.\mathit{entry}.\mathit{start} \in \overline{F_l.\mathit{entry}}$,
  \item an application function $\oc{F'}$ such that $\oc{F'}.\mathit{entry} \in \overline{F_i.\mathit{entry}}$,
  \item or the end of the block.
  \end{enumerate*}

  \begin{enumerate}
  \item

    The abstract interpretation ensures that we have initialized the arguments
    $\textit{WF}'.|args| = \pi_3(W)(\textit{WF}'.\mathit{entry}.\mathit{start})$ or
    failed a dynamic type check and terminated (thus stepping to a terminal state
    that is not an $\oerror$).
    We thus set $\rho_2 = \rho'$ and see that we have constructed $\rho_2' =
    \rho_2 \concat \{ \mathit{base} \assign sp - \textit{WF}'.|args|,
    \mathit{ret\mbox{-}addr\mbox{-}loc} \assign sp, \mathit{csr\mbox{-}vals} \assign
    R(\mathbb{CSR}) \}$.
    We further set $\mathit{ret\mbox{-}addr} = pc + 1$, $A = \textit{WF}'.|args|$,
    $sp = sp$, $R = R$, $M = M$, $C = C$, $\overline{\oc{F_i}} = \overline{\oc{F_i}}$,
    and $\overline{\oc{F_l}} = \overline{\oc{F_l}}$.
    By the abstract interpreation we have that all of the remaining checks in
    $\Frel$ pass and that the instantiated $\oc{\Phi}$ is equal to our current
    state.
    We therefore instantiate $(\later \oc{F_l}, C|_{\oc{F_l}.\mathit{instrs}}) \in \Frel$.
    If this uses the remaining steps then we are done.
    Otherwise we get that we return to $pc + 1$ with all values restored and no
    new $\trusted$ values written to the library memory, and our walk through the
    block may continue.

  \item

    Identical to the case for (1).

  \item

    The end of a block is followed by a direct jump to another block $B'$, an
    indirect block $\mathit{IB}$, or we are at an exit block.
    In the case of another block $B'$ we have by the structure of compiled
    WebAssembly code that we have instantiated $B'.inputs$.
    We thus jump to the block and follow the same proof structure as detailed here.
    The same is true of an intermediate block $\mathit{IB}$ except with the
    extra steps of setting up the inputs jumping to another block $B''$.
    Lastly if we have reached the end of an exit block then we have not touched
    the pushed return address or callee-save registers and the stack pointer is in
    the expected location.
    We thus execute $\cret{}$ or $\cgateret$ and pass the overlay monitor checks.
  \end{enumerate}
\end{proof}

\begin{lemma}[FTLR for libraries] \label{thm:appendix:wasm:lr:ftlr-libraries}
  For any number of steps $n \in \nats$ and compiled WebAssembly library $L$,
  $(n, L) \in \Lrel$.
\end{lemma}
\begin{proof}
  By unrolling the definition of $\Lrel$ and \lemref{appendix:wasm:lr:ftlr-functions}.
\end{proof}

\begin{theorem}[Adequacy of $\Lrel$] \label{thm:appendix:lr:adequacy}
  For any number of steps $n \in \nats$, library $L$ such that $(n, L) \in
  \Lrel$, program $\oc{\Phi_0} \in \programs$ using $L$, and $n' \leq n$, if
  $\oc{\Phi_0} \ostepn{n'} \oc{\Phi'}$ then $\oc{\Phi'} \neq \oerror$.
\end{theorem}
\begin{proof}
  Straightforward: by assumption for steps in the application, and by assumption about
  application code properly calling the library code and the unrolling of $\Lrel$ and
  $\Frel$.
\end{proof}

\section{Additional Image Benchmarks}
\label{appendix:img}

\begin{center}
  \includegraphics{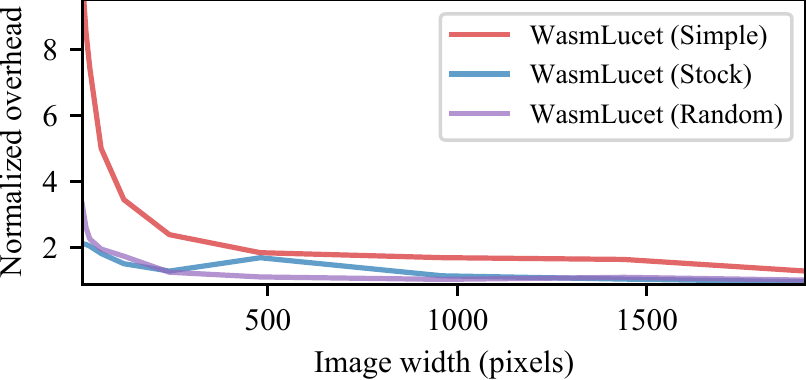}
  \captionof{figure}{
    Performance of the \trlucet heavyweight transitions included in the Lucet runtime
    on the image benchmarks in \sectionref{sec:eval}. Performance
    when rendering (a)~a simple image with one color, (b)~a stock image and
    (c)~a complex image with random pixels. The performance is the overhead
    compared to \trfast.
  }
  \label{fig:jpeg-img-lucet}
\end{center}

\begin{center}
	\includegraphics{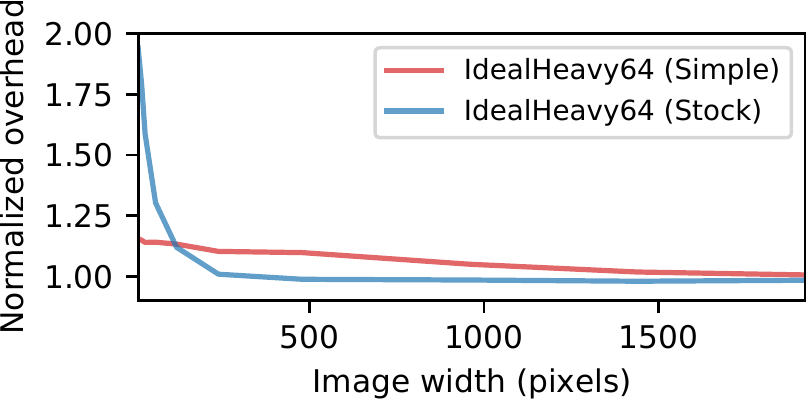}
	\captionof{figure}{
	Performance of an ideal isolation scheme (no enforcement overhead) with heavy trampolines when rendering images. Wasm compilers whose enforcement overhead is lower than this can outperform even an ideal isolation scheme that uses heavy weight transitions.
	}
	\label{fig:jpeg-img-ideal64}
\end{center}

}{}

\end{document}